\documentclass[runningheads,a4paper]{llncs}
\usepackage[T1]{fontenc}
\usepackage{graphicx}
\usepackage{mathtools}

\usepackage{amsthm}
\usepackage{dsfont}
\usepackage{amssymb}
\usepackage[table]{xcolor}
\usepackage{esvect}
\usepackage{subcaption}
\usepackage{tikz}
\usetikzlibrary{quantikz}
\usetikzlibrary{angles,quotes}
\usepackage[normalem]{ulem}
\definecolor{TSUYUKUSA}{RGB}{24, 58, 230}
\definecolor{KURENAI}{RGB}{203, 27, 69}
\usepackage[pdfstartview=FitH,pdfpagemode=UseNone,colorlinks=true,citecolor=KURENAI,linkcolor=TSUYUKUSA,backref=page,linktocpage=true]{hyperref}
\usepackage[capitalise]{cleveref}
\usepackage[nottoc]{tocbibind}
\usepackage{appendix}
\usepackage{array}

\newtheorem{fact}[theorem]{Fact}
\newcommand{\fig}[1]{\hyperref[fig:#1]{Figure~\ref*{fig:#1}}}
\newcommand{\eq}[1]{\hyperref[eq:#1]{(\ref*{eq:#1})}}
\newcommand{\lem}[1]{\hyperref[lem:#1]{Lemma~\ref*{lem:#1}}}
\newcommand{\thm}[1]{\hyperref[thm:#1]{Theorem~\ref*{thm:#1}}}
\newcommand{\defi}[1]{\hyperref[def:#1]{Definition~\ref*{def:#1}}}
\newcommand{\app}[1]{\hyperref[sec:#1]{Appendix~\ref*{sec:#1}}}
\newcommand{\fct}[1]{\hyperref[fact:#1]{Fact~\ref*{fact:#1}}}
\newcommand{\sect}[1]{\hyperref[sec:#1]{Section~\ref*{sec:#1}}}
\newcommand{\subsec}[1]{\hyperref[subsec:#1]{Subsection~\ref*{subsec:#1}}}
\newcommand{\itm}[2]{\hyperref[itm:#1]{#2}}
\newcommand{\clm}[1]{\hyperref[clm:#1]{Claim~\ref*{clm:#1}}}
\newcommand{\rmk}[1]{\hyperref[rmk:#1]{Remark~\ref*{rmk:#1}}}

\makeatletter
\newtheorem*{rep@theorem}{\rep@title}
\newcommand{\newreptheorem}[2]{%
\newenvironment{rep#1}[1]{%
 \def\rep@title{#2 \ref{##1}}%
 \begin{rep@theorem}}%
 {\end{rep@theorem}}}
\makeatother

\newreptheorem{theorem}{Theorem}

\definecolor{lightcyan}{RGB}{0.88,1,1}
\definecolor{darkgreen}{RGB}{0, 128, 0}
\definecolor{darkblue}{RGB}{0, 0, 128}

\newcommand{\hlight}[1]{\textcolor{darkgreen}{#1}}

\newcommand{\N}{\mathbb{N}}  %
\newcommand{\R}{\mathbb{R}} %
\newcommand{\C}{\mathbb{C}} %
\renewcommand{\d}{\mathrm{d}} %
\DeclareMathOperator*{\E}{\mathbb{E}}  %
\newcommand{\so}{\mathrm{SO}} %
\newcommand{\su}{\mathrm{SU}} %
\renewcommand{\i}{\mathrm{i}} %
\newcommand{\A}{\mathcal{A}}  %
\newcommand{\B}{\mathcal{B}} %
\newcommand{\CC}{\mathcal{C}} %
\newcommand{\D}{\mathcal{D}} %
\newcommand{\F}{\mathcal{F}} %
\renewcommand{\H}{\mathcal{H}} %
\newcommand{\K}{\mathcal{K}} %
\newcommand{\NN}{\mathcal{N}} %
\newcommand{\V}{\mathcal{V}} %
\newcommand{\X}{\mathcal{X}} %
\newcommand{\Y}{\mathcal{Y}} %
\renewcommand{\S}{\mathcal{S}} %
\newcommand{\SR}{\mathcal{S}_{\R}} %
\newcommand{\SC}{\mathcal{S}_{\C}} %
\newcommand{\EE}{\mathcal{E}}  %
\newcommand{\PP}{\mathcal{P}} %
\newcommand{\KK}{\widetilde{K}} %
\newcommand{\LL}{\widetilde{L}} %
\newcommand{\bit}[1]{\{0,1\}^{#1}} %
\newcommand{\set}[1]{\left\{#1\right\}} %
\newcommand{\expec}[1]{\E\!\Br{#1}} %
\newcommand{\expect}[2]{\E_{\substack{#1}}\!\Br{#2}} %
\newcommand{\prob}[2]{\underset{#1}{\mathrm{Pr}}\!\Br{#2}} %
\newcommand{\cf}{\widetilde{f}} %
\newcommand{\cg}{\widetilde{g}} %
\newcommand{\ch}{\widetilde{h}} %
\newcommand{\ck}{\widetilde{K}} %
\newcommand{\rep}[2]{(#1_i)_{i=1}^{#2}} %

\newcommand{\br}[1]{\left(#1\right)} %
\newcommand{\Br}[1]{\left[#1\right]} %
\newcommand{\st}[1]{\left\{#1\right\}} %
\newcommand{\tr}[1]{\mathrm{Tr}\!\Br{#1}} %
\newcommand{\abs}[1]{\left|#1 \right|} %
\newcommand{\norm}[1]{\left\lVert #1 \right\rVert} %
\newcommand{\agl}[2]{\theta^{\br{#1}}_{#2}} %
\newcommand{\aglp}[2]{{\theta'}^{\br{#1}}_{#2}} %
\newcommand{\lint}[1]{\left\lfloor#1\right\rfloor} %
\newcommand{\poly}[1]{\mathrm{poly}\!\br{#1}} %
\newcommand{\negl}[1]{\mathrm{negl}\!\br{#1}} %
\newcommand{\de}[1]{\mathrm{d}#1} %
\newcommand{\val}[1]{\mathrm{val}\!\br{#1}} %
\newcommand{\vall}[1]{\mathrm{val}\br{#1}} %
\newcommand{\nd}[1]{\mathcal{N}\!\br{#1}} %
\newcommand{\ketbratwo}[2]{\ket{#1} \hspace{-0.4em}\bra{#2}} %
\newcommand{\ketbra}[1]{\ketbratwo{#1}{#1}} %
\newcommand{\id}{\ensuremath{\mathds{1}}} %
\newcommand{\ogroup}[1]{\mathrm{O}\!\br{#1}} %
\newcommand{\ugroup}[1]{\mathrm{U}\!\br{#1}} %
\newcommand{\td}{\mathrm{TD}} %
\newcommand{\tv}[1]{\norm{#1}_{\mathrm{TV}}} %
\newcommand{\vdim}{\ensuremath{N}} %
\newcommand{\dimin}{\ensuremath{n}} %
\newcommand{\dimout}{\ensuremath{m}} %
\newcommand{\ncopy}{\ell} %
\newcommand{\hspacein}{\H_\mathrm{in}} %
\newcommand{\hspaceout}{\H_\mathrm{out}} %
\newcommand{\Sin}{\S(\hspacein)} %
\newcommand{\Sout}{\S(\hspaceout)} %
\newcommand{\haar}{\ensuremath{\mu}} %
\newcommand{\tensorsrss}{\ensuremath{\nu}} %
\newcommand{\qadvice}{\ensuremath{\rho}} %
\newcommand{\tp}{\otimes} %
\newcommand{\wone}[2]{W_1\!\br{#1,#2}} %

\newcommand{\secpar}{\lambda} %
\newcommand{\klen}{\kappa} %
\newcommand{\cprim}[1]{\textup{\textsf{#1}}} %
\newcommand{\prf}{\cprim{PRF}} %
\newcommand{\prp}{\cprim{PRP}} %
\newcommand{\pru}{\cprim{PRU}} %
\newcommand{\prs}{\cprim{PRS}} %
\newcommand{\prsg}{\cprim{PRSG}} %
\newcommand{\prss}{\cprim{PRSS}} %
\newcommand{\rss}{\cprim{RSS}} %
\newcommand{\srss}{\cprim{DRSS}} %
\newcommand{\csrss}{\cprim{CRSS}} %
\newcommand{\qbc}{\cprim{QBC}} %
\newcommand{\qsc}{\cprim{QSC}} %
\newcommand{\sgen}{\cprim{SG}^{n,\secpar}} %
\newcommand{\enssgen}{\cprim{SG}^{n}} %
\newcommand{\scram}{\ensuremath{\mathcal{R}}} %
\newcommand{\SG}[2]{\ensuremath{\scram_{#1}^{#2}}} %
\newcommand{\srsscons}{\cprim{RSG}^{n,\secpar}} %
\newcommand{\csrsscons}{\widetilde{\cprim{RSG}}^{n,\secpar}} %
\newcommand{\enssrsscons}{\cprim{RSG}^{n}} %
\newcommand{\enscsrsscons}{\widetilde{\cprim{RSG}}^{n}} %
\newcommand{\sgenc}{\cprim{SGC}^{n,\secpar}} %
\newcommand{\enssgenc}{\cprim{SGC}^n} %
\newcommand{\srssconsc}{\cprim{RSGC}^{n,\secpar}} %
\newcommand{\csrssconsc}{\widetilde{\cprim{RSGC}}^{n,\secpar}} %
\newcommand{\enssrssconsc}{\cprim{RSGC}^{n}} %
\newcommand{\enscsrssconsc}{\widetilde{\cprim{RSGC}}^{n}} %
\newcommand{\qotp}{\cprim{QOTP}} %
\newcommand{\prg}{\cprim{PRG}} %
\newcommand{\qprf}{\cprim{QPRF}} %
\newcommand{\qprp}{\cprim{QPRP}} %
\newcommand{\prfs}{\cprim{PRFS}} %
\newcommand{\prfsg}{\cprim{PRFSG}} %
\newcommand{\owf}{\cprim{OWF}} %
\newcommand{\bqp}{\cprim{BQP}} %
\newcommand{\qma}{\cprim{QMA}} %
\newcommand{\oracle}{\mathcal{O}} %

\usepackage[normalem]{ulem}

\newcommand{\Haar}{\mathsf{Haar}} %
\newcommand{\regx}{\mathsf{X}} %
\newcommand{\regy}{\mathsf{Y}} %

\begin{document}

\title{Quantum Pseudorandom Scramblers}
\titlerunning{Quantum Pseudorandom Scramblers}
      
\author{Chuhan Lu\inst{1} \orcidID{0009-0000-3359-320X}\and
Minglong Qin\inst{2} \orcidID{0009-0004-8760-5498}\and\\
Fang Song\inst{1} \orcidID{0000-0002-3098-6451}\and
Penghui Yao\inst{2,3} \orcidID{0000-0002-4104-2069}\and
Mingnan Zhao\inst{2} \orcidID{0009-0006-9623-9703}}
\authorrunning{C. Lu et al.}

\index{Lu, Chuhan}
\index{Qin, Minglong}
\index{Song, Fang}
\index{Yao, Penghui}
\index{Zhao, Mingnan}

\institute{Computer Science Department, Portland State University, USA\\
\email{\{chuhan,fang.song\}@pdx.edu} \and
State Key Laboratory for Novel Software Technology, Nanjing University,\\Nanjing 210023, China\\
\email{mlqin@smail.nju.edu.cn}\quad \email{\{phyao1985,mingnanzh\}@gmail.com} \and
Hefei National Laboratory, Hefei 230088, China}

\maketitle

\begin{abstract}
Quantum pseudorandom state generators (\prsg s) have stimulated exciting
developments in recent years. A \prsg, on a fixed initial (e.g.,
all-zero) state, produces an output state that is computationally
indistinguishable from a Haar random state. However, pseudorandomness
of the output state is not guaranteed on other initial states. In fact,
known \prsg\ constructions \emph{provably} fail on some initial states.

In this work, we propose and construct quantum Pseudorandom State
Scramblers (\prss s), which can produce a pseudorandom state on an
\emph{arbitrary} initial state. In the information-theoretical setting,
we obtain a scrambler which maps an arbitrary initial state to a
distribution of quantum states that is close to Haar random in
\emph{total variation distance}. As a result, our scrambler exhibits a {\em dispersing} property. Loosely, it can span an $\epsilon$-net of the state space. This significantly strengthens what
standard \prsg s can induce, as they may only concentrate on a small region of the state space provided that the average output state approximates a Haar random state.

Our \prss\ construction develops a \emph{parallel} extension of the
famous Kac's walk, and we show that it mixes \emph{exponentially}
faster than the standard Kac's walk. This constitutes the core of our
proof. We also describe a few applications of \prss s. While our \prss\ 
construction assumes a post-quantum one-way function, \prss s are
potentially a weaker primitive and can be separated from one-way
functions in a relativized world similar to standard \prsg s.

\keywords{Quantum pseudorandom states \and Kac's walk \and Pseudorandom unitary operators}
\end{abstract}
\section{Introduction}
\label{sec:intro}
Pseudorandomness is a fundamental concept in complexity theory and
cryptography, offering efficient approximation to true randomness
against computationally bounded adversaries. Recently, Ji, Liu and
Song~\cite{JLS18} introduced quantum pseudorandom state generators
($\prsg$s) as a family of quantum states $\{\ket{\phi_k}\}_{k\in\K}$,
which can be generated in polynomial time, and no
computationally-bounded quantum adversary can distinguish polynomially
many copies of $\ket{\phi_k}$ from polynomially many copies of a Haar
random state. $\prsg$s can be considered as a quantum counterpart to
classical pseudorandom generators, and can be constructed assuming the
existence of one-way functions that are hard for efficient quantum
adversaries~\cite{JLS18,BS19,BS20,AGQY22,ABF+22}. What is surprising,
$\prsg$s are proven weaker than one-way functions in a relativized
world~\cite{Kre21,KQST23}. Since one-way functions are considered the
\emph{minimal} assumption in classical cryptography, this opens up the
possibility of basing quantum cryptography on \emph{weaker}
assumptions. There have been exciting advances in recent years,
realizing a host of cryptographic tasks based on
$\prsg$s~\cite{BS20_qcoin,AQY22,MY22,AGQY22,Col23}. In addition to
cryptographic interest, pseudorandom states have also inspired new
developments for quantum gravity theory and string
theory~\cite{BFV20,KTP20,BCHJ+21,ABF+22,YE23}.

Another fundamental quantum pseudorandom primitive, \emph{pseudorandom
  unitary operators} (\pru s), was also introduced in~\cite{JLS18} as a
quantum analogue of pseudorandom functions. A \pru~is a set of
polynomially-time unitary operators that are computationally
indistinguishable from Haar random unitaries. $\pru$s clearly imply
$\prsg$s and could further enrich the toolkit in cryptography and
physics~\cite{BFV20,KTP20,BCHJ+21,YE23,GJMZ23}. Nonetheless,
constructing a provably-secure $\pru$ remains an open problem, and
progress has been slow (e.g., conjectured constructions
in~\cite{JLS18}, a stateful simulation in~\cite{AMR20}, and on the
negative side some barriers such as impossibility of $\pru$s that are
sparse or of real entries~\cite{HBEK23}). In fact, even basic
properties that are \emph{necessary} for $\pru$s have not been
achieved. It is easy to see that a $\pru$ gives a family of
polynomial-sized quantum circuits which can map an \emph{arbitrary}
pure state to a family of pseudorandom states. However, a $\prsg$ can
be viewed as a family of polynomial-sized quantum circuits which map a
specific initial state, typically $\ket{0^n}$, to a family of
pseudorandom states. Indeed, all existing constructions of $\prsg$s
necessitate a specific initial state, and it can be shown that they
\emph{fail} to produce pseudorandom states for certain initial
states. This limitation has indeed caused a variety of technical
challenges in the cryptographic applications mentioned before that
need to be addressed in ad hoc ways. It hence becomes imperative to
understand the following question and its consequences.
\begin{quote} {\it Can we construct a family of polynomial-sized
    quantum circuits which map an arbitrary input (pure) state to
    pseudorandom states?}
\end{quote}

\subsection{Our Contributions}
\label{sec:contribution}

In this work, we answer the question affirmatively as a steady step
towards bridging the gap between $\prsg$s and $\pru$s. We formally
encapsulate the property of ``scrambling'' an arbitrary input state in
a novel quantum pseudorandom primitive, termed a {\em quantum
  pseudorandom state scrambler} (\prss), which isometrically maps an
\emph{arbitrary} pure state to a pseudorandom state. We then construct
a $\prss$ based on any quantum-secure $\prf$. A central technical novelty
is to design a \emph{parallel} version of Kac's walk, which is a random walk on a unit sphere,
and prove a mixing time \emph{exponentially} faster than the standard Kac's walk~\cite{PS17}. Although Kac's walk was introduced by Kac in~\cite{Kac56} more than half a century ago and has been studied by a large body of works since then, this work, to our knowledge, is the first time to employ Kac's walk to design quantum pseudorandom objects.  

Our construction also exhibits a notable \emph{dispersing}
property. Loosely speaking, the output states of our scrambler
constitute an $\epsilon$-net on the sphere, and the distribution
closely approximates the Haar random distribution under the strong
Wasserstein distance, when sufficient randomness is supplied. Such a
powerful ``randomizing'' capability needs not be present even in
$\pru$s.

\paragraph{Overview on the construction and analysis.} Our
construction is inspired by Kac's walk, originally a model for a
Boltzmann gas~\cite{Kac56}. This approach differs from previous
constructions for $\prsg$s. Let us consider an arbitrary unit-vector
$v\in \R^\vdim$. In one step of Kac's walk, two distinct coordinates
$(i,j)$ and an angle $\theta\in [0,2\pi)$ are chosen uniformly at
random.  Then $R_\theta:= \begin{pmatrix}
                            \cos\theta & -\sin\theta\\
                            \sin\theta & \cos\theta
                          \end{pmatrix}$ is
                          applied to rotate the two-dimensional subvector $(v_i, v_j)^T$.
                          It is
                          proven that it converges to the Haar measure on
                          the unit sphere of
                          $\R^\vdim$ in $O(N\log N)$ steps~\cite{PS17}.
                                      However,  if we view the
                                      input vector as an $\dimin = \log \vdim$-qubit state, then the factor $\vdim$ in
                                      the mixing time is prohibitive for the purpose of an efficient
                                      (polynomial in $\dimin$) scrambler.

Can we \emph{parallelize} Kac's walk in hope of shaking off a factor
of $\vdim$? Notice that in Kac's walk, if any two consecutive steps
overlap on the random choices of coordinates, then they need to be
executed in \emph{sequence}. One might consider conditioning on the
event of ``collision-free'' in the coordinate choices, but this occurs
with negligibly small probability since we intend to compress
$\Omega(\vdim)$ steps into one.

We design a \emph{parallel} Kac's walk that rapidly mixes in $O(\log N)$ time,
an \emph{exponential} improvement over the original walk.
In each step, instead of working with an individual pair of coordinates,
we randomly partition the $N$ coordinates into $N/2$ pairs,
and then each pair is rotated by a random angle chosen \emph{independently}.
Although the mixing time of Kac's walk is not
directly applicable, we show that the specific path-coupling proof
strategy of~\cite{PS17} can be extended here.

We then construct a quantum circuit to implement our parallel Kac's
walk. In each step, we use a random permutation to realize the
coordinate partition, and employ a random function to compute a random
rotation angle, under a careful discretization, for each pair of
coordinates. Finally, we obtain our pseudorandom state scramblers by
replacing the random permutations and functions with quantum-secure
pseudorandom permutations and functions, which exist based on
post-quantum one-way
functions~\cite{Zhandry21_qprf}. 

The discussion so far works with real Hilbert spaces. To construct a
$\prss$ in a complex Hilbert space, we further develop a parallel
Kac's walk on complex Hilbert spaces. The construction starts likewise
by randomly partitioning $\vdim$ coordinates to $\vdim/2$ pairs, and
then applying random $2\times 2$ \emph{unitary} matrices independently
to each pair. As unitary matrices have more degrees of freedom than
real orthogonal rotation matricies, the analysis of the mixing time is
more involved. The extension of Kac's walk to a complex Hilbert space,
as well as the parallelization, has not been studied previously as far
as we are aware. This may be of independent interest.

\paragraph{Applications.} It is easy to see that $\prss$s subsume
standard $\prsg$s as well as scalable $\prsg$s. We also demonstrate
that $\prss$s can be used to achieve a black-box realization of a
variant of $\prsg$s known as pseudorandom function-like state
generators ($\prfsg$s), which in turn enable a host of cryptographic
primitives such as IND-CPA SKE and EUF-CMA MAC~\cite{AQY22,AGQY22}. A $\prfsg$ takes an additional
classical input $x$ (from a poly-size domain) and produces a
pseudorandom state. In the literature, a $\prfsg$ (with logarithmic input length) can be constructed
from $\prsg$s by measuring a part of a pseudorandom state and
then post-selecting on $x$. This inevitably is error-prone and consumes
multiple copies, i.e., multiple invocations of a $\prsg$, to evaluate
on a single $x$. Given our $\prss$ (with a sufficiently long key), we
can simply feed $\ket{x}$ as the initial state to the $\prss$, and
hence only \emph{one}, rather than polynomially-many, run of $\prss$
suffices.

We observe that the argument by Kretschmer~\cite{Kre21} also implies
that $\prss$ is \emph{strictly} weaker than one-way functions relative
to an oracle. Thus $\prss$s may further enhance the new cryptographic
landscape without assuming one-way functions. We demonstrate some use
cases of $\prss$s beyond what are already possible from $\prsg$s. For
starters, a $\prss$ enables efficient encryption of quantum messages
by effectively ``scrambling'' any initial state, and allowing
\emph{multiple copies} of the same state to be encrypted under the
same key. 
The fact that $\prss$ provides a secure encryption also enables
committing quantum states, thanks to a new characterization
of~\cite{GJMZ23}. The commitment scheme can be further made
\emph{succinct}, where the commitment message has smaller size than
the size of the message to be committed. Existing constructions rely
on potentially stronger assumptions than $\prss$s.

\subsection*{Subsequent work}

A follow-up work~\cite{AGKL23} gave a construction that is
indistinguishable from applying the tensor product of a Haar random
isometry when the input state is restricted to one of three special
families: (1) $\ket{\psi}^{\otimes q}$ for a pure state $\ket{\psi}$
and polynomially-bounded $q$; (2) $\bigotimes_{i=1}^q\ket{x_i}$; and
(3) $\otimes_{i=1}^q\ket{\phi_i}$, where every $\phi_i$ is Haar
random. Their construction requires adding an ancilla system
$\ket{0^m}$, and the security loss scales with $1/{2^m}$. As a result,
it \emph{necessarily} cannot preserve the input dimension and $m$ is
chosen to be a polynomial to obtain negligible security loss. This
also incurs poly overheads in the applications such as quantum
encryption.  In other words, it only (unitarily) scrambles states
$\ket{\psi}\ket{0^m}$ for a $\ket{\psi}$ chosen from one of the three
families above and a polynomial $m$. 
More recently, several independent works on constructing pseudorandom unitaries which are secure against non-adaptive queries have been presented.
Metger, Poremba, Sinha and Yuen~\cite{metger2024simple} proposed a construction using a composition of Clifford gates, pseudorandom functions and pseudorandom permutations.
Brakerski and Magrafta~\cite{brakerski2024realvalued} presented a construction for real-valued unitaries that look like Haar random on any polynomial-sized set of orthogonal input states.
Chen, Bouland, Brandao, Docter, Hayden, and Xu~\cite{chen2024efficient} achieved similar results via products of exponentiated sums of random permutations with random phases.
It is not clear whether these constructions are able to generate an $\epsilon$-net and realize
the \emph{dispersing} property (details in \app{srss}), a strong
randomizing property achieved by our construction.

\subsection{Discussions and Open Questions}
\label{sec:discuss}

There is a rich history of studying Kac's walk in probability and
mathematical
physics\cite{Hastings70,DS00,Janvresse03,Oliveira09,Jiang12,HJ17}. Determining
the total variation mixing time of Kac's walk is particularly
challenging, and it is currently only known to be between the order
$O(n^4\log n)$ and
$O(n^2)$~\cite{PS18}. 

There has also been extensive efforts on approximations to Haar
measures in a \emph{statistical} setting, known as state and unitary
$t$-designs~\cite{RBR+04,DCEL09}. For instance, a unitary $t$-design
mimics a Haar random unitary up to the $t$-th moment. It is known that
a unitary $t$-design can be constructed by a quantum circuit of size
polynomially in $t$, composed of Haar random single or two-qubit
gates~\cite{HL09,BHH16,Haferkamp22,HM23,ODSP23}. It is interesting to
note that a path-coupling technique in~\cite{Oliveira09} for analyzing
Kac's walk also plays an essential role in the proofs of these unitary
design results. It is reasonable to anticipate improvements on the
efficiency of the unitary designs with new advances on Kac's
walk. However, it is worth stressing that another critical component
in their proofs involving spectral gaps appears to inevitably incur a
dependency on $t$, which is a serious limitation. For instance, in
order for the output state to approximate a Haar random state when the
number of copies can be an arbitrary polynomial, we would need to pick
a superpolynomial $t$ in the unitary design. As far as we know, our
$\prss$ is the first to employ Kac's walk directly in the construction
of a quantum pseudorandom object, and the exponential improvement on
the mixing time of our parallel walk enables flipping the quantifiers,
i.e., a fixed poly-size construction that is nonetheless pseudorandom
against any polynomial-time distinguisher, a desired feature towards
$\pru$s.

Kac's walk has also found applications in algorithm
design. Recently, a fast and memory-optimal dimension-reduction
algorithm is proposed based on Kac's walk and its discrete
variants~\cite{JPSSS22}. We would like to invite more exploration of
Kac's walk in theoretical computer science broadly.

We describe several interesting open problems emerged from our work.

\begin{enumerate}

\item Is it possible to simplify the quantum circuits for these
  primitives? Can we replace random permutations by a sequence of
  parallel (pseudo) random local permutations? Can we use the same
  random rotation or even a fixed one (e.g., Hadamard transform) in a
  single iteration? Recent advances on repeated averages on
  graphs~\cite{MSW22} and orthogonal repeated
  averaging~\cite{CDSZ22,JPSSS22} allude to an affirmative answer.

\item We believe that $\prss$s, potentially weaker than $\pru$s, are
  an important primitive in its own right.
  Can we discover more applications of $\prss$s and the dispersing
  property, especially in cryptography as well as in quantum gravity
  theory? For example, we envision a form of \emph{uncloneable
    knowledge tokens} from a $\prss$ that may enable novel quantum
  proof systems and \emph{delegated} computation.

\item Is our construction of $\prss$ capable of scrambling polynomial
  quantum states? This appears to require strengthening the coupling
  technique in our current analysis, and it might be useful to analyze
  other variants of Kac's walk.

\item How far are we
  from a $\pru$? Can we get it by strengthening our parallel Kac's
  walk approach or can we show that our construction is already a
  $\pru$? By a simple hybrid argument, it suffices to prove that our parallel Kac's walk
  on $\so(N)$ converges within $\mathrm{polylog}(N)$ time in terms of
  the $L^{\infty}$ Wasserstein distance. Indeed, there has been a
  large body of work devoted to studying the speed of the convergence
  with respect to different
  metrics~\cite{DS00,Kac56,Jiang12,PS18}. One of the most relevant
  works is Oliveira's result~\cite{Oliveira09} showing a tight
  convergence time of order $O(N^2\log N)$ with respect to the
  stronger $L^2$ Wasserstein distance. Our parallelization achieves a
  quadratic speedup, which leads to an $\tilde{O}(2^n)$-time
  construction of $\pru$. Since the $L^{\infty}$ Wasserstein distance
  is a less stringent metric than the $L^2$ Wasserstein distance,
  there is hope of obtaining an improved convergence rate. To our
  knowledge, the speed of convergence of Kac's walk with respect to
  $L^{\infty}$ Wasserstein distance has not been studied, and hence
  developing new techniques to overcome the tightness of Oliveira's
  $L^2$ result would be an exciting research direction.

\end{enumerate}

\paragraph{Acknowledgment.} We thank the anonymous reviewers'
feedback. CL and FS were supported in part by the US National Science
Foundation grants CCF-2042414, CCF-2054758 (CAREER) and
CCF-2224131. MQ, PY and MZ were supported in part by National Natural
Science Foundation of China (Grant No. 62332009, 61972191), and
Innovation Program for Quantum Science and Technology (Grant
No. 2021ZD0302901). FS thanks Zhengfeng Ji and Yi-Kai Liu for
discussions on the topic of quantum pseudorandomness.

\paragraph{Organization.} \sect{prelim} contains preliminary
materials on basic notations and cryptographic
primitives. \sect{def} describes definitions and properties of our
new primitives. \sect{kac} introduces the parallel Kac's walk.
Then \sect{const1} constructs $\prss$s via implementing the
parallel Kac's walk.
\sect{app} describes applications of $\prss$s.
In \app{srss}
we introduce the dispersing \rss.
In \app{connection}, we give details on the connections between $\prss$s
and existing $\prs$ variants.
Some proofs are deferred to
\app{deferred proofs}.

\section{Preliminary}
\label{sec:prelim}
\subsection{Basic Notation}

For $n\in\N$, $\Br{n}$ denotes $\st{1,\dots,n}$.  For $x\in\bit{n}$,
we use $x_i$ to denote the $i$-th bit of $x$ and define
$\val{x} = \sum_{i=1}^n 2^{-i}x_i$. Suppose that $x$ and $y$ are bit
strings of finite length,
we denote $xy$ to be the
concatenation of $x$ and $y$.  For finite sets $\mathcal{X}$ and
$\mathcal{Y}$, we use $\mathcal{X}^\mathcal{Y}$ to denote the set of
all functions $\{f:\mathcal{X}\to \mathcal{Y}\}$.  We use
$S_{\mathcal{X}}$ to denote the \emph{permutation group} over elements
in a finite set $\mathcal{X}$.  We often write $S_{2^n}$ instead of
$S_{\bit{n}}$ to denote the permutation group over elements in
$\st{0,1}^n$.

For any symbol $x$ and $n\in\N$, $(x_i)_{i=1}^n$
represents $(x_1,\ldots, x_n)$.
With a slight abuse of notation, we let $(x_i)_{i=1}^n\subseteq S$ represent $x_i\in S$ for
all $i\in[n]$.
For $n\in\N$, $\SR^{n}$ denotes the set of all unit
vectors in $\R^n$, $\SC^{n}$ denotes the set of all unit vectors in
$\C^n$, $\so(n)$ denotes the special orthogonal group of $n\times n$
real matrices, $\su(n)$ denotes the special unitary group of
$n\times n$ complex matrices, $\ogroup{n}$ denotes the $n\times n$
orthogonal group and $\ugroup{n}$ denotes the $n\times n$ unitary
group.  For a Hilbert space $\H$, we use $\S(\H)$ to denote the set of
pure quantum states in $\H$ and $\D(\H)$ to denote the set of density
operators on $\H$.

For an $n$-dimensional vector $v$ and $i\in[n]$, we use $v[i]$ to
denote the $i$-th coordinate of $v$.  For $S\subseteq [n]$ and
$v\in\C^n$, define
\begin{align*}
  \norm{v}_{1}=\sum_{i\in [n]} \abs{v[i]}\enspace,\enspace
  \norm{v}_{1,S}=\sum_{i\in S} \abs{v[i]}\enspace,\enspace
  \norm{v}_{2}=\sqrt{\sum_{i\in [n]} \abs{v[i]}^2}\enspace.\enspace
\end{align*}

For an $n\times n$ matrix $M$ and $p\in\N$, the $p$-norm of $M$ is
defined to be $\norm{M}_p = \br{\tr{\br{M^\dagger M}^{p/2}}}^{1/p}$, and $\norm{M}_\infty$ is defined to be the largest singular value of
$M$.  The following fact will be used in our paper and is easy to
prove by the triangle inequality.

\begin{fact}\label{fact:uprod}
Given $m,n\in\N$, $U_1,\dots,U_m,V_1,\dots,V_m\in\ogroup{n}$ (or $\ugroup{n}$), then
$$\norm{U_1\dots U_m-V_1\dots V_m}_\infty\leq\sum_{i=1}^m\norm{U_i-V_i}_\infty.$$
\end{fact}

Given two density operators $\rho,\sigma\in\D(\H)$, the trace distance between $\rho$ and $\sigma$ is
$\td\!\br{\rho,\sigma}=\norm{\rho - \sigma}_1.$

Let $\V$ be a real or complex vector space, and $\epsilon>0$ be a
positive real number. For any $\S\subseteq V$, a set of vectors
$\NN \subseteq \S$ is said to be an \emph{$\epsilon$-net} of $\S$ if,
for every vector $u\in\S$, there exists a vector $v\in\NN$ such that
$\norm{u-v}_2\leq \epsilon$.

We adopt the standard quantum circuit model. A quantum circuit with
gates drawn from a finite gate set can be encoded as a binary
string. $\{ Q_\secpar : \secpar \in \N \}$ is said to be a
\emph{polynomial-time} generated family\footnote{More precisely, each
  circuit should be written as $Q_{1^\secpar}$. Note that in a
  polynomial-time generated family, then $Q_\secpar$ must have size
  polynomial in $\secpar$.}  if there exists a deterministic Turing
machine that, on any input $\secpar \in \N $, outputs an encoding of
$Q_\secpar$ in polynomial-time in $\secpar$. A quantum polynomial-time
algorithm is identified with a polynomial-time generated circuit
family. In cryptography it is conventionally to model adversaries as
\emph{non-uniform} algorithms. We model a non-uniform quantum
polynomial-time algorithm as a family
$\{Q_\secpar,\qadvice_\secpar\}_\secpar$, where $\{Q_\secpar\}$ is a
polynomial-time generated circuit family, and $\{\qadvice_{\secpar}\}$
is a collection of \emph{advice states}. $Q_\secpar$ acts on
$\qadvice_\secpar$ besides the actual input state.

\subsection{Probability Theory}
For two probability measures $\nu_1$ and $\nu_2$ defined on measurable
space $\br{\Omega,\F}$, the \emph{total variation distance} of $\nu_1$
and $\nu_2$ is defined as $$\norm{\nu_1-\nu_2}_{\mathrm{TV}} = \sup_{A\in\F}
\abs{\nu_1\!\br{A}-\nu_2\!\br{A}}\enspace.$$ Closeness in total variation distance is a strong promise. For
example, when applied to quantum states, it implies closeness in trace
distance of the average states.

\begin{lemma} \label{lem:h3h4} Let $\mu$ and $\nu$ be two arbitrary
  probability measures over $\SR^{2^n}$ ($\SC^{2^n}$).  Then for all
  $\ncopy \in\N$,

  \[
    \norm{ \E_{\ket{\psi}\sim \mu}\!\Br{ \br{\ketbra{\psi}}^{\otimes
          \ncopy} } - \E_{\ket{\varphi}\sim
        \nu}\!\Br{\br{\ketbra{\varphi}}^{\otimes \ncopy}} }_1
    \leq\norm{\mu-\nu}_{\mathrm{TV}}\enspace. \]
\end{lemma}

We denote the distribution of a random variable $X$ by
$\mathcal{L}\!\br{X}$. If $\mathcal{L}\!\br{X}=\nu$, we write
$X\sim\nu$. A coupling of two probability measures $\mu$ and $\nu$ is
a joint probability measure whose marginals are $\mu$ and $\nu$.  We
use $\Gamma(\mu,\nu)$ to denote the set of all couplings of $\mu$ and
$\nu$. For $p\geq 1$ The \emph{Wasserstein $p$-distance} between two
probability measures $\mu$ and $\nu$ is
$$
W_p(\mu,\nu) = \br{ \inf_{\gamma\in\Gamma(\mu,\nu)} \E_{(x,y)\sim \gamma}\!\Br{\norm{x-y}_2^p} }^{1/p} \enspace.
$$
The Wasserstein $\infty$-distance is $W_\infty(\mu,\nu)=\lim_{p\rightarrow\infty}W_p(\mu,\nu)$.

The following lemmas about Markov chain in \cite{LPW09}
serve as crucial tools in this work.
\begin{lemma}[Coupling Lemma]\label{lem:coupling_lemma}~\cite[Theorem 5.4]{LPW09}
	Let $K$ be the transition kernel of a Markov chain with unique stationary distribution $\nu$ on state space $\Omega$.
	Let $\st{X_t}_{t\geq 0}, \st{Y_t}_{t\geq 0}$ be two corresponding Markov chains started at $X_0 = x \in \Omega$ and $Y_0\sim\nu$.
	Define the coalescence time of the chains
	$$
	\tau(x)=\min\st{t:X_t=Y_t}\enspace.
	$$
	Assume that a coupling of $\st{X_t}_{t\geq 0}, \st{Y_t}_{t\geq 0}$ satisfies $X_t = Y_t$ for all $t \geq \tau(x)$.
	Then for any $t\geq 0$,
	$$
	\norm{\mathcal{L}(X_t)-\nu}_{\mathrm{TV}}\leq \Pr\!\Br{\tau(x)>t}\enspace.
	$$
\end{lemma}

\begin{lemma}\label{lem:fcts_tv}~\cite[Lemma 4.10, Lemma 4.11 and Equation 4.29]{LPW09}
  Let $\st{X_t}_{t\geq 0}$ be a Markov chain with unique stationary
  distribution $\nu$ on state space $\Omega$.
	Then for integers $t_1,t_2\in\N$, we have
	\begin{itemize}
		\item if $t_2\geq t_1$, then
			$$\sup_{X_0\in \Omega} \norm{\mathcal{L}\!\br{X_{t_2}}-\nu}_{\mathrm{TV}}\leq 2\sup_{X_0\in \Omega} \norm{\mathcal{L}\!\br{X_{t_1}}-\nu}_{\mathrm{TV}}\enspace.$$
		\item if $t_2=s\cdot t_1$ for some integer $s\in\N$, then
			$$\sup_{X_0\in \Omega} \norm{\mathcal{L}\!\br{X_{t_2}}-\nu}_{\mathrm{TV}}\leq \br{2\cdot \sup_{X_0\in \Omega} \norm{\mathcal{L}\!\br{X_{t_1}}-\nu}_{\mathrm{TV}}}^s\enspace.$$
	\end{itemize}
\end{lemma}

We also utilize the following lemmas, which present upper bounds on the probability of a coordinate in a Haar random vector being small: 
\begin{lemma}[Lemma 3.5 in \cite{PS17}]\label{lem:marg_of_haar}
	Let $Y\sim \mu$ where $\mu$ is the Haar measure on $\SR^n$. Then for all $1<c<\infty$ and any $1\leq i\leq n$,
	$$
	\Pr\!\Br{Y[i]^2\leq n^{-3c}}\leq 2n^{1-c}\enspace.
	$$
\end{lemma}

\begin{lemma}\label{lem:marg_of_haar_c}
	Let $Y\sim \mu_{\C}$ where $\mu_{\C}$ is the Haar measure on $\SC^n$. Then for all $1<c<\infty$ and any $1\leq i\leq n$,
	$$
	\Pr\!\Br{\abs{Y[i]}^2\leq \br{2n}^{-3c}}\leq 2\cdot \br{2n}^{1-c} \enspace.
	$$
\end{lemma}

\begin{proof}
	Let $g_1,\dots,g_{2n}$ be $2n$ i.i.d. real random variable with $\nd{0,1}$ distribution. We have
	\begin{align*}
		\Pr\!\Br{\abs{Y[i]}^2\leq \br{2n}^{-3c}} &= \Pr\!\Br{\frac{g_1^2+g_2^2}{\sum_{k=1}^{2n} g_k^2}\leq \br{2n}^{-3c}} \\
		&\leq \Pr\!\Br{\frac{g_1^2}{\sum_{k=1}^{2n} g_k^2}\leq \br{2n}^{-3c}}\leq 2\cdot \br{2n}^{1-c} \enspace.
	\end{align*}
	The last inequality follows from \lem{marg_of_haar}.
\end{proof}

\subsection{Cryptography}
In this section, we will review various definitions and results in
cryptography. Throughout this work, $\secpar$ denotes a security
parameter.

\subsubsection{Pseudorandom Functions and Pseudorandom Permutations}

\begin{definition}[Quantum-Secure Pseudorandom Function]
  Let $\K,\X$ and $\Y$ be the key space, the domain and range, all
  implicitly depending on the security parameter $\lambda$. A keyed
  family of functions $\set{\prf_k:\X\to\Y}_{k\in\K}$ is a
  quantum-secure pseudorandom function (\qprf) if the following two
  conditions hold:
\begin{enumerate}
\item \textbf{Efficient generation}. $\prf_k$ is polynomial-time computable on a classical computer.

\item \textbf{Pseudorandomness}. For any polynomial-time quantum oracle algorithm $\A$, $\prf_k$ with a random $k\gets\K$ is indistinguishable from a truly random function $f\gets\Y^\X$ in the sense that:
\[\abs{\prob{k\gets\K}{\A^{\prf_k}\!\br{1^\lambda}=1}-\prob{f\gets\Y^\X}{\A^{f}\!\br{1^\lambda}=1}}=\negl{\lambda} \enspace .\]
\end{enumerate}
\end{definition}

\begin{definition}[Quantum-Secure Pseudorandom Permutation]
  Let $\K$ be the key space, and $\X$ be both the domain and range,
  implicitly depending on the security parameter $\lambda$. A keyed
  family of permutations $\set{\prp_k\in S_\X}_{k\in\K}$ is a
  quantum-secure pseudorandom permutation (\qprp) if the following two
  conditions hold:
\begin{enumerate}
\item (\textbf{Efficient generation}). $\prp_k$ and $\prp_k^{-1}$ are polynomial-time
  computable on a classical computer.
\item (\textbf{Pseudorandomness}). For any polynomial-time quantum
  oracle algorithm $\A$, $\prp_k$ with a random $k\gets\K$ is
  indistinguishable from a truly random permutation $\sigma\gets S_\X$
  in the sense that:
\[\abs{\prob{k\gets\K}{\A^{\prp_k,\prp_k^{-1}}\!\br{1^\lambda}=1}-\prob{\sigma\gets S_\X}{\A^{\sigma,\sigma^{-1}}\!\br{1^\lambda}=1}}=\negl{\lambda} \enspace .\]
\end{enumerate}
\end{definition}

We adopt the definition of a strong quantum-secure \prp~in this paper.
And when referring to a quantum oracle algorithm
having oracle access to a permutation $\sigma$,
we imply that it has oracle access to both $\sigma$ and its inverse $\sigma^{-1}$.

Under the assumption that post-quantum one-way functions exist,
Zhandry proved the existence of $\qprf$s~\cite{Zhandry21_qprf}. $\qprp$s can be
constructed from $\qprf$s efficiently \cite{Zhandry16_qprp}.

Given two $\qprf$s $F$ and $G$, one independently samples $F_{k_1}$
from $F$ and $G_{k_2}$ from $G$. A standard hybrid argument shows that
$F_{k_1}, G_{k_2}$ are computationally indistinguishable from two
independent random functions, as stated in the following lemma.
The proof, which is deferred to \app{deferred proofs},
can be readily extended to the scenario when polynomially many pseudorandom primitives (or random primitives) are given.

\begin{lemma}\label{lem:twokeys}
  Let keyed families of functions $F: \K_1 \times \X_1 \to \Y_1$ and
  $G: \K_2 \times \X_2 \to \Y_2$ be $\qprf$s. Then we have for any
  polynomial-time quantum oracle algorithm $\A$,
  \begin{align*}
    \abs{
    \prob{ k_1\gets\K_1, k_2\gets \K_2 }{\A^{F_{k_1}, G_{k_2}}\!\br{1^\lambda}=1}-\prob{f\gets\Y_1^{\X_1},g\gets\Y_2^{\X_2}}{\A^{f,g}\!\br{1^\lambda}=1}
    }=\negl{\lambda} \enspace .
  \end{align*}

  It also holds if $\X_2=\Y_2$, $G$ is a family of $\qprp$s and
  $g\gets\Y_2^{\X_2}$ is replaced by $g\gets S_{\X_2}$.
\end{lemma}

\subsubsection{Quantum Pseudorandomness}
The concept of quantum pseudorandom state generators was originally
introduced in~\cite{JLS18}.

\begin{definition}[Quantum Pseudorandom State Generator]
 Let $\K$ be a key space and $\H$ be a Hilbert space.  $\K$ and $\H$
 depend on the security parameter $\lambda$.  A pair of
 polynomial-time quantum algorithms $\br{K,G}$ is a pseudorandom
 state generator (\prsg) if the following holds:
\begin{itemize}
	\item \textbf{Key Generation.}
	$K(1^{\lambda})$ chooses a uniform $k\in\K$ and outputs it as the key.
	\item \textbf{State Generation.}
	For all $k\in\K$, $G(1^{\lambda},k)$
	outputs a quantum state $\ket{\phi_k}\in\S(\H)$.
	\item \textbf{Pseudorandomness.}
	Any polynomially many copies of $\ket{\phi_k}$ with the same random $k$
	is computationally indistinguishable from
	the same number of copies of a Haar random state.
	More precisely, for any $n\in\N$, any efficient quantum algorithm $\A$
	and any $\ncopy \in\poly{\lambda}$,
   \[\abs{\prob{k\gets \K}{\A\!\br{\ket{\phi_k}^{\otimes \ncopy}}=1}-\prob{\ket{\psi}\gets \mu}{\A\!\br{\ket{\psi}^{\otimes \ncopy}}=1}}=\negl{\lambda} \enspace ,\]
   where $\mu$ is the Haar measure on $\S(\H)$.
\end{itemize}
We call the keyed family of quantum states $\st{\phi_k}_{k\in\K}$ a
pseudorandom quantum state (\prs) in $\H$.
\end{definition}

$\prsg$s exist assuming the existence of $\qprf$s. Given any $\qprf$
$\prf:\K\times\bit{n} \to \bit{n}$ (where $\K$ and $N=2^n$ are
implicitly functions of the security parameter $\lambda$),
\cite{JLS18} constructed a $\prs$ $\st{\phi_k}_{k\in\K}$, referred to
\emph{(pseudo)random phase states}, as follows:
\[
  \ket{\phi_k} = \frac{1}{\sqrt{N}} \sum_{x\in\bit{n}}
  \omega_N^{\prf_k(x)}\ket{x}
\]
for $k\in \K$ and $\omega_N = e^{\i\frac{2\pi}{N}}$.  Additionally,
they conjectured the variant with binary phase (i.e., replacing
$\omega_N$ with -1) remains a $\prs$, and this was later confirmed
in~\cite{BS19}.

It is worth noting that both of these constructions rely on state
generation algorithms that require a specific initial state, typically
the all-zero state $\ket{0}^{\otimes n}$.  If we were to use a
different initial state, such as the equally weighted superposition
state $\ket{+}^{\otimes n}$, their state generation algorithms would
fail to produce a pseudorandom state.  Therefore, the specific initial
state is crucial for the success of these constructions.

\section{Pseudorandom State Scramblers}
\label{sec:def}
We describe our new primitive \emph{quantum Pseudorandom State
  Scramblers} (\prss). A $\prss$ is capable of generating a
pseudorandom state on an arbitrary initial state, addressing the
limitation of acting on one specific initial state.

\begin{definition}[Pseudorandom State Scrambler]
  \label{def:prssb}
  Let $\hspacein$ and $\hspaceout$ be Hilbert spaces of dimensions
  $2^{\dimin}$ and $2^{\dimout}$ respectively with
  $\dimin, \dimout \in\N$ and $\dimin \leq \dimout$. Let
  $\K=\bit{\klen}$ be a key space, and $\secpar$ be a security
  parameter. A \emph{pseudorandom state scrambler} (\prss) is an
  ensemble of isometric operators
  \[ \SG{}{\dimin,\dimout} : = \{\{\SG{k}{\dimin, \dimout,\secpar}:
    \hspacein \to \hspaceout \}_{k\in \K}\}_\secpar \, , \]
satisfying:
  \begin{itemize}
  \item \textbf{Pseudorandomness}. For any $\ncopy = \poly{\secpar}$,
    any $\ket{\phi}\in \Sin$, and any non-uniform poly-time quantum
    adversary $\A$,

    \[\abs{\prob{k\gets \K}{\A\!\br{\ket{\phi_k}^{\otimes \ncopy }}=1}
        -\prob{\ket{\psi}\gets \mu}{\A\!\br{\ket{\psi}^{\otimes
              \ncopy}}=1}} = \negl{\secpar}\, ,\]

    where $\ket{\phi_k} := \SG{k}{n,m,\secpar} \ket{\phi}$ and $\haar$
    is the Haar measure on $\Sout$.

  \item \textbf{Uniformity}. $\SG{}{\dimin,\dimout}$ can be uniformly
    computed in polynomial time. That is, there is a deterministic
    Turing machine that, on input
    $(1^{\dimin}, 1^{\dimout},1^\secpar,1^\klen)$, outputs a quantum
    circuit $Q$ in $\poly{\dimin,\dimout,\secpar,\klen}$ time such
    that for all $k\in \K$ and $\ket{\phi}\in\Sin$
    \[ Q\ket{k}\!\ket{\phi} = \ket{k}\!\ket{\phi_k}\, , \] where
    $\ket{\phi_k} := \SG{k}{n,m,\secpar}\! \ket{\phi}$.

  \item \textbf{Polynomially-bounded key length}.
    $\klen = \log |\K| = \poly{\dimout,\secpar}$. As a result,
    $\SG{}{\dimin,\dimout}$ can be computed efficiently in time
    $\poly{\dimin,\dimout,\secpar}$.
  \end{itemize}

\end{definition}

By strengthening the pseudorandomness condition in \prss, we define random state scramblers as follows.

\begin{definition}[Random State Scrambler]
  \label{def:rssb}
  Let $\hspacein$ and $\hspaceout$ be Hilbert spaces of dimensions
  $2^{\dimin}$ and $2^\dimout$ respectively with
  $\dimin, \dimout \in\N$ and $\dimin \leq \dimout$. Let
  $\K=\bit{\klen}$ be a key space, and $\secpar$ be a security
  parameter. A \emph{random state scrambler} (\rss) is an ensemble of
  isometric operators
  $\SG{}{\dimin,\dimout} := \{ \SG{}{\dimin,\dimout,\secpar}
  \}_\secpar$ with
  $\SG{}{\dimin,\dimout,\secpar} : = \{\SG{k}{\dimin,
    \dimout,\secpar}: \hspacein \to \hspaceout \}_{k\in \K}$ satisfying:

  \begin{itemize}

  \item \textbf{Statistical Pseudorandomness}. For any
    $\ncopy = \poly{\secpar}$, and any $\ket{\phi}\in \Sin$,

    \[ \td\br{\expect{k\gets \K}{\ketbra{\phi_k}^{\otimes \ncopy}},
        \expect{\ket{\psi}\in \haar}{ \ketbra{\psi}^{\otimes \ncopy}}}
      = \negl{\secpar} \, , \]

    where $\ket{\phi_k} := \SG{k}{n,m,\secpar} \ket{\phi}$ and $\haar$
    is the Haar measure on $\Sout$.

  \item \textbf{Uniformity}. $\SG{}{\dimin,\dimout}$ can be uniformly
    computed in polynomial time. That is, there is a deterministic
    Turing machine that, on input
    $(1^{\dimin}, 1^{\dimout},1^\secpar,1^\klen)$, outputs a quantum
    circuit $Q$ in $\poly{\dimin,\dimout,\secpar,\klen}$ time such
    that for all $k\in \K$ and $\ket{\phi}\in\Sin$
    \[ Q\ket{k}\!\ket{\phi} = \ket{k}\!\ket{\phi_k}\, , \] where
    $\ket{\phi_k} := \SG{k}{n,m,\secpar}\! \ket{\phi}$.
  \end{itemize}
\end{definition}

\begin{table}[t!]
  \centering
  \renewcommand{\arraystretch}{1.5}
  \newcolumntype{C}[1]{>{\centering\arraybackslash}m{#1}} 
  \begin{tabular}[h!]{| C{0.25\textwidth} |C{.25\textwidth} | C{0.4\textwidth} |}
    \hline
    \textbf{Random} & \multicolumn{1}{c|}{\textbf{Pseudorandom}} & \multicolumn{1}{c|}{\textbf{Main property}}
    \\ [0.5ex]
    \hline
    Haar unitary
           & PRU $\{U_k\}$ & $\{U_k\} \approx_c \text{Haar unitary}$ \\
    \hline
    $\rss$ & $\prss$ $\{\SG{k}{}\}$ & $\forall \ket{\phi}, \{\SG{k}{}\ket{\phi}\} \approx
                                      \text{Haar state}$ \newline
                                      (trace distance or comp. indist.)\\
    \hline
    Haar state & $\prsg$ $\{\SG{k}{}\}$ & for some \emph{fixed}
                                          $\ket{\phi}$ (e.g.,
                                          $\ket{0}$) \newline $\{\SG{k}{}\ket{\phi}\} \approx_c
                                          \text{Haar state}$  \\
    \hline
  \end{tabular}
  \vspace{0.5em}
  \caption{A collection of quantum random and pseudorandom
    objects.}
  \label{tab:prims}
\end{table}

\subsection{Properties of Pseudorandom State Scramblers}
\label{sec:property}

We discuss basic characteristics of the new primitives, as well as
their relationships with pseudorandom state generators and their
siblings.

\paragraph{Unitary to isometry.}
It is sufficient to construct \prss s from $\H$ to $\H$, since we can
construct \prss s from $\H_1$ to $\H_2$ ($n<m$) in the following way.
Let $\SG{}{m,m}:=\{\SG{}{m,m,\secpar}\}_\secpar$ be a \prss~with
$\SG{}{m,m,\secpar}:=\{\SG{k}{m,m,\secpar}: \H_2 \rightarrow \H_2\}$.
For all $\secpar\in \N$ and $k\in \K$, we define
$\SG{k}{n,m,\secpar} = \SG{k}{m,m,\secpar}\br{\id \tp \ket{0}^{\tp (m-n)}}$
where $\id$ is the identity of $\H_1$.
It is not hard to verify that
$\SG{}{n,m}$ is a \prss~from $\H_1$ to
$\H_2$.
We may write $\SG{}{m}$ instead of $\SG{}{m,m}$ when $m=n$.

\paragraph{Connections with Existing \prs~variants.} Several
definitions of quantum pseudorandomness on states with slight
variations have been proposed and constructed since the regular $\prs$
has been introduced. Brakerski and Shmueli~\cite{BS20} introduced
scalable pseudorandom states (scalable $\prs$s) to eliminate the
dependence between the state size and the security parameter. This
modification aids in assuring the security when the state size $n$ is
much smaller than the security parameter $\secpar$. Ananth, Qian and
Yuen \cite{AQY22} introduced pseudorandom function-like states ($\prfs$
s), which extend \prs s by augmenting with classical
inputs alongside the secret key. Although the security is initially
based on pre-selected classical queries to the $\prfs$ generator, the
subsequent work \cite{AGQY22} relaxes this to allow adversaries making
adaptive (classical or quantum) queries resulting in three levels of
security. The following theorem states that $\prss$s subsume the
original $\prs$s and those variants. The proof is deferred to
\app{connection}.
\begin{theorem}
  $\prsg$s, scalable $\prsg$s, and $\prfsg$s can be constructed via
  invoking $\prss$s in a black-box manner.
\end{theorem}

\paragraph{Oracle Separation from $\owf$s.}

According to~\cite[Theorem 2]{Kre21}, $\pru$s exist relative to a
quantum oracle $\oracle$, even when $\bqp^\oracle = \qma^\oracle$,
indicating the non-existence of one-way functions. Since $\pru$s imply
$\prss$s, we obtain the same oracle separation result for $\prss$s.

\begin{theorem}
  There exists a quantum oracle $\oracle$ relative to which \prss s
  exist, but $\bqp^\oracle = \qma^\oracle$.
\label{them:oraclesep}
\end{theorem}

\section{Parallel Kac's Walk}
\label{sec:kac}
In this section, we design a \emph{parallel version} of the standard Kac's
walk on $\SR^{n}$~\cite{PS17} and demonstrate that it mixes exponentially faster with 
respect to the metrics of our interest. 
We assume $n=2m$ for some $m\in\N$ throughout this section.

\subsection{Parallel Kac's Walk on Real Space}
Before introducing our parallel Kac's walk, we first review the
standard one. The standard Kac's walk on vectors within a real Hilbert
space is a Markov process. At each discrete time
$t$, we randomly select two coordinates $(i,j)$ of the vector, and
then apply a two-dimensional rotation to the corresponding subvector
with an angle $\theta$ drawn randomly and uniformly. After a
predetermined number of steps, the Markov chain converges to a Haar
distribution over the unit sphere. It is proved in \cite{PS17} that
the mixing time of Kac's walk on $\SR^{n}$ with respect to the total
variation distance is $\Theta\!\br{n\log n}$. The formal definition of
Kac's walk is given below.

\begin{definition}\label{def:okac}
Kac's walk on $\SR^{n}$ is a discrete-time
Markov chain $\st{X_t\in \SR^n}_{t\geq 0}$.  At each time $t$, two
coordinates $i^{(t)},j^{(t)}\in [n]$ and an angle
$\theta^{(t)}\in[0,2\pi)$ are chosen uniformly at random. $X_{t+1}$ is
obtained by the following update rules:
$$
\br{\begin{matrix}
	X_{t+1}[i^{(t)}]\\
	X_{t+1}[j^{(t)}]
\end{matrix}}
=
\Br{
\begin{matrix}
	\cos(\theta^{(t)}) & -\sin(\theta^{(t)})\\
	\sin(\theta^{(t)}) & \cos(\theta^{(t)})
\end{matrix}
}\!
\br{\begin{matrix}
	X_{t}[i^{(t)}]\\
	X_{t}[j^{(t)}]
\end{matrix}}\enspace,
$$
$$
X_{t+1}[k]=X_t[k]\quad\text{ for } \quad k\notin\st{i^{(t)},j^{(t)}}\enspace.
$$
We denote the Kac's walk as
$G:[n]\times[n]\times [0,2\pi)\times \SR^n\to \SR^n$
such that
\begin{equation}\label{eq:kaconestep}
X_{t+1}=G\!\br{i^{(t)},j^{(t)},\theta^{(t)}, X_t}\enspace.
\end{equation}	
\end{definition}

In our parallel Kac's walk,
instead of randomly rotating
one subvector,
we simultaneously rotate
$m$ subvectors.
Here we give its formal definition.

\begin{definition}
	The parallel Kac's walk
	is a discrete-time Markov chain $\st{X_t\in \SR^n}_{t\geq 0}$.
	At each step $t$,
	the parallel Kac's walk
	first selects a random perfect matching of the set $\st{1,\dots,n}$,
	denoted by
	$$P_t = \st{\br{i^{(t)}_1, j^{(t)}_1},\dots,\br{i^{(t)}_m, j^{(t)}_m}}\enspace,$$
	where
	$\bigcup_{k=1}^{m}\!\st{i^{(t)}_k, j^{(t)}_k}=\st{1,\dots, n}$.
	Then $m$ independent angles
	$\agl{t}{1},\dots,\agl{t}{m}\in [0,2\pi)$
		are chosen uniformly at random.
	For every pair $\br{i^{(t)}_k, j^{(t)}_k}$ in $P_t$, it sets
	$$
	\br{\begin{matrix}
		X_{t+1}[i^{(t)}_k]\\
		X_{t+1}[j^{(t)}_k]
	\end{matrix}}
	=
	\Br{
	\begin{matrix}
		\cos(\agl{t}{k}) & -\sin(\agl{t}{k})\\
		\sin(\agl{t}{k}) & \cos(\agl{t}{k})
	\end{matrix}
	}\!
	\br{\begin{matrix}
		X_{t}[i^{(t)}_k]\\
		X_{t}[j^{(t)}_k]
	\end{matrix}}\enspace.
	$$
	Let
	$F:\br{[n]\times[n]}^m\times [0,2\pi)^{m}\times \SR^n\to \SR^n$
	denote the map associated with the above random walk such that
\begin{equation}\label{eq:F}
  X_{t+1}=F\!\br{P_t,\agl{t}{1},\dots,\agl{t}{m},X_t}\enspace.
\end{equation}
\end{definition}

In one step of the parallel Kac's walk,
we obtain $m$ distinct coordinate pairs
by randomly sampling a perfect matching $P_t$ of set $[n]$.
For each pair, a rotation angle is selected independently and uniformly at random. 
Recall the notation in \defi{okac}.
Let $X_{t,1}=X_t$ and $X_{t,k+1} = G\!\br{i^{(t)}_k, j^{(t)}_k, \agl{t}{k},X_{t,k}}$ for $1\leq k\leq m$.
It is evident that
	$$ X_{t,m+1} = X_{t+1}= F\!\br{P_t,\agl{t}{1},\dots,\agl{t}{m},X_t}\enspace.$$
We can observe that taking one step of the parallel Kac's walk can be
viewed as taking $m = n/2$ steps in the original Kac's walk when there
are \emph{no collisions} in the pairing step. All the subvectors being
rotated in a single step of the parallel Kac's walk are distinct, and
thus not independent.
Consequently, the results for the original Kac's walk cannot be
directly applied. Fortunately, by enhancing the coupling technique for
analyzing the mixing time of the standard Kac's walk, we are able to
prove that the parallel Kac's walk rapidly mixes in time $O(\log n)$
with respect to two different metrics: (1) the Wasserstein
$1$-distance; and (2) the total variation distance.

In the context of the Wasserstein  $1$-distance, after walking $T$
steps, the difference between the output distribution of a parallel
Kac's walk and the normalized Haar measure decays exponentially as $T$
grows, which leads to a $O(\log n)$ mixing time. Formally,

\begin{theorem}\label{thm:mixing_time_w1}
	Let $\st{X_t\in\SR^{n}}_{t\geq 0}$
	be a Markov chain that evolves
	according to the parallel Kac's walk.
	Then, for sufficiently large $n$, $c > 0$, and $T = 10(c+1)\log n$,
	$$
	\sup_{X_0\in \SR^n}\! \wone{\mathcal{L}\!\br{X_{T}}}{\mu} \leq  \frac{1}{2^{c\log n}}\enspace,
	$$
	where $\mu$ is the normalized Haar measure on $\SR^n$.
\end{theorem}

Furthermore, we get a stronger result regarding the total variation distance:

\begin{theorem}\label{thm:mixing_time}
	Let $\st{X_t\in\SR^{n}}_{t\geq 0}$
	be a Markov chain that evolves
	according to the parallel Kac's walk.
	Then, for sufficiently large $n$, $c>515$ and $T = c\log n$,
	$$
	\sup_{X_0\in \SR^n}\! \norm{\mathcal{L}\!\br{X_{T}}-\mu}_{\mathrm{TV}} \leq  \frac{1}{2^{\br{c/515-1}\log n-1}}\enspace,
	$$
	where $\mu$ is the normalized Haar measure on $\SR^n$.
\end{theorem}

Notably, while the Wasserstein $1$-distance is a weaker metric
compared to the total variation distance,  \thm{mixing_time_w1}
provides an adequate foundation for constructing a $\prss$.
Additionally, the analysis of \thm{mixing_time} further reveals a
\emph{dispersing property} of our construction of $\rss$.
The remainder of this section is devoted to proving \thm{mixing_time_w1}. The proof for \thm{mixing_time}, along with an explanation of the dispersing property, is deferred to \app{srss}.

\begin{figure}[!t]
	\centering
	\begin{tikzpicture}
		\coordinate (o) at (0,0);
		\coordinate (i)  at (1,0);
		\coordinate (j) at (0,1);
		\coordinate (xt) at (1,1.5);
		\coordinate (yt) at (0.5,-1);
		\coordinate (xt1) at (-1.5,1);
		\coordinate (yt1) at (-0.930257, 0.620171);
		\draw[->] (o) -- (xt) node [above] {$X_t$};
		\draw[->] (o) -- (xt1) node [above] {$X_{t+1}$};
		\draw[->] (o) -- (yt) node [below] {$Y_t$};
		\draw[->] (o) -- (yt1) node [below] {$Y_{t+1}$};
		\pic[draw, "$\textcolor{orange}{\varphi}$", draw = orange, ->, angle eccentricity=1.5, angle radius=8] {angle = i--o--xt1};
		\pic[draw, "\textcolor{red}{$\theta$}", draw = red, ->, angle eccentricity=1.1, angle radius=45] {angle = xt--o--xt1};
		\pic[draw, "\textcolor{blue}{$\theta'$}", draw = blue, ->, angle eccentricity=1.3, angle radius=25] {angle = yt--o--yt1};
		\draw[->] (-2,0) -- (2,0) node[right] {$i$};
		\draw[->] (0,-2) -- (0,2) node[above] {$j$};
	\end{tikzpicture}
    \caption{Transformation of subcoordinates $X_t[i,j]$ and $Y_t[i,j]$}
    \label{fig:phase1}
\end{figure}
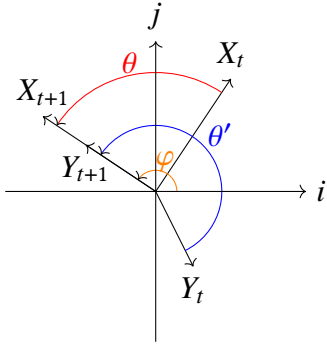

\subsubsection{The Proportional Coupling.}
\label{sec:propc}
Our technique for proving the mixing time in  \thm{mixing_time_w1} accommodates
the  \emph{proportional coupling}~\cite{PS17} that sufficiently reduces 
the distance between two copies of Kac's walk.
At each time $t$ in the proportional coupling (illustrated in \fig{phase1}),
an angle $\theta$ is chosen uniformly at random from $[0,2\pi)$ for rotating the
subvector $\br{X_t[i], X_t[j]}$,
where indices $i$ and $j$ are picked as in \defi{okac}.
The angle ${\theta'}$ is specifically selected
for $\br{Y_t[i], Y_t[j]}$ to
make it collinear with
$\br{X_t[i], X_t[j]}$,
i.e., they share the same argument $\varphi$.
Taking into account the marginal distribution, both $\theta$ and
$\theta'$ are drawn from the uniform distribution over the interval
$[0,2\pi)$, validating the proportional coupling for two Kac's walks.

Following a similar idea, we define the proportional coupling
of two copies of the parallel Kac's walk, which couples each pair of indices from the randomly sampled perfect matching using the proportional coupling.
\begin{definition}[Proportional Coupling for the Parallel Kac's Walk]\label{def:prop_cpl}
We define a coupling of two copies $\st{X_t}_{t\geq 0}, \st{Y_t}_{t\geq 0}$ of the parallel Kac's walk in the following way:
Fix $X_t$, $Y_t\in \SR^n$.
\begin{enumerate}
	\item Choose a perfect matching $P_t = \st{\br{i^{(t)}_1, j^{(t)}_1},\dots,\br{i^{(t)}_m, j^{(t)}_m}}$ and $m$ angles $\agl{t}{1},\dots,\agl{t}{m}\in[0,2\pi)$  uniformly at random, and set $$X_{t+1}=F\!\br{P_t,\agl{t}{1},\dots,\agl{t}{m},X_t}\enspace.$$
	\item Sample $m$ angles ${\aglp{t}{1}},\dots, {\aglp{t}{m}}$ in the following manner: for every $1\leq k\leq m$,
		\begin{enumerate}
		\item choose $\varphi_k\in [0,2\pi)$ uniformly at random among all angles that satisfy
			\begin{align*}
			X_{t+1}[i^{(t)}_k] = \sqrt{X_{t}[i^{(t)}_k]^2+X_{t}[j^{(t)}_k]^2}\cos(\varphi_k)\enspace,\\
			X_{t+1}[j^{(t)}_k] = \sqrt{X_{t}[i^{(t)}_k]^2+X_{t}[j^{(t)}_k]^2}\sin(\varphi_k)\enspace,
			\end{align*}
		\item and then choose ${\aglp{t}{k}}\in [0,2\pi)$ uniformly among the angles that satisfy  \begin{align*}
		\cos({\aglp{t}{k}}) \cdot Y_{t}[i^{(t)}_k] - \sin({\aglp{t}{k}}) \cdot Y_{t}[j^{(t)}_k] &= \sqrt{Y_{t}[i^{(t)}_k]^2+Y_{t}[j^{(t)}_k]^2}\cos(\varphi_k)\enspace,\\
		\sin({\aglp{t}{k}}) \cdot Y_{t}[i^{(t)}_k] + \cos({\aglp{t}{k}}) \cdot Y_{t}[j^{(t)}_k] &= \sqrt{Y_{t}[i^{(t)}_k]^2+Y_{t}[j^{(t)}_k]^2}\sin(\varphi_k)\enspace.
	\end{align*}
	\end{enumerate}
	And set $Y_{t+1}=F\br{P_t,{\aglp{t}{1}},\dots, {\aglp{t}{m}},Y_t}$.
\end{enumerate}
\end{definition}

In this coupling scheme, we enforce $X_t$ and $Y_t$
to employ an identical random matching ($P_t$ in step 1)
to generate all the $m$ pairs of coordinates.
And then we sample $m$ rotation angles for $X_t$ and
obtain $X_{t+1}$ by rotating the $m$ coordinate pairs
by their corresponding angles.
Next, in step 2, we determine the rotation angle
for each coordinate pair of $Y_t$.
For the $k$-th pair,
our objective is to select a suitable angle ${\aglp{t}{k}}$
such that the two-dimensional sub-vector 
$(Y_{t+1}[i^{(t)}_k], Y_{t+1}[j^{(t)}_k])$
aligns collinearly with
$(X_{t+1}[i^{(t)}_k], X_{t+1}[j^{(t)}_k])$.
To achieve this, we ensure that
$(Y_{t+1}[i^{(t)}_k], Y_{t+1}[j^{(t)}_k])$
shares the same argument $\varphi_k$ as
$(X_{t+1}[i^{(t)}_k], X_{t+1}[j^{(t)}_k])$.
Typically, the values of angles $\varphi_k$ and $\aglp{t}{k}$
are uniquely determined.
However, in the scenario where either
$(X_{t}[i^{(t)}_k], X_{t}[j^{(t)}_k])$
or $(Y_{t}[i^{(t)}_k], Y_{t}[j^{(t)}_k])$ equals the zero vector, 
all angles satisfy the required conditions.
In such cases, we resort to uniform random selection
for determining the angles.

\begin{remark}\label{rmk:sign}
	This coupling forces $X_{t+1}[i]Y_{t+1}[i]\geq 0$ for all  $i\in[n]$ since the signs are determined by the same arguments.
\end{remark}
In each step of our coupling scheme,
a quarter of the distance between vectors $X_t$ and $Y_t$ is reduced, 
which is formally shown in 
\begin{lemma} \label{lem:contraction_lemma}
	Let $X_0,Y_0\in \SR^n$. For $t\geq 0$, we couple $(X_{t+1},Y_{t+1})$ conditioned on $(X_{t},Y_{t})$ according to the proportional coupling defined in \defi{prop_cpl}.
	We define
	$$
	A_{t}[i] = X_t[i]^2\enspace,\quad B_{t}[i] = Y_t[i]^2\enspace.
	$$
	Then for any $l\in\N$, we have
	$$
	\E\!\Br{\sum_{i=1}^n \br{A_l[i]-B_l[i]}^2}\leq 2\cdot\br{1-\frac{1}{4}}^l\enspace.
	$$
\end{lemma}
\begin{proof}
	Fix $X_t,Y_t\in \SR^n$. Let $(X_{t+1},Y_{t+1})$ obtained from $(X_{t},Y_{t})$ by applying the coupling defined in \defi{prop_cpl}. Recall that $n=2m$. Let $N=\frac{n!}{2^{m}m!}$ be the number of perfect matchings for $[n]$. To keep the notations short, the perfect matching $\st{\br{i^{(t)}_1, j^{(t)}_1},\dots,\br{i^{(t)}_m, j^{(t)}_m}}$ at step $t$ is denoted by $\br{\overrightarrow{i^{(t)}},\overrightarrow{j^{(t)}}}$.

We have
	\begin{align}
          &\E\!\Br{\sum_{i=1}^n \br{A_{t+1}[i]-B_{t+1}[i]}^2}
            \nonumber \\
          =& \frac{1}{N}\sum_{\br{\overrightarrow{i^{(t)}},\overrightarrow{j^{(t)}}}}\underbrace{\E\!\Br{\left.\sum_{i=1}^n \br{A_{t+1}[i]-B_{t+1}[i]}^2\right|P_t=\br{\overrightarrow{i^{(t)}},\overrightarrow{j^{(t)}}}}}_{(\star)} \enspace .
	\label{eqn:t1}
	\end{align}

By the definition of the parallel Kac's walk, we have

\begin{align}
		(\star)=&\sum_{k=1}^m \E\!
		\Br{\br{\br{A_{t}[i^{(t)}_k]+A_{t}[j^{(t)}_k]}\cos(\varphi_k)^2-\br{B_{t}[i^{(t)}_k]+B_{t}[j^{(t)}_k]}\cos(\varphi_k)^2}^2}\nonumber\\
		&\quad+\sum_{k=1}^m \E\!
		\Br{\br{\br{A_{t}[i^{(t)}_k]+A_{t}[j^{(t)}_k]}\sin(\varphi_k)^2-\br{B_{t}[i^{(t)}_k]+B_{t}[j^{(t)}_k]}\sin(\varphi_k)^2}^2}\nonumber\\
		=&\frac{3}{4}\sum_{k=1}^m
		\br{\br{A_{t}[i^{(t)}_k]+A_{t}[j^{(t)}_k]}-\br{B_{t}[i^{(t)}_k]+B_{t}[j^{(t)}_k]}}^2\nonumber\\
		=&\underbrace{\frac{3}{4}\sum_{k=1}^m
		\br{\br{A_{t}[i^{(t)}_k]-B_{t}[i^{(t)}_k]}^2+\br{A_{t}[j^{(t)}_k]-B_{t}[j^{(t)}_k]}^2}}_{(\star\star)} \nonumber\\	&\quad\quad+\underbrace{\frac{3}{4}\sum_{k=1}^m2\br{A_{t}[i^{(t)}_k]-B_{t}[i^{(t)}_k]}\br{A_{t}[j^{(t)}_k]-B_{t}[j^{(t)}_k]}}_{(\star\star\star)}\enspace,\label{eqn:t2}
\end{align}
where the second equality is by $\E\!\Br{\cos(\varphi_k)^4}=\E\!\Br{\sin(\varphi_k)^4}=3/8$.

As $\st{\br{i^{(t)}_1, j^{(t)}_1},\dots,\br{i^{(t)}_m, j^{(t)}_m}}$ is a perfect matching, we have
\begin{equation}\label{eqn:t3}
  (\star\star)=\frac{3}{4}\sum_{i=1}^{n}\br{A_t[i]-B_t[i]}^2 \enspace .
\end{equation}
Combining Eqs.~\eqref{eqn:t1}\eqref{eqn:t2}\eqref{eqn:t3}, we obtain
\begin{align}\label{eqn:t4}
  &\E\!\Br{\sum_{i=1}^n \br{A_{t+1}[i]-B_{t+1}[i]}^2}= \frac{3}{4}\sum_{i=1}^{n}\br{A_t[i]-B_t[i]}^2+\underbrace{\frac{1}{N}\sum_{\br{\overrightarrow{i^{(t)}},\overrightarrow{j^{(t)}}}}(\star\star\star)}_{(4\star)}\enspace .
\end{align}

For the last term,
\begin{align}
(4\star)=&\frac{3}{2N}\sum_{\br{\overrightarrow{i^{(t)}},\overrightarrow{j^{(t)}}}}\sum_{k=1}^m\br{A_{t}[i^{(t)}_k]-B_{t}[i^{(t)}_k]}\br{A_{t}[j^{(t)}_k]-B_{t}[j^{(t)}_k]}\nonumber\\
		=&\frac{3}{2N}\cdot \frac{(n-2)!}{2^{m-1}(m-1)!}\sum_{i < j} \br{A_t[i]-B_t[i]}\br{A_t[j]-B_t[j]}\nonumber\\
		=&\frac{3\cdot m}{2n(n-1)}\br{\br{\sum_{i=1}^n\br{A_t[i]-B_t[i]}}^2 - \sum_{i=1}^{n}\br{A_t[i]-B_t[i]}^2}\nonumber\\
=&-\frac{3}{4(n-1)}\sum_{i=1}^{n}\br{A_t[i]-B_t[i]}^2\enspace.\label{eqn:t5}
	\end{align}

	Combining Eqs.~\eqref{eqn:t4}\eqref{eqn:t5}, we have
	\begin{align*} 
		\E\!\Br{\sum_{i=1}^n \br{A_{l}[i]-B_{l}[i]}^2} &= \E\!\Br{\E\!\Br{\left.\sum_{i=1}^n \br{A_{l}[i]-B_{l}[i]}^2 \right|X_{l-1},Y_{l-1} } }\\
		&\leq \frac{3}{4}\E\!\Br{\sum_{i=1}^{n}\br{A_{l-1}[i]-B_{l-1}[i]}^2}\\
		&\leq \br{\frac{3}{4}}^l \sum_{i=1}^{n}\br{A_{0}[i]-B_{0}[i]}^2\leq 2\cdot\br{\frac{3}{4}}^l\enspace.
	\end{align*}
\end{proof}

\begin{proof}[Proof of \thm{mixing_time_w1}]
	Let $T=10(c+1)\log n$ for $c>0$. We couple two copies $\st{X_t}_{t\geq 0}$ and $\st{Y_t}_{t\geq 0}$ of the parallel Kac's walk with starting points $X_0=x\in \SR^n$ and $Y_0\sim \mu$, by applying the proportional coupling. We have
	\begin{equation*}
		\wone{\mathcal{L}\!\br{X_{T}}}{\mu} = \wone{\mathcal{L}\!\br{X_{T}}}{\mathcal{L}\!\br{Y_{T}}} \leq \expec{\norm{X_T-Y_T}_2} \leq \br{\expec{\norm{X_T-Y_T}_2^4}}^{1/4}\enspace.
	\end{equation*}
	Then by Cauchy-Schwarz inequality, we have 
	\begin{equation}\label{eq:w11}
		\wone{\mathcal{L}\!\br{X_{T}}}{\mu} \leq \br{n\expec{\norm{X_T-Y_T}_4^4}}^{1/4}\enspace.
	\end{equation}
	Note that the proportional coupling forces $X_T[i]Y_T[i]\geq0$ for all $i\in[n]$. Therefore, for all $i\in[n]$
	$$
	\abs{X_T[i]-Y_T[i]} \leq \abs{X_T[i]+Y_T[i]} \enspace .
	$$
	This gives us
	\begin{equation}\label{eq:w12}
		\norm{X_T-Y_T}_4^4 = \sum_{i=1}^n\br{X_T[i]-Y_T[i]}^4 \leq \sum_{i=1}^n\br{X_T[i]^2-Y_T[i]^2}^2\enspace .
	\end{equation}
	Combining Eqs. \eq{w11} and \eq{w12}, we have
	\begin{align*}
		\wone{\mathcal{L}\!\br{X_{T}}}{\mu} 
		\leq &\br{n\expec{\sum_{i=1}^n\br{X_T[i]^2-Y_T[i]^2}^2}}^{1/4}\\
		(\text{\lem{contraction_lemma}})
		\leq &\br{2n\br{\frac{3}{4}}^{T}}^{1/4}
		\leq \frac{1}{2^{c\log n}} \enspace.
	\end{align*}

\end{proof}

\subsection{Parallel Kac's Walk on Complex Space}
In this section, we extend the parallel Kac's walk to complex vectors.
In the real case,
for each pair of coordinates, we uniformly select a matrix according
to Haar measure on $\so\!\br{2}$.  In the complex case, we will
naturally choose a matrix from $\su\!\br{2}$ according to Haar measure
on it.  The Haar random unitary in $\su\!\br{2}$ can be obtained by
sampling three random angles \cite{ZK94}.  Let
\begin{align}\label{eq:parhaar}
U\!\br{\alpha,\beta,\theta} =
\br{
\begin{matrix}
	e^{\i\alpha }\cos(\theta ) & -e^{\i\beta }\sin(\theta)\\
	e^{-\i\beta }\sin(\theta ) & e^{-\i\alpha }\cos(\theta)
\end{matrix}
}\enspace.	
\end{align}
If we pick $\alpha,\beta\in[0,2\pi)$ and $\zeta\in[0,1)$ uniformly at random and set $\theta = \arcsin \sqrt{\zeta}$, then $U\br{\alpha,\beta,\theta} $ is a Haar random unitary on $\su(2)$.

\paragraph{Kac's walk on complex vectors}
We define Kac's walk on $\SC^n$ as a discrete-time Markov chain $\st{X_t\in \SC^{n}}_{t\geq 0}$.
At each time $t$,
two coordinates $i^{(t)},j^{(t)}\in\st{1,\dots,n}$ and two angles $ \alpha^{(t)},\beta^{(t)}\in[0,2\pi)$ are chosen uniformly at random. Additionally, a real number $\zeta^{(t)}\in[0,1)$ is selected uniformly at random and compute

\[
	\theta^{(t)}=\arcsin{\sqrt{\zeta^{(t)}}}\enspace.
\]
$X_{t+1}$ is obtained by the following update rules:
$$
\br{\begin{matrix}
	X_{t+1}[i^{(t)}]\\
	X_{t+1}[j^{(t)}]
\end{matrix}}
=
U\!\br{\alpha^{(t)},\beta^{(t)},\theta^{(t)}}
\br{\begin{matrix}
	X_{t}[i^{(t)}]\\
	X_{t}[j^{(t)}]
\end{matrix}}\enspace,
$$
$$
X_{t+1}[k]=X_t[k]\quad\text{ for } \quad k\notin\st{i^{(t)},j^{(t)}}\enspace.
$$
We denote the Kac's walk on complex vectors as
$G_{\C}:[n]\times[n]\times [0,2\pi)^3\times \SC^{n}\to \SC^{n}$
such that
$$X_{t+1}=G_{\C}\!\br{i^{(t)},j^{(t)}, \alpha^{(t)},\beta^{(t)},\theta^{(t)}, X_t} \enspace .$$

\paragraph{Parallel Kac's walk on complex vectors}
In a parallel Kac's walk, we choose a perfect matching at each step and apply the one-step Kac's walk $G_{\C}$ on each pair. More specifically, the parallel Kac's walk on complex vectors is
a discrete-time Markov chain $\st{X_t\in \SC^{n}}_{t\geq 0}$.
At each step $t$,
it first selects a random perfect matching of the set $[n]$,
denoted by
$$P_t = \st{\br{i^{(t)}_1, j^{(t)}_1},\dots,\br{i^{(t)}_m, j^{(t)}_m}}$$
where
$\bigcup_{k=1}^{m}\st{i^{(t)}_k, j^{(t)}_k}=[n]$.
And then $2m$ independent angles
$$\alpha_1^{(t)},\dots,\alpha_m^{(t)},\beta_1^{(t)},\dots,\beta_m^{(t)}\in[0,2\pi)$$
are chosen
uniformly at random.
Additionally, $m$ independent real numbers $\zeta_1^{(t)}\dots,\zeta_m^{(t)}\in[0,1)$ are selected uniformly at random and compute
$$
	\theta_k^{(t)}=\arcsin\br{\sqrt{\zeta_k^{(t)}}}
$$
for all $k\in\st{1,\dots,m}$.
Then for every pair $\br{i^{(t)}_k, j^{(t)}_k}$ in $P_t$, it sets
$$
\br{\begin{matrix}
	X_{t+1}[i^{(t)}_k]\\
	X_{t+1}[j^{(t)}_k]
\end{matrix}}
=
U\!\br{\alpha^{(t)}_k,\beta^{(t)}_k,\theta^{(t)}_k}
\br{\begin{matrix}
	X_{t}[i^{(t)}_k]\\
	X_{t}[j^{(t)}_k]
\end{matrix}}\enspace.
$$
Let
$F_{\C}:\br{[n]\times[n]}^m\times [0,2\pi)^{3m}\times \SC^{n}\to \SC^{n}$
denote the map associated with the above random walk such that
\[X_{t+1}=F_{\C}\!\br{P_t,\st{\alpha_k^{(t)}}_{k=1}^m,\st{\beta_k^{(t)}}_{k=1}^m,\st{\theta_k^{(t)}}_{k=1}^m,X_t}\enspace .\]

As the number of steps increases, the output distribution of the parallel Kac's walk on complex vectors converges exponentially fast to the Haar measure in terms of Wasserstein-$1$ distance and total variation distance.
Formally,

\begin{theorem}\label{thm:mixing_time_w1_c}
	Let $\st{X_t\in\SC^{n}}_{t\geq 0}$
	be a Markov chain that evolves
	according to the parallel Kac's walk
	on complex vectors.
	Then, for sufficiently large $n$, $c>0$ and $T = 10(c+1)\log n$,
	$$
	\sup_{X_0\in \SC^n}\! \wone{\mathcal{L}\!\br{X_{T}}}{\mu} \leq  \frac{1}{2^{c\log n}}\enspace,
	$$
	where $\mu_{\C}$ is the Haar measure on $\SC^{n}$.
\end{theorem}

\begin{theorem}\label{thm:mixing_time_c}
	Let $\st{X_t\in\SC^{n}}_{t\geq 0}$
	be a Markov chain that evolves
	according to the parallel Kac's walk
	on complex vectors.
	Then, for sufficiently large $n$, $c>515$ and $T = c\log n$,
	$$
	\sup_{X_0\in \SC^n}\! \norm{\mathcal{L}\!\br{X_{T}}-\mu}_{\mathrm{TV}} \leq  \frac{1}{2^{\br{c/515-1}\log n-1}}\enspace,
	$$
	where $\mu_{\C}$ is the Haar measure on $\SC^{n}$.
\end{theorem}

The proofs of above theorems follow a similar line of reasoning as the proof in the real case.
Therefore, to avoid redundancy, we defer the complete proof to \app{deferred proofs}.

\section{Constructions of \rss s and \prss s}
\label{sec:const1}
In this section, we present a family of circuits that implements
$\rss$s, specifically realizing the parallel Kac's walk on the real
(complex) unit sphere. To obtain circuits for $\prss$s, one can simply
replace the random primitives with their post-quantum secured
pseudorandom counterparts.

\subsection{Constructing (P)RSS over the Real Space}
We begin by constructing a unitary gate that
simulates a single step of the parallel Kac's walk on $\SR^{2^n}$.
In every step, we denote the corresponding permutation by $\sigma \in S_{2^n}$.
And we use the function $f:\bit{n-1}\to \bit{d}$ to manage the precision of the rotation angle
that was originally chosen from the interval $[0, 2\pi)$, where $d$ is the parameter controlling the precision of the rotation angle.
Specifically, for every $\sigma$ and $f$, we define a unitary gate $K_{\sigma,f}=U_{\sigma^{-1}} W_f U_{\sigma}$, where
\begin{align} \label{eq:k}
	U_\sigma = \sum_{x\in\bit{n}} \ketbratwo{\sigma(x)}{x}, \quad
	W_f = \sum_{y\in\bit{n-1}}
			\begin{pmatrix}
				\cos\br{\theta_{y}}	 &	-\sin\br{\theta_{y}} \\
				\sin\br{\theta_{y}}	 &	\cos\br{\theta_{y}}
			\end{pmatrix}
			\otimes \ketbra{y} \enspace,
\end{align}
and $
\theta_y=2\pi\cdot\val{f(y)}$
is the rotation angle for every subvector $\br{\sigma^{-1}( 0y),\sigma^{-1}(1y) }$, $
y\in\bit{n-1}$.
In \fig{kac_circ}, we show a quantum circuit that realizes $K_{\sigma,f}$.

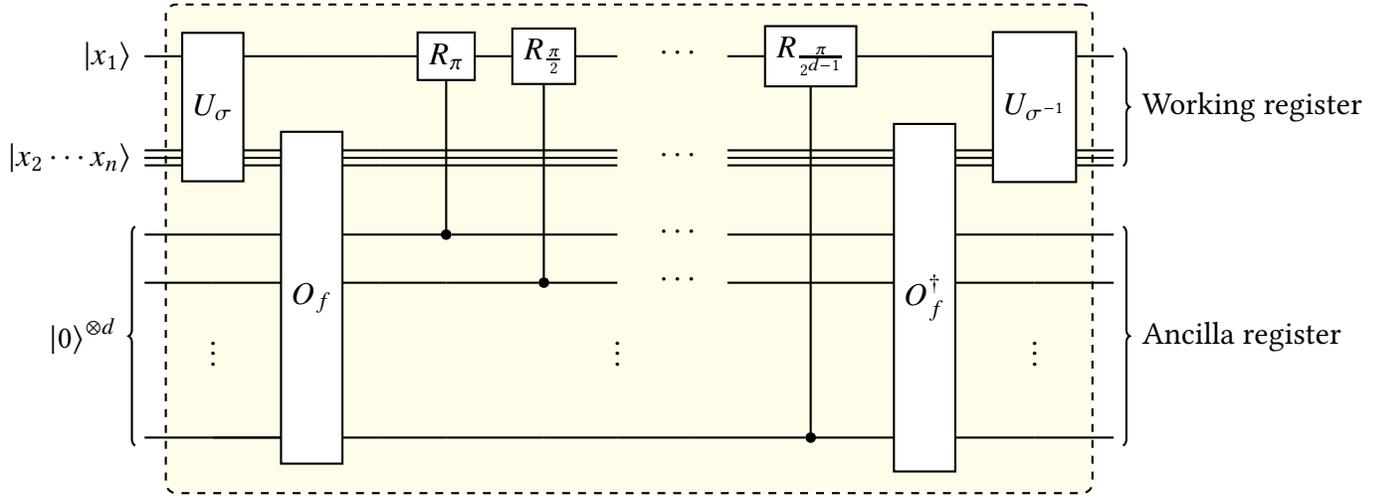
\begin{figure}[bt]
  \centering
  \resizebox{\textwidth}{!}{
  \begin{quantikz}
  \lstick{\ket{x_1}} & \gate[wires=2]{U_{\sigma}} \gategroup[6,steps=10,style={dashed, rounded corners,fill=yellow!10, inner xsep=2pt}, background, label style={yshift=0.2cm}]{ } &  \qw & \qw & \gate{R_{\pi}} & \gate{R_{\frac{\pi}{2}}} &\qw & \cdots\quad & \gate{R_{\frac{ \pi}{2^{d-1}}}}&\qw &\gate[wires=2]{U_{\sigma^{-1}}} &\qw \rstick[2]{Working register} \\
  \lstick{\ket{x_2\cdots x_n}}  & \qwbundle[alternate]{} & \gate[5, nwires={4}, bundle={1}]{O_f} &    \qwbundle[alternate]{}  &  \qwbundle[alternate]{} & \qwbundle[alternate]{} & \qwbundle[alternate]{}& \cdots\quad & \qwbundle[alternate]{}&  \gate[5, nwires={4}, bundle={1}]{O_f^\dagger} & \qwbundle[alternate]{}&\qwbundle[alternate]{}\\
  \lstick[wires=4]{$\ket{0}^{\otimes d}$} & \qw & &\qw &\ctrl{-2}&\qw &\qw &\cdots\quad &\qw &  &\qw &\qw \rstick[4]{Ancilla register} \\
   & \qw & &\qw &\qw & \ctrl{-3}&\qw &\cdots\quad &\qw &  & \qw &\qw\\
   & \vdots  & &  &  & &   \vdots &  & & & \vdots\\
   & \qw & \qw &\qw & \qw & \qw & \qw &\qw & \ctrl{-5} & &\qw & \qw
  \end{quantikz}}
  \caption{Circuit diagram for the construction of the $K_{\sigma,f}$}
  \label{fig:kac_circ}
\end{figure}

The circuit consists of:
\begin{enumerate}
	\item Permutation: a unitary $U_\sigma$
		  which transforms $\ket{x}$
		  to $\ket{\sigma{(x)}}$
		  for any $x\in\bit{n}$.
		  This unitary can be implemented via making quires to oracles $O_{\sigma}$ and $O_{\sigma^{-1}}$, and using $n$ ancilla qubits: for any $x\in\bit{n}$,
		  $$
		  \ket{x}\ket{0}
		  \overset{O_\sigma}{\longrightarrow} \ket{x}\ket{\sigma(x)}
		  \overset{SWAP}{\longrightarrow} \ket{\sigma(x)}\ket{x}
		  \overset{O_{\sigma^{-1}}}{\longrightarrow} \ket{\sigma(x)}\ket{0}\enspace .
		  $$
		  We omit this detail in the above figure for the sake of conciseness.
	\item Implementing rotation operator $W_f$: 
	\begin{enumerate}
		\item an oracle $O\!_f$
	      which queries $f\!\br{x_2,\dots, x_n}$
	      and stores the $d$-bit result in the ancilla qubits.
		\item $d$ controlled-rotation gates.
	      The $i$-th ancilla qubit controls
	      $R_{\frac{\pi}{2^{i-1}}}$ gate acting on the first qubit, where the gate $R_\theta$ denotes the rotation transformation
	      $\begin{pmatrix}
\cos\theta&-\sin\theta\\\sin\theta&\cos\theta
\end{pmatrix}$.
		\item an oracle $O\!_f$ again for uncomputing the ancilla qubits.
	\end{enumerate} 
	\item Inverse permutation: a unitary $U_{\sigma^{-1}}$.
\end{enumerate}
\paragraph{Remark.} The gate $K_{\sigma,f}$ approximates one step of the parallel Kac's walk.
It starts by partitioning the computational basis (indices) 
into $2^{n-1}$ pairs based on a selected permutation $\sigma$.
For each pair $\br{\sigma^{-1}( {0y}), \sigma^{-1}( {1y})}$
labeled by $y\in\bit{n-1}$,
the gate applies a rotation with an approximated angle $\theta_{y}$ indicated by $f$
to the corresponding two dimensional subvector.

\subsubsection{Stepwise State Evolution}
\noindent To gain insight into the functionality of $K_{\sigma,f}$,
we assume the initial state to be a pure state
$$\ket{\varphi}= \sum_{x\in\bit{n}} p_x\ket{x}.$$
First, to pair up the indices by applying $U_\sigma$, 
the initial state is transformed into
\begin{align*}
\sum_{x\in\bit{n}} p_x\ket{\sigma(x)}\otimes \ket{0^d}
&=\sum_{x'\in\bit{n}}p_{\sigma^{-1}(x')}\ket{x'}\otimes \ket{0^d}\\
&=\sum_{y\in\bit{n-1}}\br{p_{\sigma^{-1}( {0y})}\ket{0} + p_{\sigma^{-1}( {1y})}\ket{1}}
\otimes\ket{y}\otimes \ket{0^d}\enspace .
\end{align*}
\noindent Then, to rotate each subvector, the oracle $O_f$ stores $f(y)$ in the ancilla register as control qubits, resulting in the state
$$\sum_{y\in\bit{n-1}} \br{p_{\sigma^{-1}( {0y})}\ket{0} + p_{\sigma^{-1}( {1y})}\ket{1}}
\otimes\ket{y}\otimes \ket{f(y)} \enspace.$$
Next, a series of controlled-rotation gates are applied to 
the first qubit, rotating it by an angle of 
$\theta_{y}= 2\pi\cdot \val{f(y)}$.
Therefore, we have the following state:
$$\sum_{y\in\bit{n-1}} \br{p'_{\sigma^{-1}( {0y})}\ket{0} + p'_{\sigma^{-1}( {1y})}\ket{1}}
\otimes\ket{y}\otimes \ket{f(y)}$$
where
\begin{align*}
	p'_{\sigma^{-1}( {0y})} = \cos\br{\theta_{y}} \cdot p_{\sigma^{-1}( {0y})}-\sin\br{\theta_{y}} \cdot p_{\sigma^{-1}( {1y})}\enspace,\\
	p'_{\sigma^{-1}( {1y})} = \sin\br{\theta_{y}} \cdot p_{\sigma^{-1}( {1y})}+\cos\br{\theta_{y}} \cdot p_{\sigma^{-1}( {1y})}\enspace.
\end{align*}
After reverting the ancilla qubits and applying the inverse permutation, we obtain the output state
{\small
\begin{eqnarray*}
	\sum_{y\in\bit{n-1}} \br{p'_{\sigma^{-1}( {0y})}\ket{\sigma^{-1}( {0y})} + p'_{\sigma^{-1}( {1y})}\ket{\sigma^{-1}( {1y})}}\otimes \ket{0^d} = \sum_{x\in\bit{n}} p'_x\ket{x} \otimes \ket{0^d} \enspace.
\end{eqnarray*}
}

\subsubsection{Constructing \rss s}\label{sec:rss_in_real}
We first define an ensemble $\enssrsscons$ of unitary operators that represents applying
$K_{\sigma,f}$ for $T$-step with i.i.d. random selections of permutations 
and functions. Then, we prove that such an ensemble forms an \rss.

\begin{definition} \label{def:kactorss}
	Let $n,T,d\in\N$,
	and $\H$ be  a real Hilbert space with dimension $2^n$.
	An ensemble of unitary operators
	$\enssrsscons \coloneq \st{\srsscons }_{\secpar}$ with
	$$
	\srsscons \coloneq \st{\srsscons_{\rep{\sigma}{T},\rep{f}{T}} : \H\to\H}
	_{\rep{\sigma}{T} \subseteq S_{2^n},\rep{f}{T}\subseteq \{f:\bit{n-1}\to {\bit{d}}\}}$$
	is define as
	$$
	\srsscons_{\rep{\sigma}{T},\rep{f}{T}}  = K_{\sigma_{T},f_{T}}\cdots K_{\sigma_{2},f_{2}}K_{\sigma_{1},f_{1}}
	$$
	where $K_{\sigma,f}=U_{\sigma^{-1}} W_f U_{\sigma}$ is defined in \eq{k}.
\end{definition}

\begin{theorem}\label{thm:kactorss}
	Let $n\in\N$, $d = \log^2\!\secpar+\log^2\!n$ and $T = 10 (\secpar + 1) n$.
	The ensemble of unitary operators
	$\enssrsscons$ defined in \defi{kactorss}
	is an \rss.
\end{theorem}

To prove \thm{kactorss}, we define a new ensemble of (infinitely many) unitary operators $\enscsrsscons \coloneq \st{\csrsscons }_{\secpar}$ with
	$$\csrsscons \coloneq
	\st{\csrsscons_{\rep{\sigma}{T},\rep{\cf}{T}}:\H\to\H}_{{\rep{\sigma}{T}\subseteq S_{2^n},\rep{\cf}{T}\subseteq \{ f:\bit{n-1}\to [0,1) } \}}
	$$ and
	$$
	\csrsscons_{\rep{\sigma}{T},\rep{\cf}{T}} =
	\widetilde{K}_{\sigma_{T},\widetilde{f}_{T}} \cdots
	\widetilde{K}_{\sigma_{2},\widetilde{f}_{2}}
	\widetilde{K}_{\sigma_{1},\widetilde{f}_{1}}
	$$
	where $\widetilde{K}_{\sigma,\widetilde{f}}=U_{\sigma^{-1}} \widetilde{W}_{\widetilde{f}} U_{\sigma}$
	and $\widetilde{W}_{\widetilde{f}}$ is defined to be
	\begin{align} \label{eq:hatk}
			\widetilde{W}_{\widetilde{f}} = \sum_{y\in\bit{n-1}}
				\begin{pmatrix}
					\cos\br{\widetilde{\theta}_{y}}	 &	-\sin\br{\widetilde{\theta}_{y}} \\
					\sin\br{\widetilde{\theta}_{y}}	 &	\cos\br{\widetilde{\theta}_{y}}
				\end{pmatrix}
				\otimes \ketbra{y} \enspace,
		\end{align}
	in which
	$ \widetilde{\theta}_{y} = 2\pi \cdot \widetilde{f} \! \br{y} $
	for $y \in \bit{n-1}$.
	
	$\enssrsscons$ and $\enscsrsscons$ differ in the way the angles are chosen.
	In $\enssrsscons$, the angles are selected
	from the discrete set $\st{2\pi \cdot \frac{i}{2^d} : i \in \st{0,1,\ldots,2^d-1}}$,
	while in $\enscsrsscons$, the angles are chosen
	from the interval $[0,2\pi)$.
	For uniformly random $\sigma$ and $\cf$, applying gate $\ck_{\sigma,\cf}$
	results in the selection of
	a random matching on the computational basis,
	with each pair in the matching
	being rotated by a random angle in $[0,2\pi)$
	determined by the corresponding value of $\cf$.
	This is exactly one step of parallel Kac's walk
	described in \sect{kac}.
	$\enscsrsscons$ serves as an intermediate scrambler in the proof of \thm{kactorss}.
	To analyse the difference between $\enssrsscons$ and $\enscsrsscons$, we need the following lemma.
	
	\begin{lemma} \label{lem:k_hk_close}
	Let $\sigma\in S_{2^n}$
	and $\widetilde{f}:\bit{n-1}\to[0,1)$.
	Let $f$ be the function satisfying for any $y\in\bit{n-1}$,
	$f(y)$ is the $d$ digits after the binary point in $\widetilde{f}(y)$.
	Then
	$$
	{
		\norm{K_{\sigma,f} - \widetilde{K}_{\sigma,\widetilde{f}}
	}_\infty }
	\leq 2^{1-d} \pi \enspace ,
	$$
	where $K_{\sigma,f}=U_{\sigma^{-1}} W_f U_{\sigma}$
	is defined in \eq{k} and
	$\widetilde{K}_{\sigma,\widetilde{f}}=U_{\sigma^{-1}} \widetilde{W}_{\widetilde{f}} U_{\sigma}$
	is defined in \eq{hatk}.
	\end{lemma}
	
	\begin{proof}
	\begin{align*}
	{\norm{K_{\sigma,f} - \widetilde{K}_{\sigma,\widetilde{f}}}_\infty}
	=~&{\norm{U_{\sigma^{-1}}\br{W_f-\widetilde{W}_{\widetilde{f}}}U_{\sigma}}_\infty}\\
	=~&{\norm{\sum_{y\in\bit{n-1}}\begin{pmatrix}
	\cos\theta_y-\cos\widetilde{\theta}_y&-\br{\sin\theta_y-\sin\widetilde{\theta}_y}\\
	\sin\theta_y-\sin\widetilde{\theta}_y&\cos\theta_y-\cos\widetilde{\theta}_y
	\end{pmatrix}\otimes\ketbra{y}}_\infty}\\
	=~&{\max_{y\in\bit{n-1}}\set{2\abs{\sin\frac{\theta_y-\widetilde{\theta}_y}{2}}\norm{\begin{pmatrix}
	-\sin\frac{\theta_y+\widetilde{\theta}_y}{2}&-\cos\frac{\theta_y+\widetilde{\theta}_y}{2}\\
	\cos\frac{\theta_y+\widetilde{\theta}_y}{2}&-\sin\frac{\theta_y+\widetilde{\theta}_y}{2}\\
	\end{pmatrix}}_\infty}}\\
	\leq~&{\max_{y\in\bit{n-1}}\set{\abs{\theta_y-\widetilde{\theta}_y}}}
	\leq 2^{1-d}\pi\enspace.
	\end{align*}
	\end{proof}

\begin{proof}[Proof of \thm{kactorss}]	
It is easy to see that the uniformity condition is satisfied.
	Let $\kappa$ denote the key length.
	Quantum circuit $\srsscons$ applies $\srsscons_{\rep{\sigma}{T},\rep{f}{T}}$ after reading $\rep{\sigma}{T}$ and $\rep{f}{T}$.
	To implement $\srsscons_{\rep{\sigma}{T},\rep{f}{T}}$, we need to realize each of
	the $T = 10(\lambda+1)n$ unitary gates $K$.
	Since each gate $K$ can be implemented in $\poly{n,\secpar,\klen}$ time,
	the total construction time for $\srsscons_{\rep{\sigma}{T},\rep{f}{T}}$
	is also $\poly{n,\secpar,\klen}$.
	
Thus, it suffices to prove the requirement of \emph{Statistical Pseudorandomness} is satisfied.
	Fix $\ket{\eta}\in\S(\H)$. Define three distributions:
	\begin{itemize}
		\item $\nu$ be the distribution of
			$\srsscons_{\rep{\sigma}{T}, \rep{f}{T}}\!\!\ket{\eta}$
			with independent and uniformly random permutations $\rep{\sigma}{T}\subseteq S_{2^n}$
			and random functions $\rep{f}{T}$ $\subseteq \{f: \bit{n-1}\to\bit{d}\}$.
		\item $\widetilde{\nu}$ be the distribution of
			$\csrsscons_{\rep{\sigma}{T},\rep{\cf}{T}}\!\!\ket{\eta}$
			with independent and uniformly random permutations $\rep{\sigma}{T}\subseteq S_{2^n}$,
			and random functions $\rep{\cf}{T} $ $\subseteq \{f: \bit{n-1}\to[0,1)\}$.
		\item $\mu$ be the Haar measure on $\SR^{2^n}$.
	\end{itemize}
	
	We first prove that the trace distance between $\nu$ and $\widetilde{\nu}$ is negligible.
	To this end, we construct a coupling $\gamma_0$ of $\nu$ and $\widetilde{\nu}$ by using the same permutation $\sigma_t$ and letting $f_t$ be the function satisfying $f_t(y)$ is the $d$ digits after the binary point in $\widetilde{f}_t(y)$ for all $y\in\bit{n-1}$. Therefore, for any $\br{\ket{\phi}, \ket{\varphi}}\sim \gamma_0$,
	we have
	\begin{align}
		\norm{\ket{\phi} - \ket{\varphi}}_2 ~
	 =	~& \norm{
			\csrsscons_{\rep{\sigma}{T},\rep{\cf}{T}}\!\!\ket{\eta} -
			\srsscons_{\rep{\sigma}{T},\rep{f}{T}}\!\!\ket{\eta}
		}_2 \nonumber\\
	\leq  ~& \norm{
			\csrsscons_{\rep{\sigma}{T},\rep{\cf}{T}} -
			\srsscons_{\rep{\sigma}{T},\rep{f}{T}}
		}_\infty \nonumber\\
		\leq ~& 2^{1-d}\pi T = \frac{20\pi (\lambda+1)n}{\lambda^{\log \lambda}\cdot n^{\log n}} \enspace, \nonumber
	\end{align}
	where the last inequality is from \fct{uprod} and \lem{k_hk_close}.
	Thus, for any $l\in\poly{\secpar,n}$
	\begin{align} \label{eq:h1_h2_rss}
		~&\norm{
			\E_{\ket{\phi}\sim \nu}\!\Br{\br{\ketbra{\phi}}^{\otimes l}}
			-
			\E_{\ket{\varphi}\sim \widetilde{\nu}}\!\Br{\br{\ketbra{\varphi}}^{\otimes l}}
		}_1 \nonumber\\
		\leq ~& \E_{\br{\ket{\phi}, \ket{\varphi}}\sim \gamma_0}\!\Br{
		\norm{
			\br{\ketbra{\phi}}^{\otimes l}
			-
			\br{\ketbra{\varphi}}^{\otimes l}
		}_1} \nonumber\\
		\leq ~& l \E_{\br{\ket{\phi}, \ket{\varphi}}\sim \gamma_0}\!\Br{
		\norm{
			\ketbra{\phi}
			-
			\ketbra{\varphi}
		}_1} \nonumber\\
		\leq ~& l
		\br{\E_{\br{\ket{\phi}, \ket{\varphi}}\sim \gamma_0}\!\Br{
		\norm{
			\ket{\phi}\br{ \bra{\phi} - \bra{\varphi}}
		}_1} + \E_{\br{\ket{\phi}, \ket{\varphi}}\sim \gamma_0}\!\Br{
		\norm{
			\br{ \ket{\phi} - \ket{\varphi}}\bra{\varphi}
		}_1}}\nonumber\\
		\leq ~& 2l
		\E_{\br{\ket{\phi}, \ket{\varphi}}\sim \gamma_0}\!\Br{
		\norm{ \ket{\phi} - \ket{\varphi} }_2 } \leq \frac{40\pi (\lambda+1)nl}{\lambda^{\log \lambda}\cdot n^{\log n}} \enspace.
	\end{align}
	As for the trace distance between $\widetilde{\nu}$ and $\mu$, note that $\widetilde{\nu}$ is the output distribution of $T$-step parallel Kac's walk. Thus by \thm{mixing_time_w1}, we have
	\begin{align*}
		\wone{\widetilde{\nu}}{\mu} \leq \frac{1}{2^{\lambda n}} \enspace.
	\end{align*}
	So there exists a coupling of $\widetilde{v}$ and $\mu$, denoted by $\gamma_1$, that achieves
	\begin{align*}
		\E_{\br{\ket{\varphi}, \ket{\psi}}\sim \gamma_1}
		\!\Br{ \norm{\ket{\varphi}-\ket{\psi}}_2 }\leq \frac{3}{2^{\lambda n}} \enspace .
	\end{align*}
	Therefore, similar to Eq. \eq{h1_h2_rss}, we have for any $l\in\poly{\secpar,n}$
	\begin{align} \label{eq:h2_h3_rss}
		\norm{
			\E_{\ket{\varphi}\sim \widetilde{\nu}}\!\Br{\br{\ketbra{\varphi}}^{\otimes l}}
			-
			\E_{\ket{\psi}\sim \mu}\!\Br{\br{\ketbra{\psi}}^{\otimes l}}
		}_1 
		\leq 2l
		\E_{\br{\ket{\varphi}, \ket{\psi}}\sim \gamma_1}\!\Br{
		\norm{ \ket{\varphi} - \ket{\psi} }_2 } \leq \frac{6l}{2^{\lambda n}}\enspace .
	\end{align}
	Finally, by the triangle inequality, Eqs.~\eq{h1_h2_rss} and \eq{h2_h3_rss}, we have
	\begin{align*}
		\norm{
			\E_{\ket{\phi}\sim \nu}\!\Br{\br{\ketbra{\phi}}^{\otimes l}} -
			\E_{\ket{\psi}\sim \mu}\!\Br{\br{\ketbra{\psi}}^{\otimes l}} }_1 \leq
			\frac{40\pi (\lambda+1)nl}{\lambda^{\log \lambda}\cdot n^{\log n}} +
			\frac{6l}{2^{\lambda n}} = \negl{\secpar} \enspace.
	\end{align*}
	This establishes the \emph{Statistical Pseudorandomness} property.
\end{proof}

\subsubsection{Constructing \prss~}
\label{sec:prss_in_real}
We construct a \prss~by replacing the random functions and permutations used in \rss~with \qprf s and \qprp s.
\begin{definition} \label{def:kactoprss}
	Let $n, T\in\N$,
	$\H$ be a real Hilbert space with dimension $2^n$,
	$\tau:\K_1\times\bit{n}\to\bit{n}$
	be a \qprp~with key space $\K_1$ and
	$F :\K_2\times\bit{n-1}\to\bit{d}$
	be a \qprf~with key space $\K_2$.
	An ensemble of unitary operators
	$\enssgen \coloneq \st{\sgen}_{\secpar}$ with
	$
		\sgen \coloneq \st{\sgen_k :\H\to\H}_{k\in \br{\K_1\times\K_2}^T}
	$
	is defined as
	$$
		\sgen_k = K_{\tau_{r_T},F_{s_T}}\cdots K_{\tau_{r_2},F_{s_2}}K_{\tau_{r_1},F_{s_1}}
	$$
	for $k = \br{r_1,s_1,r_2,s_2,\dots,r_T,s_T}\in \br{\K_1\times\K_2}^T$,
	where $K_{\sigma,f}=U_{\sigma^{-1}} W_f U_{\sigma}$ is defined in \eq{k}.

\end{definition}

\begin{theorem}\label{thm:realprss}
	Let $n\in\N$, $d = \log^2\!\secpar+\log^2\!n$ and $T = 10 (\secpar + 1) n$.
	The ensemble of unitary operators
	$\enssgen$ defined in \defi{kactoprss} is a \prss.
\end{theorem}

\begin{proof}
Due to the efficiency of $\tau$ and $F$, the key length is bounded by $2T\cdot\poly{n,d} = \poly{n, \secpar}$. Thus the condition of polynomial-bounded key length is satisfied.
	To implement $\sgen_k$,
	we need to realize each of the $T = 10 (\secpar + 1) n$ unitary gates $K$
	that make up $\sgen_k$.
	Since each $K$ can be realized in $\poly{n,\secpar}$ time (efficiency of $\tau$ and $F$),
	the overall construction time for $\sgen_k$ will be $\poly{n,\secpar}$. Thus the uniformity is also satisfied.
	
	We now prove the pseudorandomness property.
	To this end, we consider three hybrids for an arbitrary $\ket{\phi}\in\S(\H)$ and $l\in\poly{\secpar,n}$:
	\begin{itemize}
		\item[H1:] \label{itm:h1}
		$\ket{\phi_k}^{\otimes l}$ for
		$\ket{\phi_k} = \sgen_k\!\ket{\phi}$
		where
		$k \leftarrow \br{\K_1\times\K_2}^T$ is chosen uniformly at random.
		
		\item[H2:] \label{itm:h2}
		$\ket{\varphi_{\rep{\sigma}{T},\rep{f}{T}}}^{\otimes l}$ for
		$\ket{\varphi_{\rep{\sigma}{T},\rep{f}{T}}} = \srsscons_{\rep{\sigma}{T},\rep{f}{T}}\!\!\ket{\phi}$
			with independently and uniformly random permutations $\rep{\sigma}{T}\subseteq S_{2^n}$
			and random functions $\rep{f}{T}\subseteq\{f: \bit{n-1}\to\bit{d}\}$.
			$\srsscons_{\rep{\sigma}{T},\rep{f}{T}}$
			is defined in \defi{kactorss}.
		
		\item[H3:] \label{itm:h3}
		$\ket{\psi}^{\otimes l}$ for $\ket{\psi}$ chosen according to the Haar measure $\mu$ on $\SR^{2^n}$.
	\end{itemize}
	
	We first prove that \itm{h1}{H1} and \itm{h2}{H2}
	are computationally indistinguishable.
	By the quantum-secure property of $\tau$ and $F$,
	we know the following two situations are computationally indistinguishable for any polynomial-time quantum oracle algorithm $\A$ (see \lem{twokeys}):
	\begin{itemize}
		\item given oracle access to
		$\tau_{r_1},\cdots,\tau_{r_T}$ and $F_{s_1},\cdots,F_{s_T}$
		where
		$\rep{r}{T}\subseteq\K_1$ and $\rep{s}{T}\subseteq\K_2$
		are independently and uniformly random keys.
		
		\item given oracle access to
		independent and uniformly random permutations $\rep{\sigma}{T}$ $\subseteq S_{2^n}$
		and random functions $\rep{f}{T}\subseteq \{f: \bit{n-1}\to\bit{d}\}$.
	\end{itemize}
	Thus, we have for any
	polynomial-time quantum algorithm $\A$,
	$$
	\abs{
		\Pr\Br{\A\br{\ket{\phi_k}^{\otimes l}} = 1} -
		\Pr\Br{\A\br{\ket{\varphi_{\rep{\sigma}{T},\rep{f}{T}}}^{\otimes l}} = 1}
	} = \negl{\secpar} \enspace .
	$$
	
	For \itm{h2}{H2} and \itm{h3}{H3},
	they are statistically indistinguishable since $\enssrsscons$ defined in \defi{kactorss}
	is an \rss~by \thm{kactorss}.
	Finally, by the triangle inequality we establish \itm{h1c}{H1} and \itm{h3c}{H3}
	are computationally indistinguishable. This accomplishes the proof.
\end{proof}

\subsection{Constructing (P)RSS over the Complex Space}
\label{sec:rss_in_complex}
This section provides constructions of a \rss~and a \prss~over $\C$.
Similar to the case over $\R$,
our initial step is to create a unitary gate
that can be utilized to simulate a single iteration
of parallel Kac's walk within a complex Hilbert space.
Fix $f,g,h:\bit{n-1}\to\bit{d}$ and $\sigma\in S_{2^n}$.
Let $\widehat{L}_{\sigma,f,g,h}=U_{\sigma^{-1}} \widehat{Q}_{f,g,h} U_{\sigma}$, where $U_{\sigma}$ is defined as before and $\widehat{Q}_{f,g,h}$ is 
{\small
\begin{align}\label{eq:widehatq}
	\sum_{y\in\bit{n-1}}\!\!
	\left(
		\begin{matrix}
			e^{\i\br{\frac{\alpha_y+\beta_y}{2}} } & 0\\
			0 & e^{-\i\br{\frac{\alpha_y+\beta_y}{2}}}
		\end{matrix}
	\right)\!\!\!
	\begin{pmatrix}
		\cos\theta_{y}&-\sin\theta_{y}\\
		\sin\theta_{y}&\cos\theta_{y}
	\end{pmatrix}\!\!\!
	\left(
		\begin{matrix}
			e^{\i\br{\frac{\alpha_y-\beta_y}{2}} } & 0\\
			0 & e^{-\i\br{\frac{\alpha_y-\beta_y}{2}}}
		\end{matrix}
	\right)
	\otimes\ketbra{y} ,
\end{align}
}
in which
\[
\theta_y=\arcsin\br{\sqrt{\val{f(y)}}}\enspace, \enspace
\alpha_y=2\pi\cdot\val{g(y)}\enspace , \enspace
\beta_y=2\pi\cdot\val{h(y)}\enspace .
\]
Here we decompose $U\!\br{\alpha_y,\beta_y,\theta_y}$ into a product of three matrices according to the formula
\begin{align}\label{eq:decomp_u}
	\left(
		\begin{matrix}
			e^{\i\alpha }\cos\theta & -e^{\i\beta }\sin\theta\\
			e^{-\i\beta }\sin\theta & e^{-\i\alpha }\cos\theta
		\end{matrix}
	\right) =
	\left(
		\begin{matrix}
			e^{\i(\frac{\alpha+\beta}{2}) } & 0\\
			0 & e^{-\i(\frac{\alpha+\beta}{2})}
		\end{matrix}
	\right)\!\!
	\left(
		\begin{matrix}
			\cos\theta & -\sin\theta\\
			\sin\theta & \cos\theta
		\end{matrix}
	\right)\!\!
	\left(
		\begin{matrix}
			e^{\i(\frac{\alpha-\beta}{2}) } & 0\\
			0 & e^{-\i(\frac{\alpha-\beta}{2})}
		\end{matrix}
	\right).
\end{align}
We approximate $\widehat{Q}_{f,g,h}$ by another unitary $Q_{f,g,h}$ which can be constructed as follows applying the similar technique in \sect{const1}:
\begin{itemize}
	\item apply $O\!_f, O\!_g, O\!_h$ and store the results $f(y), g(y), h(y)$ in ancilla qubits,
	\item calculate three parameters $\gamma_y^{+}\approx\frac{\vall{g(y)}+\vall{h(y)}}{2}$, $\gamma_y^{-}\approx\frac{\vall{g(y)}-\vall{h(y)}}{2}$ and $\xi_y\approx \frac{2}{\pi}\arcsin\br{\sqrt{\val{f(y)}}}$ with a precision up to $d$ bits after the binary point,
	\item use $\gamma_y^{+}$, $\xi_y$ and $\gamma_y^{-}$ in the above step to construct a series of controlled gates on the first qubit which approximates the three matrices in the RHS of \eq{decomp_u},
	\item uncompute all the ancilla qubits.
\end{itemize}
As a result, we construct a unitary gate $L_{\sigma,f,g,h}=U_{\sigma^{-1}} Q_{f,g,h} U_{\sigma}$ where $Q_{f,g,h}$ is
\begin{align}\label{eq:normalq}
	\sum_{y\in\bit{n-1}}
	\left(
		\begin{matrix}
			e^{\i2\pi\gamma_y^{+}} & 0\\
			0 & e^{-\i2\pi\gamma_y^{+}}
		\end{matrix}
	\right)
	\begin{pmatrix}
		\cos\br{\frac{\pi}{2}\xi_{y}}&-\sin\br{\frac{\pi}{2}\xi_{y}}\\
		\sin\br{\frac{\pi}{2}\xi_{y}}&\cos\br{\frac{\pi}{2}\xi_{y}}
	\end{pmatrix}
	\left(
		\begin{matrix}
			e^{\i2\pi\gamma_y^{-} } & 0\\
			0 & e^{-\i2\pi\gamma_y^{-}}
		\end{matrix}
	\right)
	\otimes\ketbra{y} ,
\end{align}
and for any $y\in\bit{n-1}$,
	{\small\[
		\abs{\frac{\pi}{2}\xi_y-\theta_y}\leq 2^{-d-1} \pi \enspace , \enspace
		\abs{2\pi\gamma_y^{+}-\frac{\alpha_y+\beta_y}{2}}\leq 2^{1-d}\pi \enspace , \enspace
		\abs{2\pi\gamma_y^{-}-\frac{\alpha_y-\beta_y}{2}}\leq 2^{1-d}\pi \enspace .
	\]}

By utilizing the gate $L_{\sigma,f,g,h}$
together with random permutations and random functions,
we can implement the following scheme to produce an \rss:

\begin{definition} \label{def:kactosrssc}
	Let $n,T,d\in\N$,
	and $\H$ be a complex Hilbert space with dimension $2^n$.
	An ensemble of unitary operators
	$ \enssrssconsc \coloneq \st{\srssconsc }_{\secpar} $ with $\srssconsc\coloneq$
	{\small
	$$\st{\srssconsc_{\rep{\sigma}{T},\rep{f}{T},\rep{g}{T},\rep{h}{T}}}
	_{\rep{\sigma}{T} \subseteq S_{2^n},\rep{f}{T},\rep{g}{T},\rep{h}{T} \subseteq \{f: \bit{n-1}\to {\bit{d}}\}}$$}
	is define as
	$$
		\srssconsc_{\rep{\sigma}{T},\rep{f}{T},\rep{g}{T},\rep{h}{T}}
		= L_{\sigma_{T},f_{T},g_{T},h_{T}}\cdots
		  L_{\sigma_{2},f_{2},g_{2},h_{2}}
		  L_{\sigma_{1},f_{1},g_{1},h_{1}}
	$$
	where $L_{\sigma,f,g,h}=U_{\sigma^{-1}} Q_{f,g,h} U_{\sigma}$ is defined in \eq{normalq}.
\end{definition}

\begin{theorem}\label{thm:kactorssc}
	Let $n\in\N$, $d = 2\br{\log^2\!\secpar + \log^2\!n}$ and $T = 10 (\secpar + 1) n$.
	The ensemble of unitary operators
	$ \enssrssconsc$
	defined in \defi{kactosrssc}
	is an \rss.
\end{theorem}

Based on this, we obtain a \prss~by substituting the random functions and permutations utilized in \rss~with their quantum-secure pseudorandom counterparts.

\begin{definition} \label{def:kactoprssc}
	Let $n, T\in\N$,
	$\H$ be a complex Hilbert space with dimension $2^n$,
	$\tau:\K_1\times\bit{n}\to\bit{n}$
	be a \qprp~with key space $\K_1$ and
	$F :\K_2\times\bit{n-1}\to\bit{d}$
	be a \qprf~with key space $\K_2$.
	An ensemble of unitary opertors
	$\enssgenc \coloneq \st{ \sgenc }_{\secpar}$ with
	$ \sgenc = \st{\sgenc_k : \H\to\H}_{k\in \br{\K_1\times\K_2\times\K_2\times\K_2}^T} $
	is defined as
	$$\sgenc_k=
			L_{\tau_{r_T},F_{u_T},F_{s_T},F_{t_T}}\cdots
			L_{\tau_{r_2},F_{u_2},F_{s_2},F_{t_2}}
			L_{\tau_{r_1},F_{u_1},F_{s_1},F_{t_1}}$$
	for $k=\br{r_1,u_1,s_1,t_1,r_2,u_2,s_2,t_2,\dots,r_T,u_T,s_T,t_T}\in \br{\K_1\times\K_2\times\K_2\times\K_2}^T$,
	where $L_{\sigma,f,g,h}=U_{\sigma^{-1}} Q_{f,g,h} U_{\sigma}$ is defined in \eq{normalq}.
\end{definition}

\begin{theorem}\label{thm:complexprss}
	Let $n\in\N$, $d = 2\br{\log^2\!\secpar + \log^2\!n}$ and $T = 10 (\secpar + 1) n$.
	The ensemble of unitary operators
	$ \enssgenc$
	defined in \defi{kactoprssc}
	is a \prss.
\end{theorem}

The detailed proofs of two theorems above share similarities with the real case.
For this reason, we attach them in \app{deferred proofs}.

\section{Applications}
\label{sec:app}
Since pseudorandom state scramblers subsume pseudorandom state
generators and its siblings in the literature, all applications enabled
by $\prsg$s can also be obtained from $\prss$s. This includes for instance
symmetric-key encryption and commitment of classical messages as well
as secure computation. In this section, we showcase a few novel
applications beyond what $\prsg$s are capable of.

\subsection{Compact Quantum Encryption}
\label{sec:cqenc}

Because $\prss$s map any initial state to a pseudorandom output state,
we can readily employ them to encrypt quantum messages. Furthermore,
it turns out that $\prss$-based quantum encryption schemes offer
improvements in terms of \emph{compactness}, a point we discuss below.

We start by recalling the well-known Quantum One-Time Pad, which is the
quantum analogue of one-time pad and achieves perfect secrecy. Given
an $n$-qubit state $\ket{\psi}$, we sample a uniform $2n$-bit key
$k = k_1 \| k_2$ with $k_1,k_2 \in\bit{n}$ and encrypt $\ket{\psi}$ by
\[ \ket{\psi_k} = \qotp_{k}\!\ket{\psi} = X^{k_1}Z^{k_2}\! \ket{\psi} \,
  , \]
where $X$ and $Z$ are Pauli operators applied on each qubit of
$\ket{\psi}$.

We can reduce the key length by using pseudorandom keys. For instance,
given a pseudorandom generator $\prg: \bit{n} \to \bit{2n}$, we can
expand a uniform $n$-bit key under $\prg$ and use $\prg(k)$ as the key
to \qotp. Namely we encrypt by
\[ \ket{\psi_k} = \qotp_{\prg(k)}\!\ket{\psi} \, . \]
We refer to this scheme as prg-\qotp.

These two schemes are secure if the same key is never used more than
\emph{once}. One can extend it to multi-time security with
\emph{hybrid} encryption, using in addition a post-quantum secure
encryption for \emph{classical} bits. For concreteness, we use a
post-quantum $\prf_k: \bit{n} \to \bit{2n}$. To encrypt $\ket{\psi}$,
we sample a uniformly random string $r$, and use $\prf_k(r)$ as the
key to \qotp, i.e., we output cipherstate $\br{r, \ket{\psi_{k,r}}}$
where
\[ \ket{\psi_{k,r}} = \qotp_{\prf_k(r)}\!\ket{\psi} \, . \] We call this
scheme prf-\qotp.

Now suppose we have a $\prss$ ($\{\SG{k}{n,m} \}$) with key space
$\K = \bit{\klen}$, and for simplicity we assume that $n=m$ and we
ignore them in the notation. We can construct three encryption
schemes, analogous to each of the schemes above.

\begin{itemize} 
\item \prss-enc: on random key $k$ and state $\ket{\psi}$, output
  $\ket{\psi_k} : = \SG{k}{}\!\ket{\psi}$.
\item prg-\prss-enc: given a $\prg: \bit{n} \to \K$, on random key $k$
  and state $\ket{\psi}$, output
  $\ket{\psi_k} : = \SG{\prg(k)}{}\!\ket{\psi}$.
\item prf-\prss-enc: given a $\prf: \bit{n}\to \K $, the key is a
  random key $k$ for the $\prf$. On state $\ket{\psi}$, output
  $\br{r, \ket{\psi_{k,r}}}$, where $r \gets \bit{n}$ and

  \[ \ket{\psi_{k,r}} = \SG{\prf_k(r)}{}\!\ket{\psi} \, . \]
\end{itemize}

\begin{table}[h!]
  \centering
  \renewcommand{\arraystretch}{1.5}
  \newcolumntype{C}[1]{>{\centering\arraybackslash}m{#1}} 
  \begin{tabular}[h!]{| C{0.18\textwidth} |C{.38\textwidth} | C{0.38\textwidth} |}
    \hline
    $\ncopy$ copies of $\ket{\psi}$ & \multicolumn{1}{c|}{(prg-)\qotp} & \multicolumn{1}{c|}{\hlight{(prg-)\prss-enc}}
    \\ [0.5ex]
    \hline
    $\ncopy = 1$ & $\ket{\psi_k} = \qotp_{k\text{ or } \prg(k)}
                   \!\ket{\psi}$ & $\ket{\psi_k} =
                                 \prss_{k \text{ or } \prg(k)}\!\ket{\psi}$ \\
    \hline
     $\ncopy > 1$ & $\br{\ket{\psi_{k_1}},\ldots,
                   \ket{\psi_{k_\ncopy}}}$ & $\br{\ket{\psi_{k}},\ldots,
                                             \ket{\psi_k}}$ \\
    \hline 
    Comparison & Need to exchange $\ncopy$ indep. keys $k_1,
                 \ldots,
                 k_\ncopy$ &  \hlight{Single} key for any polynomial
                             $\ncopy$ \\
    \hline 
  \end{tabular}

  \smallskip

  \begin{tabular}[h!]{| C{0.18\textwidth} | C{.38\textwidth} | C{0.38\textwidth} |}
    \hline
    $\ncopy$ copies of $\ket{\psi}$ & \multicolumn{1}{c|}{prf-\qotp} & \multicolumn{1}{c|}{\hlight{prf-\prss-enc}}
    \\ [0.5ex]
    \hline
    $\ncopy = 1$ & $\br{r, \ket{\psi_{k,r}} = \qotp_{\prf_k(r)}\!\ket{\psi}}$
                                                                       & $\br{r,
                                                                         \ket{\psi_{k,r}}} = \prss_{\prf_k(r)}
                                                                         \!\ket{\psi}$  \\
    \hline
    $\ncopy > 1$ & $\br{\ldots, \br{r_j,
                   \ket{\psi_{k,r_j}}}, \ldots}$ & $\br{\hlight{r},
                                                   \ket{\psi_{k,r}},\ldots,
                                                   \ket{\psi_{k,r}}}$
    \\
    \hline 
    Comparison & Cipher size grows by $\ncopy$
                 factor
                                                                       &  Cipher size
                                                                         grows by
                                                                         $\hlight{\frac{1}{2}}(\ncopy+1)$\\
    \hline 
  \end{tabular}
  \vspace{1em}
  \caption{Advantages of PRSS-based encryptions: maintaining single
    key instead of linear number of keys or reducing the cipher size
    growth factor by half.}
  \label{tab:prssenc}
\end{table}

\paragraph{Advantages of \prss-based quantum encryption.} One distinct
benefit of \prss-enc over $\qotp$ is that we can encrypt
\emph{multiple copies} of a state $\ket{\psi}$ using \prss-enc under
the \emph{same} key $k$. This follows from the multi-copy
indistinguishability in our $\prss$ definition. In contrast, $\qotp$
needs independent keys to encrypt each copy of $\ket{\psi}$. This
considerably improves \emph{compactness}, and it holds similarly in
the other two types of schemes.

A related concept called \emph{quantum private broadcasting} has been investigated by
Broadbent, Gonz\`alez-Guill\'en and Schuknecht \cite{BGS22}. They employ (symmetric)
$t$-designs to encrypt $t$ copies of an $n$-qubit quantum message. While the key length in their construction scales logarithmically with $t$, it grows exponentially with $n$.
Our PRSS-based scheme maintains a key size of $\poly{n}$.

We stress that this applies only to encrypting multiple copies of the
\emph{same} input state. If we want to encrypt different states, then
fresh keys in (prg-)\prss-enc or randomness in prf-\prss-enc should be
used.

\subsection{Succinct Quantum State Commitment}
\label{sec:sqsc}

Next we show how $\prss$ enables quantum commitment. Bit commitment is
a fundamental primitive in cryptography. A sender Alice commits to an
input bit $b$ to a receiver Bob in the \emph{commit} phase, which can
be revealed later in the \emph{open} phase. This naturally extends to
committing bit strings. Two properties are essential.
\begin{itemize}
\item \emph{Hiding.} Bob is not able to learn the message $b$ before the open
  phase. 
\item \emph{Binding.} Alice cannot fool Bob to accept a different message
  $b' \ne b$ in the open phase.
\end{itemize}

We will focus on \emph{non-interactive} commitment schemes where both
commit and open phases consist of a single message from the sender to
the receiver. If the protocol involves exchanging and processing
quantum information, we call it a quantum bit commitment (\qbc)
scheme. $\qbc$ has been extensively studied, and it is shown that
$\qbc$ can be constructed based on standard $\prsg$s~\cite{AQY22,MY22}.

In a similar vein, one can also consider committing to a
\emph{quantum} input state, and this is called quantum state
commitment (\qsc). $\qsc$ has proven useful such as in zero-knowledge
proof systems for $\qma$~\cite{BJSW20,BG22}.

Recently, Gunn, Ju, Ma and Zhandry give a systematic treatment on
$\qsc$~\cite{GJMZ23}. They propose a new characterization of binding
termed \emph{swap-binding}. They show a striking hiding-binding
\emph{duality} theorem for (non-interactive) quantum commitment:
binding holds if the \emph{opening} register held by the sender hides
the input state. This significantly simplifies proving binding. They
then construct binding commitment schemes which in addition are
\emph{succinct}, where the register containing the commitment has a
\emph{smaller} size than the message state.\footnote{Note that hiding
  is not required in these succinct schemes.}

\paragraph{Succinct $\qsc$ from $\prss$.} The succinct $\qsc$ schemes
by~\cite{GJMZ23} are based on post-quantum one-way functions or the
potentially weaker primitive of pseudorandom unitary operators
(\pru). We show below the viability of building succinct commitment on
\prss. 
Specifically, we observed that a \emph{succinct} $\prss$ implies a succinct one-time
quantum encryption, and it was shown in~\cite{GJMZ23} that a succinct one-time quantum encryption gives a succinct $\qsc$ scheme.
Hence (succinct one-time) \prss s offer an alternate approach of realizing succinct one-time quantum encryptions based on potentially weaker assumptions than one-way functions, and could be weaker than the instantiation via \pru s in~\cite{GJMZ23}.
Meanwhile, one-time quantum encryption does not seem to follow immediately from $\prs$
or other primitives implied by $\prs$.

\begin{theorem}
  Assuming a succinct $\prss$, i.e., $|\K| < 2^\dimin$, there exists a
  succinct $\qsc$.
\label{thm:sqsc-prss}
\end{theorem}

\begin{proof} This follows from a generic claim in~\cite{GJMZ23}. They
  show that any one-time secure quantum encryption scheme with
  \emph{succinct} keys, where the key is shorter than the state to be
  encrypted, readily gives a succinct \qsc. A $\prss$ is a secure
  quantum encryption as discussed above. Succinctness translates if
  $\prss$'s key length is shorter than the size of the input
  state. This is stated below. We choose not to fully spell out the
  syntax and definitions of the involved primitives for the sake of
  clarity, and refer the readers to~\cite{GJMZ23}.
\end{proof}

\begin{lemma} 
  Assuming a succinct $\prss$, i.e., $|\K| < 2^\dimin$, there exists a
  succinct one-time quantum encryption scheme.
  \label{lemma:otqenc-prss}
\end{lemma}

How to instantiate a succinct $\prss$? Our construction
is not immediately succinct, because the key length
$\Omega(\secpar\cdot\dimin)$. We can remedy this by using a
pseudorandom generator to expand a key shorter than $\dimin$ into
pseudorandom keys for each iteration (\qprf~and \qprp).

\bibliographystyle{splncs04}
\bibliography{ref}

\newpage
\appendix
\section{Dispersing \rss}
\label{sec:srss}
As mentioned before, our parallel Kac's walk also mixes rapidly
in terms of total variation distance, which endows
our scramblers with a unique dispersing property.
We first give the formal definitions
of a random scrambler with a dispersing property in \sect{drss}.
Then, \sect{tvmixing} analyses the total variation mixing time of the parallel Kac's walk on both real and complex Hilbert spaces.
Lastly in \sect{provedrss}, we utilize this rapid mixing property to demonstrate that the ensembles of unitary operators we construct in \sect{const1} do exhibit a dispersing property.

\subsection{Definitions}
\label{sec:drss}

We introduce the concept of \emph{dispersing random state scramblers} (\srss s),
which ensure the approximation of Haar randomness with respect to Wasserstein distance.

\begin{definition}[Dispersing Random State Scrambler]
  \label{def:trssb}
  Let $\hspacein$ and $\hspaceout$ be Hilbert spaces of dimensions
  $2^{\dimin}$ and $2^\dimout$ respectively with
  $\dimin, \dimout \in\N$ and $\dimin \leq \dimout$. Let
  $\K=\bit{\klen}$ be a key space, and $\secpar$ be a security
  parameter. A \emph{dispersing random state scrambler} (\srss) is an
  ensemble of isometric operators
  $\SG{}{\dimin,\dimout} := \{ \SG{}{\dimin,\dimout,\secpar}
  \}_\secpar$ with
  $\SG{}{\dimin,\dimout,\secpar} : = \{\SG{k}{\dimin,
    \dimout,\secpar}: \hspacein \to \hspaceout \}_{k\in \K}$ satisfying:

  \begin{itemize}
	
  \item \textbf{Sphere Coverage.}
  There exist $\epsilon= \negl{\secpar}$ such that for any $\ket{\phi}\in\Sin$, the family of states $\{\SG{k}{n,m}\!\ket{\phi}\}_{k\in\K}$ forms an $\epsilon$-net of $\Sout$.
  
  \item \textbf{Wasserstein Approximation of Haar randomness.}  There
    exist $\delta, \delta' = \negl{\secpar}$ such that for any $\ket{\phi}\in\Sin$,

    \begin{itemize}

    \item Let $\nu$ be the distribution of $\SG{k}{n,m}\!\ket{\phi}$
      with uniformly random $k \gets \K$, and $\mu$ be the Haar
      measure on $\Sout$.  Then, there exists a distribution
      $\widetilde{\nu}$ such that

      \[
        \tv{\mu-\tilde{\nu}}\leq \delta, \quad \text{and} \quad
        W_{\infty}(\tensorsrss, \widetilde{\nu}) \le \delta' .\]

    \end{itemize}

  \item \textbf{Uniformity}. $\scram^{n,m}$ can be uniformly computed in
    polynomial time. That is, there is a deterministic Turing machine
    that, on input $(1^{\dimin}, 1^{\dimout},1^\secpar,1^\klen)$,
    outputs a quantum circuit $Q$ in
    $\poly{\dimin,\dimout,\secpar,\klen}$ time such that for all
    $k\in \K$ and $\ket{\phi}\in\Sin$
    \[ Q\ket{k}\!\ket{\phi} = \ket{k}\!\ket{\phi_k}\, , \] where
    $\ket{\phi_k} := \SG{k}{n,m,\secpar}\! \ket{\phi}$.
  \end{itemize}
\end{definition}

In particular, small Wasserstein distance
implies small trace distance between the average states drawn from the
two distributions.

\begin{proposition}
A \srss~is an \rss~with the same parameters.
\end{proposition}
\begin{proof}
  It suffices to prove that Wasserstein approximation of Haar
  randomness implies statistical pseudorandomness. Let $\nu$ and $\mu$
  be the distribution of the output states of a \srss~and the Haar
  measure, respectively. By the assumption, there exists a
  distribution $\widetilde{\nu}$ on $\Sout$ and a coupling $\gamma$ of
  $\nu$ and $\widetilde{\nu}$ such that
  \[\prob{(\ket{\psi},\ket{\psi'})\sim\gamma}{\norm{\ket{\psi}-\ket{\psi'}}_2}\le\negl{\secpar}.\]
  Notice that $\ell$ is polynomial in $\lambda$. By the triangle inequality,
  \[\td\br{\expect{\ket{\psi}\sim \nu}{ \ketbra{\psi}^{\otimes \ncopy}},
        \expect{\ket{\psi}\sim \widetilde{\nu}}{ \ketbra{\psi}^{\otimes \ncopy}}}\le\negl{\secpar}.\]
  By \lem{h3h4}, the condition that $\tv{\mu-\widetilde{\nu}}\leq \negl{\secpar}$ implies
  \[\td\br{\expect{\ket{\psi}\sim \widetilde{\nu}}{ \ketbra{\psi}^{\otimes \ncopy}},
        \expect{\ket{\psi}\sim \haar}{ \ketbra{\psi}^{\otimes \ncopy}}}
      \le \negl{\secpar}.\]
  The result follows from the triangle inequality.
\end{proof}

Moreover, we introduce a continuous version of random state scrambler, where continuous randomness is allowed.

\begin{definition}[Continuously Random State Scrambler]
  \label{def:ctrss}
  Let $\hspacein$ and $\hspaceout$ be Hilbert spaces of dimensions
  $2^{\dimin}$ and $2^\dimout$ respectively with
  $\dimin, \dimout \in\N$ and $\dimin \leq \dimout$. Let $\K$ be a
  (continuous) key space, and $\secpar$ be a security parameter. A
  \emph{continuously random state scrambler} (\csrss) is an ensemble
  of isometric operators
  $\SG{}{\dimin,\dimout} := \{ \SG{}{\dimin,\dimout,\secpar}
  \}_\secpar$ with
  $\SG{}{\dimin,\dimout,\secpar} : = \{\SG{k}{\dimin,
    \dimout,\secpar}: \hspacein \to \hspaceout \}_{k\in \K}$ satisfying:

  \begin{itemize}

  \item \textbf{Total-Variation Approximation of Haar randomness.} Let $\ket{\phi}\in\Sin$ be an arbitrary pure state. Let
    $\tensorsrss$ be the distribution of $ \SG{k}{n,m,\secpar}\ket{\phi}$ with uniformly
    random $k\gets \K$, and $\mu$ be the Haar measure on $\Sout$. Then
    there exists $\delta = \negl{\secpar}$ such that the total
    variation distance between $\nu$ and $\mu$ is at most $\delta$,
    i.e., $\norm{\nu-\mu}_{\mathrm{TV}}\leq \delta$.
  \end{itemize}

\end{definition}

\subsection{Total Variation Mixing Time of the Parallel Kac's Walk}
\label{sec:tvmixing}

We assume $n=2m$ for some $m\in\N$ throughout this section.

\subsubsection{Background: Two-Phase Proof Strategy in \cite{PS17}}
It is proved in \cite{PS17} that the total variation mixing time
of Kac's walk on $\SR^{n}$ is $\Theta\!\br{n\log n}$.

\begin{theorem}[Theorem 1, \cite{PS17}]
	Let $\st{X_t\in\SR^{n}}_{t\geq 0}$
	be a Markov chain that evolves
	according to Kac's walk.
	Then, for sufficiently large $n$, and $T > 200 n \log n$,
	$$
	\sup_{X_0\in \SR^n}\! \norm{\mathcal{L}\!\br{X_{T}}-\mu}_{\mathrm{TV}} = O\!\br{\frac{1}{\poly{n}}}\enspace,
	$$
	where $\mu$ is the normalized Haar measure on $\SR^n$.
\end{theorem}

We informally revisit their proof approach that
utilizes the famous coupling lemma
(see \lem{coupling_lemma}).
The coupling lemma offers a practical approach to estimate
the mixing time of a Markov chain by comparing the behavior
of two coupled random walks.
The total variation distance at time $T$ is bounded by
the probability that two coupled random walks are distinct at time $T$.
Through a two-phase coupling of
$\st{X_t}_{t\geq 0}$ and $\st{Y_t}_{t\geq 0}$,
they show that the probability that two copies are not equal at time $T$,
i.e., $\Pr\!\Br{X_T\neq Y_T}$, approaches zero when $n$ is sufficiently large and $T > 200 n \log n$.
Specifically, the two-step coupling consists of
an initial \emph{contracting phase}
followed by a subsequent \emph{coalescing phase}.
The contracting phase aims to sufficiently reduce the distance
between two copies of Kac's walk
so that during the coalescing phase, they can be
further fine-tuned to coalesce.
These two phases are described below, with a focus on
how the random angles are coupled.

\textbf{Contracting Phase (from $t=0$ to $T_0$).}
The contracting phase starts from time $0$
and continues until time $T_0$.
In this phase,
$\st{X_t}$ and $\st{Y_t}$ undergo the \emph{proportional coupling}.
which aims at reducing the distance
between two copies of Kac's walk.
The proportional coupling is introduced in \sect{kac}.

\textbf{Coalescing Phase (from $t = T_0$ to $T_0+T_1$).}
This phase employs a non-Markovian coupling\footnote{
The non-Markovian coupling refers to a situation
where the transition between states depend
not only on the current state but also on the future states,
violating the memoryless property of a standard Markov process.}
starting from a close-by pair $X_{T_0}$ and $Y_{T_0}$.
Let $T_1$ be determined later.
 Initially,
 the coupling independently and identically samples $T_1$ pairs of coordinates
 $\st{\br{i_t,j_t}}_{t=T_0}^{T_0+T_1-1}$ all at once,
 which are subsequently used to generate $T_1+1$ partitions of $[n]$,
 denoted by $\set{\PP_t}_{t=T_0}^{T_0+T_1}$.
 The construction of these partitions is done inductively in reverse order.
 The last partition is enforced to be
 $\PP_{T_0+T_1} = \{\{1\}, \{2\}, \ldots, \{n\}\}$.
 Starting from $t=T_1+T_0-1$ and decrementing down to $T_0$,
 the construction of partition $\PP_{t}$ uses the
 chosen coordinate pair $(i_t,j_t)$ as a guide.
 Specifically, it is generated by merging two sets in $\PP_{t+1}$:
 one set includes $i_t$, and the other includes $j_t$;
 while leaving other sets untouched.
 Then, the value of $T_1$ is determined such that $\PP_{T_0}=\set{[n]}$ with high probability.

The aim of this phase is to ultimately coalesce $X$ and $Y$.
To see how to achieve this, we introduct the event $\A_t$ in which, at time $t$,
$$
\sum_{i\in S} X_t [i]^2 = \sum_{i\in S} Y_t [i]^2,\quad \forall S\in \PP_{t}\enspace
$$
and
\[X_t[k]Y_t[k]\geq0,\quad k\in\set{i_{t-1},j_{t-1}}.\]
Intuitively, event $\A_t$ states that if we partition $X_t$ and $Y_t$ based on $\PP_{t}$,
then both $X_t$ and $Y_t$ carry equal significance
within each segment $S$ at time $t$;
and meanwhile the corresponding updated subvectors share the same sign at time $t-1$.
Conditioning on $\PP_{T_0} = \st{\st{1,\cdots,n}}$, which holds with high probability,
it is not hard to verify that
$\A_{T_0}$ occurs
and
$\cap_{t=T_0}^{T_0+T_1}\A_{t}$ implies that the corresponding entries
in vectors $X_{T_0+T_1}$ and $Y_{T_0+T_1}$ are equal.
Thus, to prove that $X$ and $Y$ are identical by the end of this phase, it suffices to prove that all events occur with a high probability during the process.

So, 
the non-Markovian coupling aims to ensure that
$\A_{t+1}$ takes place with a high probability,
conditioned on all previous events occur.
This is achieved by sampling $\theta$ and $\theta'$ from a ``good'' joint distribution,
which makes sure that both marginal distributions are uniformly distributed on $[0,2\pi)$.
Such a desirable distribution is made possible by the entry-wise closeness achieved during the first phase.

Intuitively, the parallel Kac's walk is expected
to have a mixing time that is only on the order of $O(\log n)$,
saving factor of $n$ in the
mixing time of the original Kac's walk.
However, the mixing time for the parallel Kac's walk cannot be derived from the mixing time for the original Kac's walk, directly.
Fortunately, through careful modifications to the two-phase coupling approach above, we have discovered a logarithmic mixing time for the parallel Kac's walk, resulting in exponential speedup compared to the original random walk.

\subsubsection{Real Case}
As we introduced in \sect{kac}, the total variation distance
between the output distribution of
a parallel Kac's walk after $T$ steps
and the normalized Haar measure on $\SR^n$
decays exponentially as $T$ grows. We restate the theorem here.

\begin{reptheorem}{thm:mixing_time}
	Let $\st{X_t\in\SR^{n}}_{t\geq 0}$
	be a Markov chain that evolves
	according to the parallel Kac's walk.
	Then, for sufficiently large $n$, $c>515$ and $T = c\log n$,
	$$
	\sup_{X_0\in \SR^n}\! \norm{\mathcal{L}\!\br{X_{T}}-\mu}_{\mathrm{TV}} \leq  \frac{1}{2^{\br{c/515-1}\log n-1}}\enspace,
	$$
	where $\mu$ is the normalized Haar measure on $\SR^n$.
\end{reptheorem}

To prove \thm{mixing_time}, we will use the coupling lemma
(see \lem{coupling_lemma})
and extend the two-phase coupling method described in \cite{PS17}
to accommodate parallel Kac's walks.
We have already extended the proportional coupling
used in the contracting phase in \sect{propc}.
We now introduce the non-Markovian coupling
employed in the coalescing phase to ensure that 
$X$ and $Y$ converge to an identical state,
and integrate these two couplings into
a comprehensive two-phase coupling
to establish the mixing time of the parallel Kac's walk.

\paragraph{Coalescing Phase: the Non-Markovian Coupling}
The non-markovian coupling is defined in \defi{non_mark_cpl}.
As we will see later,
if the initial vectors of two parallel Kac's walks,
namely $X_{T_0}$ and $Y_{T_0}$,
are close,
this coupling guarantees
a high probability of
collision between $X_T$ and $Y_T$,
when $T$ is sufficiently large.

\begin{definition}[Non-Markovian Coupling]\label{def:non_mark_cpl}
Fix $T_0\leq T\in\N$. We couple $\st{X_t}_{T_0\leq t\leq T}, \st{Y_t}_{T_0\leq t\leq T}$ in the following way:
\begin{enumerate}
	\item For each $T_0\leq t < T$, choose a perfect matching $$P_t = \st{\br{i^{(t)}_1, j^{(t)}_1},\dots,\br{i^{(t)}_m, j^{(t)}_m}}$$ uniformly at random.

	\item Set $\PP_{T,1} = \st{\st{1},\dots,\st{n}}$, and define a sequence of partitions $$\st{\PP_{t,k}}_{T_0\leq t<T,~ 1\leq k\leq m+1}$$ of $[n]$ inductively by the process:
		\begin{enumerate}
			\item If $k=m+1$, let $\PP_{t,k}=\PP_{t+1,1}$.
			\item If $1\leq k\leq m$, write $\PP_{t,k+1}=\st{S_1(t,k+1),\dots, S_{l_{t,k+1}}(t,k+1)}$ with $S_r(t,k+1)\subseteq [n]$ for $1\leq r\leq l_{t,k+1}$. Let $u_{t,k},~ v_{t,k}$ be the indices such that
			$$i^{(t)}_{k}\in S_{u_{t,k}}(t,k+1)\quad \text{and}\quad j^{(t)}_{k}\in S_{v_{t,k}}(t,k+1)\enspace.$$
			\begin{enumerate}
				\item If $u_{t,k}= v_{t,k}$, set $\PP_{t,k}=\PP_{t,k+1}$.
				\item If $u_{t,k}\neq v_{t,k}$, construct $\PP_{t,k}$ by merging $S_{u_{t,k}}(t,k+1)$ and $S_{v_{t,k}}(t,k+1)$ in $\PP_{t,k+1}$.
			\end{enumerate}
		\end{enumerate}
		
	\item If $\PP_{T_0,1}=\st{[n]}$, we couple $\st{X_t}_{T_0\leq t\leq T}, \st{Y_t}_{T_0\leq t\leq T}$ in the following way:
	\begin{itemize}
			\item Define the set
			\begin{equation}\label{eq:H}
H=\st{(t,k):T_0\leq t<T,~1\leq k\leq m,~\PP_{t,k}\neq\PP_{t,k+1}} \enspace.
\end{equation}
			
			\item Fix $T_0\leq t<T$, $X_{t}$ and $Y_{t}$, and we couple $X_{t+1}$ and $Y_{t+1}$ in the following way:
				\begin{enumerate}
					\item Set $X_{t,1}=X_t$ and $Y_{t,1}=Y_t$.
					\item For $1\leq k\leq m$,
						\begin{enumerate}
							\item If $(t,k)\notin H$, uniformly choose $\agl{t}{k}\in[0,2\pi)$. Let $$X_{t,k+1} = G(i_{k}^{(t)},j_{k}^{(t)},\agl{t}{k},X_{t,k})\ \text{and}\  Y_{t,k+1} = G(i_{k}^{(t)},j_{k}^{(t)},{\aglp{t}{k}},Y_{t,k})$$ where $G(\cdot)$ is defined in \cref{eq:kaconestep} and ${\aglp{t}{k}}$ is obtained in the same way as the proportional coupling defined in \defi{prop_cpl}.
							
							\item If $(t,k)\in H$, let $\theta_0$ be the angle satisfies
							$$
							X_{t,k}[i_{k}^{(t)}]=\sqrt{X_{t,k}[i_{k}^{(t)}]^2+X_{t,k}[j_{k}^{(t)}]^2}\cos(\theta_0)\enspace,$$
							$$X_{t,k}[j_{k}^{(t)}]=\sqrt{X_{t,k}[i_{k}^{(t)}]^2+X_{t,k}[j_{k}^{(t)}]^2}\sin(\theta_0)\enspace,
							$$
							$\theta_0'$ be the angle satisfies
							$$
							Y_{t,k}[i_{k}^{(t)}]=\sqrt{Y_{t,k}[i_{k}^{(t)}]^2+Y_{t,k}[j_{k}^{(t)}]^2}\cos(\theta_0')\enspace,$$
							$$Y_{t,k}[j_{k}^{(t)}]=\sqrt{Y_{t,k}[i_{k}^{(t)}]^2+Y_{t,k}[j_{k}^{(t)}]^2}\sin(\theta_0')\enspace,
							$$
							and then choose the best distribution $\nu$ among all joint distributions on $[0,2\pi)\times [0,2\pi)$ with both marginal distributions uniformly distributed on $[0,2\pi)$ which maximizes the probability of the following events when $(\theta,\theta')\sim\nu$:
								\begin{align*}
									\sum_{i\in S_{r}(t,k+1)} X_{t,k+1} [i]^2 &= \sum_{i\in S_{r}(t,k+1)} Y_{t,k+1} [i]^2\enspace, \enspace1\leq r\leq l_{t,k+1}\\
									X_{t,k+1} [i]\cdot Y_{t,k+1}[i]&\geq 0\enspace, \enspace i\in\st{i_{k}^{(t)},j_{k}^{(t)}}\enspace,
								\end{align*}
								where
								$$X_{t,k+1} = G(i_{k}^{(t)},j_{k}^{(t)},\theta-\theta_0,X_{t,k})\ \text{and}\  Y_{t,k+1} = G(i_{k}^{(t)},j_{k}^{(t)},\theta'-\theta_0',Y_{t,k})\enspace.$$
								Then choose $(\agl{t}{k},{\aglp{t}{k}})\sim \nu$, and set
								$$X_{t,k+1} = G(i_{k}^{(t)},j_{k}^{(t)},\agl{t}{k}-\theta_0,X_{t,k})\ \text{and}\  Y_{t,k+1} = G(i_{k}^{(t)},j_{k}^{(t)},{\aglp{t}{k}}-\theta_0',Y_{t,k})\enspace.$$
						\end{enumerate}
					\item Set $X_{t+1} = X_{t,m+1}$ and $Y_{t+1} = Y_{t,m+1}$.
				\end{enumerate}
		\end{itemize}
	
	\item If $\PP_{T_0,1}\neq \st{[n]}$, for $T_0\leq t\leq T$, we couple $X_{t+1}$ and $Y_{t+1}$ in the following way: choose $m$ independent angles $\agl{t}{1},\dots,\agl{t}{m}\in [0,2\pi)$ uniformly at random and set
	$$X_{t+1}=F\br{P_t,\agl{t}{1},\dots,\agl{t}{m},X_t}\quad \text{and}\quad Y_{t+1}=F\br{P_t,\agl{t}{1},\dots,\agl{t}{m},Y_t}\enspace,$$
where $F(\cdot)$ is given in \cref{eq:F}.
\end{enumerate}
\end{definition}

Step 1 samples $T-T_0$ matchings,
generating all coordinate pairs
that will be updated in the succeeding process.
Step 2 utilizes this matchings to construct a series of partitions of $[n]$ in a back propagation manner.
Starting from $\PP_{T,1} = \st{ \st{1},\dots,\st{n} }$,
it sequentially construct 
$$\PP_{T-1,m+1} \enspace
\enspace\PP_{T-1,m}\enspace \enspace\cdots\enspace
\enspace\PP_{T-1,1}\enspace 
\enspace\PP_{T-2,m+1}\enspace \enspace\cdots\enspace
\enspace\PP_{T-2,1}\enspace \enspace\cdots\enspace
\enspace\PP_{T_0,1}\enspace.$$
$\PP_{t,m+1}$ is set equal to $\PP_{t+1,1}$ directly.
For $1\leq k\leq m$, $\PP_{t,k}$ is obtained based on $\PP_{t,k+1}$
and the $k$-th pair of coordinates in matching $P_t$.
If two coordinates of the $k$-th pair
belong to different components in partition $\PP_{t,k+1}$,
we merge these two components.
Otherwise, $\PP_{t,k}$ is set equal to $\PP_{t,k+1}$.
This series of partitions thus consists of random partitions of set $[n]$
and with high probability the first partition $\PP_{T_0,1}$ 
is $\st{[n]}$ (see \lem{good_start}). This follows from the argument for bounding the probability of the
connectivity of Erd\"os-R\'enyi graphs~\cite[Theorem
7.3]{Bollobas01}. The proof is deferred to \app{deferred proofs}.

\begin{lemma}\label{lem:good_start}
	Fix $c>0$ and $T_0\in\N$. Let $l=5(1+c)\log n$ and $T=T_0+l$.	Then we have for $n$ sufficiently large,
	$$
		\Pr\!\Br{\PP_{T_0,1}\neq\st{[n]}} \leq 2n^{-c}\enspace,
	$$
	where $\PP_{T_0,1}$ is defined in \defi{non_mark_cpl}.
\end{lemma}

If $\PP_{T_0,1} = \st{[n]}$,
step 3 serves as the crucial step of
this coupling.
To provide a clearer explanation of how this coupling technique works,
\defi{AT1} introduces a series of events $\st{\A(t,k)}$.
Intuitively, $\A(t,k)$ indicates that
the $(k-1)$-th updated coordinates in both vectors have the same signs at time $t$,
and
both vectors have the same weight within each component of the partition $\PP_{t,k}$.

\begin{definition}\label{def:AT1}
Let $\A(T_0,1)$ denote the event
\begin{align*}
	\sum_{i\in S_{r}(T_0,1)} X_{T_0,1} [i]^2 &= \sum_{i\in S_{r}(T_0,1)} Y_{T_0,1} [i]^2, 1\leq r\leq l_{T_0,1}\enspace.
\end{align*}
For other $T_0\leq t \leq T$ and
$1\leq k \leq m+1$,
we define the event $\A(t,k)$ as
\footnote{In Eq.~\eq{cond2}, if $k=1$, then $i\in\st{ i_{m}^{(t-1)},j_{m}^{(t-1)} }$.}
	\begin{align}
		\sum_{i\in S_{r}(t,k)} X_{t,k} [i]^2 &= \sum_{i\in S_{r}(t,k)} Y_{t,k} [i]^2\enspace, \enspace1\leq r\leq l_{t,k}\label{eq:cond1}\\
		X_{t,k} [i]\cdot Y_{t,k}[i]&\geq 0\enspace, \enspace i\in\st{i_{k-1}^{(t)},j_{k-1}^{(t)}}
		\enspace. \label{eq:cond2}
	\end{align}
\end{definition}

It is worthy noting that under the assumption $\PP_{T_0,1} = \st{[n]}$,
$\A(T_0,1)$ occurs since the sum of squares of the components of a unit vector is equal to $1$.
And if $\{\A(t,k)\}$ take place in the entire process,
we can conclude that $X_T = Y_T$ because 
$\A(T,1)$ guarantees the corresponding coordinates have the same absolute values
and \rmk{sign} together with \eq{cond2} guarantees that they have the same signs as well.
So we want to couple the rotation angles to ensure that
$\A(t,k+1)$ occurs with a high probability,
given that all previous events have already taken place.
The coupling of the rotation angles are divided into two cases:
whether the coordinate pair $\br{i^{(t)}_k, j^{(t)}_k}$ is a “merge point” in the construction process of partitions.
If not, the rotation angles are coupled using the proportional coupling (see step 3.b.i).
For non-“merge point”,
the proportional coupling does not affect the partition
and also maintain the weight within each component,
and thus $\A(t,k+1)$ must occur conditioned on $\A(t,k)$.
In the other case, we sample the rotation angles from a “good” joint distribution which maximizes the probability that $\A(t,k+1)$ occurs (see step 3.b.ii).

We next show why such a “good” joint distribution exists.
Before that, we set some notations.
For $T_0\leq t \leq T$ and $1\leq k \leq m+1$, define
\begin{align}\label{eqn:Atk}
A_{t,k}[i] = X_{t,k}[i]^2\enspace,\quad B_{t,k}[i] = Y_{t,k}[i]^2\enspace.
\end{align}
For $T_0\leq t_1,t_2\leq T$ and
$1\leq k_1,k_2\leq m+1$,
we define a partial order
$\sqsubseteq$ as
$$
(t_1,k_1)\sqsubseteq (t_2,k_2)\quad \text{iff} \quad
\br{t_1<t_2} \vee \br{t_1=t_2 \wedge k_1\leq k_2 }\enspace.
$$
Note that if $\br{i^{(t)}_k, j^{(t)}_k}$ is a “merge point”,
$\PP_{t,k+1}$ differs from $\PP_{t,k}$ only in
the components in which $i^{(t)}_k$ and $ j^{(t)}_k$ are,
namely $S_{u_{t,k}}(t,k+1)$ and $S_{v_{t,k}}(t,k+1)$.
To see whether $\A(t,k+1)$ occurs given $\A(t,k)$, we only need to check
whether \eq{cond1} holds for $r = u_{t,k}$, that is,
$$
\sum_{i\in S_{u_{t,k}}(t,k+1)} X_{t,k} [i]^2 
= \sum_{i\in S_{u_{t,k}}(t,k+1)} Y_{t,k} [i]^2\enspace.
$$
Rewrite 
$$
A=\sum_{i\in S_{u_{t,k}}(t,k+1)\backslash i^{(t)}_k} X_{t,k}[i]^2\enspace,
\quad \quad
B=X_{t,k} [i^{(t)}_k]^2+X_{t,k} [j^{(t)}_k]^2\enspace,
$$
$$
C=\sum_{i\in S_{u_{t,k}}(t,k+1)\backslash i^{(t)}_k} Y_{t,k}[i]^2
\enspace,\quad \quad
D=Y_{t,k}[i^{(t)}_k]^2+Y_{t,k}[j^{(t)}_k]^2\enspace.
$$
We will have
$$
\sum_{i\in S_{u_{t,k}}(t,k+1)} X_{t,k} [i]^2  = A + B\cos({\agl{t}{k}})^2
\enspace,
\sum_{i\in S_{u_{t,k}}(t,k+1)} Y_{t,k} [i]^2 = C + D\cos({\aglp{t}{k}})^2\enspace.
$$
The following lemma states that $\abs{A-C}$ and $\abs{B-D}$ are bounded by the initial distance of two vectors. Its proof is the same as Lemma 4.4 in \cite{PS17}.
\begin{lemma}
\label{lem:never_drift}
	Fix $T_0<T$, and couple two chains $\st{X_t}_{T_0\leq t\leq T}, \st{Y_t}_{T_0\leq t\leq T}$ using the non-Markovian coupling defined in \defi{non_mark_cpl}.
	Fix $T_0\leq t_0 \leq T$ and $1\leq k_0 \leq m+1$.
	Then, on the event $\bigcap _{(t,k)\sqsubseteq (t_0,k_0)}\A(t,k)\cap\st{\PP_{T_0,1}=\st{[n]}}$, we have
	$$
	\norm{A_{t,k}-B_{t,k}}_{1,S}\leq \norm{A_{T_0,1}-B_{T_0,1}}_{1}
	$$
	for all $(t,k)\sqsubseteq(t_0,k_0)$
	and $S\in\PP_{t,k}$.
	Moreover, for all $(t,k)\sqsubseteq(t_0,k_0)$
	,
	$$
	\norm{A_{t,k}-B_{t,k}}_{1}\leq n\norm{A_{T_0,1}-B_{T_0,1}}_{1}\enspace.
	$$
\end{lemma}

Knowing that $A,B$ and $C,D$ are close,
the following lemma states the existence of a good distribution $\nu$ for ${\agl{t}{k}}$ and ${\aglp{t}{k}}$ such that
$\sum_{i\in S_{u_{t,k}}(t,k+1)} X_{t,k} [i]^2 
$ agrees with $\sum_{i\in S_{u_{t,k}}(t,k+1)} Y_{t,k} [i]^2$ with high probability.

\begin{lemma}[Lemma 4.6 in \cite{PS17}] \label{lem:good_dist}
	Fix positive reals $1<p<q'<q/2$. Let $\theta,\theta'\sim\mathrm{Unif}[0,2\pi)$ and let
	$$
	S=A+B\cos(\theta)^2 \quad\text{ and }
	\quad S'=C+D\cos(\theta')^2
	$$
	for some $0\leq A,B,C,D\leq 1$ that satisfy
	$$
	\abs{A-C},\abs{B-D}\leq n^{-q}\quad \text {and}\quad B,D\geq n^{-p}\enspace.
	$$
	Then for sufficiently large $n$, there exists a coupling of $\theta,\theta'$ so that
	$$
	\Pr\!\Br{S=S'}\geq 1-6\times 10^3n^{-c}
	$$
	and
	$$
	\cos(\theta)\cos(\theta')\geq 0 \quad\text{ and } \quad
	\sin(\theta)\sin(\theta')\geq 0
	$$
	where $c=\min\br{\frac{q'}{2},q-2q'}>0$.
\end{lemma}

\paragraph{Proof of the Total Variation Mixing Time}
\begin{proof}[Proof of \thm{mixing_time}]
Let $a=66,~b=24,~ T_0=500\log n, ~ T_1=15\log n,~ T=T_0+T_1 = 515\log n$. We construct a coupling of two copies $\st{X_t}_{t\geq 0}$ and $\st{Y_t}_{t\geq 0}$ of the parallel Kac's walk with starting points $X_0=x\in \SR^n$ and $Y_0\sim \mu$. The coupling is as follows:
\begin{enumerate}
	\item couple $\st{X_t}_{0\leq t\leq T_0}, \st{Y_t}_{0\leq t\leq T_0}$ by using the proportional coupling defined in \defi{prop_cpl},
	\item couple $\st{X_t}_{T_0\leq t\leq T}, \st{Y_t}_{T_0\leq t\leq T}$ by using the non-Markovian coupling defined in \defi{non_mark_cpl}.
\end{enumerate}
Define the events
\begin{align*}
	\EE_1 &= \st{\norm{A_{T_0}-B_{T_0}}_1\geq n^{-a}}\enspace,\\
	\EE_2 &= \st{\PP_{T_0,1}\neq \st{\st{1,\dots,n}}}\enspace,\\
	\EE_3 &= \st{X_{T}\neq Y_{T}}\enspace.
\end{align*}
By \lem{coupling_lemma},
{\small\begin{align}
	\sup_{X_0\in \SR^n} \norm{\mathcal{L}\!\br{X_{T}}-\mu}_{\mathrm{TV}}
	\leq \sup_{X_0\in \SR^n}\Pr\!\Br{\EE_3}\leq \sup_{X_0\in \SR^n}
	\br{\Pr\!\Br{\EE_1}+\Pr\!\Br{\EE_2}+\Pr\!\Br{\EE_3\cap \EE_1^c \cap \EE_2^c}}\enspace.\label{eq:tv}
\end{align}}
By Markov's inequality, we have
\begin{align}
	\Pr\!\Br{\EE_1} &= \Pr\!\Br{\norm{A_{T_0}-B_{T_0}}_1\geq n^{-a}}\nonumber\\
	&\leq \Pr\!\Br{\norm{A_{T_0}-B_{T_0}}_2\geq n^{-a-1/2}}\nonumber\\
	\text{(\lem{contraction_lemma})}&\leq n^{2a+1}\cdot2\cdot\br{\frac{3}{4}}^{T_0}\leq \frac{1}{n^2}\enspace. \label{eq:event1}
\end{align}
Moreover, by \lem{good_start}, we have
\begin{align}
	\Pr\!\Br{\EE_2}\leq 2n^{-2}\enspace.\label{eq:event2}
\end{align}
In order to bound $\Pr\!\Br{\EE_3\cap \EE_1^c \cap \EE_2^c}$,
recall the definition of $\A(t,k)$ in \defi{AT1}.
It is evident that $\bigcap_{(t,k)\sqsubseteq (T,1)}\A(t,k)$ implies $\EE_3^c$. Thus, $\EE_3\cap \EE_1^c \cap \EE_2^c$ implies $\bigcup_{(t,k)\sqsubseteq (T,1)}\A(t,k)^c \cap \EE_1^c \cap \EE_2^c$. So, we have
\begin{align}
\Pr\!\Br{\EE_3\cap \EE_1^c \cap \EE_2^c}
&\leq
\Pr\!\Br{\bigcup_{(t,k)\sqsubseteq (T,1)}\A(t,k)^c \cap \EE_1^c \cap \EE_2^c} \nonumber\\
&\leq
\sum_{t=T_0}^{T-1} \sum_{k=1}^{m} \Pr\!\Br{\A(t,k+1)^c \cap \br{\bigcap_{(t',k')\sqsubseteq (t,k)}\A(t',k')} \cap \EE_1^c \cap \EE_2^c}\nonumber\\
&\quad \quad + \Pr\!\Br{\A(T_0,1)^c \cap \EE_1^c \cap \EE_2^c}\enspace.\label{eq:e3}
\end{align}
Notice that if $\EE_2^c$ happens, we have
$$
\sum_{i\in [n]} X_{T_0,1} [i]^2 = \sum_{i\in [n]} Y_{T_0,1} [i]^2 = 1\enspace,\\
$$
and the proportional coupling forces $X_{T_0,1} [i]\cdot Y_{T_0,1} [i] = X_{T_0} [i]\cdot Y_{T_0} [i]\geq 0$ for all $1\leq i\leq n$.
Therefore, $A(T_0,1)$ must occur, i.e.,
\begin{align}
	\Pr\!\Br{\A(T_0,1)^c \cap \EE_1^c \cap \EE_2^c}=0\enspace.\label{eq:imps}
\end{align}
Combining \eq{e3} and \eq{imps}, we have
\begin{equation}
	\Pr\!\Br{\EE_3\cap \EE_1^c \cap \EE_2^c}\leq
	\sum_{t=T_0}^{T-1} \sum_{k=1}^{m} \Pr\!\Br{\A(t,k+1)^c \cap \br{\bigcap_{(t',k')\sqsubseteq (t,k)}\A(t',k')} \cap \EE_1^c \cap \EE_2^c}	\enspace.\label{eq:e3final}
\end{equation}

We are now left to find a upper bound for $$\Pr\!\Br{\A(t,k+1)^c \cap \br{\bigcap_{(t',k')\sqsubseteq (t,k)}\A(t',k')} \cap \EE_1^c \cap \EE_2^c}$$ when $T_0\leq t\leq T-1$ and $1\leq k\leq m$. To this end, we define
$$
\B(t,k) = \st{\min_{(t',k')\sqsubseteq (t,k):t'\leq t}\min_{1\leq i\leq n} Y_{t',k'}[i]^2\geq n^{-b}}\enspace.
$$
Note that
\begin{align}\label{eq:e3single}
&\Pr\!\Br{\A(t,k+1)^c \cap \br{\bigcap_{(t',k')\sqsubseteq (t,k)}\A(t',k')} \cap \EE_1^c \cap \EE_2^c}\nonumber\\
\leq &\Pr\!\Br{\A(t,k+1)^c \cap \br{\bigcap_{(t',k')\sqsubseteq (t,k)}\A(t',k')} \cap \B(t,k) \cap \EE_1^c \cap \EE_2^c} + \Pr\!\Br{\B(t,k)^c}\enspace.	
\end{align}
By \lem{marg_of_haar} and a union bound over all $(t',k')$ such that $(t',k')\sqsubseteq (t,k)$ and $t'\leq t$, we have for sufficiently large $n$,
\begin{equation}\label{eq:e3single2}
	\Pr\!\Br{\B(t,k)^c} \leq 15n^{3-\frac{b}{3}}\log(n)\enspace.
\end{equation}

Next, we consider two cases of the term
$$\Pr\!\Br{\A(t,k+1)^c \cap \br{\bigcap_{(t',k')\sqsubseteq (t,k)}\A(t',k')} \cap \B(t,k) \cap \EE_1^c \cap \EE_2^c}$$
in \eq{e3single}: $(t,k)\notin H$ and $(t,k)\in H$,
where $H$ is defined in \defi{non_mark_cpl}.
In the case that $(t,k)\notin H$, we have $\PP_{t,k} = \PP_{t,k+1}$ and we apply the proportional coupling. Thus
$\A(t,k)$ implies $\A(t,k+1)$ which means
\begin{align}\label{eq:e3single11}
	\Pr\!\Br{\A(t,k+1)^c \cap \br{\bigcap_{(t',k')\sqsubseteq (t,k)}\A(t',k')} \cap \B(t,k) \cap \EE_1^c \cap \EE_2^c}=0\enspace.
\end{align}

In the other case that $(t,k)\in H$,
let
$$
A=\sum_{i\in S_{u_{t,k}}(t,k+1)\backslash i^{(t)}_k} X_{t,k}[i]^2\enspace,
\quad \quad
B=X_{t,k} [i^{(t)}_k]^2+X_{t,k} [j^{(t)}_k]^2\enspace,
$$
$$
C=\sum_{i\in S_{u_{t,k}}(t,k+1)\backslash i^{(t)}_k} Y_{t,k}[i]^2
\enspace,\quad \quad
D=Y_{t,k}[i^{(t)}_k]^2+Y_{t,k}[j^{(t)}_k]^2\enspace,
$$
$$
S = A + B\cos({\agl{t}{k}})^2
\enspace,\quad \quad
S'= C + D\cos({\aglp{t}{k}})^2\enspace.
$$
On the event $\bigcap_{(t',k')\sqsubseteq (t,k)}\A(t',k') \cap \B(t,k) \cap \EE_1^c \cap \EE_2^c$, we have by \lem{never_drift}
\begin{align*}
	\abs{A-C} \leq \norm{A_{t,k}-B_{t,k}}_1\leq n\norm{A_{T_0}-B_{T_0}}_1 \leq n^{1-a}\enspace.
\end{align*}
Similarly,
\begin{align*}
	\abs{B-D} \leq \norm{A_{t,k}-B_{t,k}}_1\leq n\norm{A_{T_0}-B_{T_0}}_1 \leq n^{1-a}\enspace.
\end{align*}
Moreover, $D\geq n^{-b}$ and $B\geq D-\abs{B-D}\geq n^{-b}$ for sufficiently large $n$.
Then apply \lem{good_dist} with $p=b,~ q=a-1$ and $q'=\frac{2(a-1)}{5}$, we know there exists a distribution $\nu_0$ such that when $({\agl{t}{k}},{\aglp{t}{k}})\sim \nu_0$, we have
$$
	\cos({\agl{t}{k}})\cos({\aglp{t}{k}})\geq 0 \quad\text{ and }\quad
	\sin({\agl{t}{k}})\sin({\aglp{t}{k}})\geq 0
$$
$$
\Pr_{({\agl{t}{k}},{\aglp{t}{k}})\sim \nu_0}
\Br{ S\neq S' | \bigcap_{(t',k')\sqsubseteq (t,k)}\A(t',k') \cap \B(t,k) \cap \EE_1^c \cap \EE_2^c} \leq 6\times 10^3n^{-\frac{a-1}{5}}\enspace.
$$
We choose the best distribution which maximizes the probability of event described in \eq{cond1} and \eq{cond2}, so
\begin{align} \label{eq:e3single12}
	&\Pr\!\Br{\A(t,k+1)^c \cap \br{\bigcap_{(t',k')\sqsubseteq (t,k)}\A(t',k')} \cap \B(t,k) \cap \EE_1^c \cap \EE_2^c} \nonumber\\
	\leq  &\Pr\!\Br{\left.\A(t,k+1)^c \right|\bigcap_{(t',k')\sqsubseteq (t,k)}\A(t',k') \cap \B(t,k) \cap \EE_1^c \cap \EE_2^c} \nonumber\\
	\leq &\Pr_{({\agl{t}{k}},{\aglp{t}{k}})\sim \nu_0}\Br{\left.\A(t,k+1)^c \right|\bigcap_{(t',k')\sqsubseteq (t,k)}\A(t',k') \cap \B(t,k) \cap \EE_1^c \cap \EE_2^c} \nonumber \\
	= &\Pr_{({\agl{t}{k}},{\aglp{t}{k}})\sim \nu_0}\Br{\left.S\neq S' \right|\bigcap_{(t',k')\sqsubseteq (t,k)}\A(t',k') \cap \B(t,k) \cap \EE_1^c \cap \EE_2^c} \leq 6\times 10^3n^{-\frac{a-1}{5}}\enspace.
\end{align}

Combining \eq{e3final}, \eq{e3single}, \eq{e3single2}, \eq{e3single11} and \eq{e3single12}, we have for sufficiently large $n$
\begin{align} \label{eq:event3}
	\Pr\!\Br{\EE_3\cap \EE_1^c \cap \EE_2^c}\leq \sum_{t=T_0}^{T-1} \sum_{k=1}^{m} 6\times 10^3n^{-\frac{a-1}{5}} + 15n^{3-\frac{b}{3}}\log(n) \leq   \frac{1}{n^2}\enspace.
\end{align}

By \eq{tv}, \eq{event1}, \eq{event2} and \eq{event3}, we have for sufficiently large $n$ and $T=515\log n$
\begin{align*}
	\sup_{X_0\in \SR^n} \norm{\mathcal{L}\br{X_{T}}-\mu}_{\mathrm{TV}} \leq \sup_{X_0\in \SR^n}\Pr\!\Br{\EE_3} \leq  \frac{1}{2n} \enspace.
\end{align*}

As for $T= c \log n$ where $c > 515$,
by \lem{fcts_tv}
we have
\begin{align*}
	\sup_{X_0\in \SR^n} \norm{\mathcal{L}\br{X_{T}}-\mu}_{\mathrm{TV}}\leq 2\br{\frac{1}{n}}^{\lint{\frac{T}{515\log n}}} \leq  \frac{1}{2^{\br{c/515-1}\log n - 1}}\enspace.	
\end{align*}

\end{proof}

\subsubsection{Complex Case}
The output distribution of
the parallel Kac's walk on complex vectors
after $T$ steps
approaches the Haar measure
on the complex unit sphere of $\C^n$
exponentially fast as $T$ grows.
We restate the theorem here.
\begin{reptheorem}{thm:mixing_time_c}
	Let $\st{X_t\in\SC^{n}}_{t\geq 0}$
	be a Markov chain that evolves
	according to the parallel Kac's walk
	on complex vectors.
	Then, for sufficiently large $n$, $c>515$ and $T = c\log n$,
	$$
	\sup_{X_0\in \SC^n}\! \norm{\mathcal{L}\!\br{X_{T}}-\mu}_{\mathrm{TV}} \leq  \frac{1}{2^{\br{c/515-1}\log n-1}}\enspace,
	$$
	where $\mu_{\C}$ is the Haar measure on $\SC^{n}$.
\end{reptheorem}

We defer the complete proof to \app{deferred proofs},
as it is similar to the proof of \thm{mixing_time}.

\subsection{Construction of \srss}
\label{sec:provedrss}

\subsubsection{Real Case}
The following theorem proves that the \rss~we construct over real space is also a \srss.

\begin{theorem}\label{thm:kactosrss}
	Let $n\in\N$, $d = \log^2\!\secpar+\log^2\!n$ and $T = 515 (\secpar + 1) n$.
	The ensemble of unitary operators
	$\enssrsscons$ defined in \defi{kactorss}
	is a \srss.
\end{theorem}

To prove \thm{kactosrss}, recall the ensemble of unitary operators $\enscsrsscons \coloneq \st{\csrsscons }_{\secpar}$ we define in \sect{rss_in_real}.
We have the following proposition for $\enscsrsscons$.

\begin{proposition}
	For $T = 515 (\secpar + 1) n $,
	the ensemble of unitary operator $\enscsrsscons$
	is a \csrss.
\end{proposition}

\begin{proof}
	Note that a uniformly random
	$\csrsscons_{\rep{\sigma}{T},\rep{\cf}{T}}$
	corresponds to a $T$-step parallel Kac's walk on $\SR^{2^n}$.
	The proposition then follows from \thm{mixing_time} and the definition of the \csrss.
\end{proof}

Let $\ket{\eta}\in\S(\H)$ be an arbitrary real state.
Set
\begin{equation}\label{eq:NN}
	\NN = \st{\srsscons_{\rep{\sigma}{T},\rep{f}{T}}\!\!\ket{\eta}}
	\enspace\enspace \text{and}\enspace\enspace
	\widetilde{\NN} = \st{\csrsscons_{\rep{\sigma}{T},\rep{\cf}{T}}\!\!\ket{\eta}}\enspace.
	\end{equation}
We need the following two lemmas. Both proofs are deferred to \app{deferred proofs}.

\begin{lemma}\label{lem:realNNfullmap}
$\widetilde{\NN}= \SR^{2^n}$.
\end{lemma}

\begin{lemma}\label{lem:realNNepsnet}
There exists an $\epsilon  = \negl{\secpar}$  such that $\NN$
	is an $\epsilon$-net for real vectors in $\S(\H)$, where $\NN$ is defined in \cref{eq:NN}.
\end{lemma}

\begin{proof}[Proof of \thm{kactosrss}]	
It is easy to see that the uniformity condition is satisfied.
	Let $\kappa$ denote the key length.
	Quantum circuit $\srsscons$ applies $\srsscons_{\rep{\sigma}{T},\rep{f}{T}}$ after reading $\rep{\sigma}{T}$ and $\rep{f}{T}$.
	To implement $\srsscons_{\rep{\sigma}{T},\rep{f}{T}}$, we need to realize each of
	the $T = 515 (\secpar + 1) n$ unitary gates $K$.
	Since each gate $K$ can be implemented in $\poly{n,\secpar,\klen}$ time,
	the total construction time for $\srsscons_{\rep{\sigma}{T},\rep{f}{T}}$
	is also $\poly{n,\secpar,\klen}$.

Combining with \lem{realNNepsnet}, it suffices to prove that there exists a good distribution $\widetilde{\nu}$ satisfying the requirement in \defi{trssb}.
	Fix $\ket{\eta}\in\S(\H)$. Define three distributions:
	\begin{itemize}
		\item $\nu$ be the distribution of
			$\srsscons_{\rep{\sigma}{T}, \rep{f}{T}}\!\!\ket{\eta}$
			with independent and uniformly random permutations $\rep{\sigma}{T}\subseteq S_{2^n}$
			and random functions $\rep{f}{T}$$\subseteq \{f: \bit{n-1}\to\bit{d}\}$.
		\item $\widetilde{\nu}$ be the distribution of
			$\csrsscons_{\rep{\sigma}{T},\rep{\cf}{T}}\!\!\ket{\eta}$
			with independent and uniformly random permutations $\rep{\sigma}{T}\subseteq S_{2^n}$,
			and random functions $\rep{\cf}{T} \subseteq\{f: \bit{n-1}\to[0,1)\}$.
		\item $\mu$ be the Haar measure on $\SR^{2^n}$.
	\end{itemize}
	Note that $\widetilde{\nu}$ is the output distribution of $T$-step parallel Kac's walk. Thus by \thm{mixing_time}, we have
	\begin{align}\label{eq:closetv}
		\tv{\widetilde{\nu} - \mu} \leq \frac{1}{2^{\secpar n - 1}} = \negl{\secpar}\enspace.
	\end{align}
	We are left to show the Wasserstein $\infty$-distance between $\nu$ and $\widetilde{\nu}$ is negligible.
	To this end, we construct a coupling $\gamma_0$ of $\nu$ and $\widetilde{\nu}$ by using the same permutation $\sigma_t$ and letting $f_t$ be the function satisfying $f_t(y)$ is the $d$ digits after the binary point in $\widetilde{f}_t(y)$ for all $y\in\bit{n-1}$. Therefore
	\begin{align}\label{eq:closewinfty}
		&W_{\infty} ( \nu, \widetilde{\nu})\nonumber \\
		=
		~& \lim_{p\to\infty}\br{ \inf_{\gamma\in\Gamma(\nu,\widetilde{\nu})} \E_{(\ket{v},\ket{u})\sim \gamma}\!\Br{\norm{\ket{v}-\ket{u}}_2^p} }^{1/p} \nonumber\\
		\leq ~& \lim_{p\to\infty}\br{ \E_{(\ket{v},\ket{u})\sim \gamma_0}\!\Br{\norm{\ket{v}-\ket{u}}_2^p} }^{1/p}
		~\overset{\text{(Eq.~\eq{uvclose})}}{\leq} ~
		\frac{1030\pi (\secpar+1)n}{\secpar^{\log \secpar}n^{\log n}}= \negl{\secpar} \enspace .
	\end{align}
	This completes the proof.
\end{proof}

\subsubsection{Complex Case}
Likewise, the ensemble of unitary operators built over complex space in \sect{rss_in_complex} turns out to be a \srss. The proof is provided in \app{deferred proofs}.

\begin{theorem}\label{thm:kactosrssc}
	Let $n\in\N$, $d = 2\br{\log^2\!\secpar + \log^2\!n}$ and $T = 515 (\secpar + 1) n$.
	The ensemble of unitary operators
	$ \enssrssconsc$
	defined in \defi{kactosrssc}
	is a \srss.
\end{theorem}

\newpage
\section{Connections with Existing PRS variants}
\label{sec:connection}
In this section, we formally demonstrate that the existing $\prs$ and
its variants can all be constructed from $\prss$ in a black-box
manner.
\paragraph{(Scalable) Pseudorandom States.}
Originally, the definition of $\prs$ in~\cite{JLS18} identifies the
number of qubits $\dimin$ as the security parameter. Consequently,
the security of $\prs$ is not guaranteed when $n$ is small. This issue
was addressed by~\cite{BS20} through defining 
\emph{scalable $\prsg$s}, in which $n$ and $\secpar$ are treated
separately. This allows for the tuning of security when $\dimin <
\secpar$. We rephrase the scalable definition in a style that is
congruent with $\prss$.

\begin{definition}[(Scalable) \prsg]
  Let $\K=\{0,1\}^{\klen}$ be a key space,
  $\H$ be a Hilbert space of dimension $2^{\dimin}$ with
  $\dimin\in\N$,
  $\secpar$ be a security parameter.  
  A \emph{(scalable) pseudorandom state generator (\prsg)}~is an ensemble of unitaries
   \[\mathcal{G}^n:=\{\{\mathcal{G}^{n,\secpar}_{k}:\H\rightarrow \H\}_{k\in\K}\}_{\lambda}\]
   satisfying

   \begin{itemize}
   \item \textbf{Pseudorandomness.}  Any polynomially many $\ncopy$ copies of
     $\ket{\phi_k}$ with the same random $k$ is computationally
     indistinguishable from the same number of copies of a Haar random
     state.  More precisely, for any $n\in\N$, any efficient quantum
     algorithm $\A$ and any $\ncopy \in\poly{\lambda}$,
    \[\abs{\prob{k\gets \K}{\A\!\br{\ket{\phi_k}^{\otimes \ncopy}}=1}-\prob{\ket{\psi}\gets \mu}{\A\!\br{\ket{\psi}^{\otimes \ncopy}}=1}}=\negl{\lambda} \enspace,\]
    where $\ket{\phi_k}:=\mathcal{G}^{n,\lambda}_{\kappa}\!\ket{0^n}$ and $\mu$ is the Haar measure on $\S(\H)$.

  \item \textbf{Uniformity.}  $\mathcal{G}^n$ can be uniformly
    computed in polynomial time. That is, there is a deterministic
    Turing machine that, on input $(1^{\dimin},1^\secpar,1^\klen)$,
    outputs a quantum circuit $Q$ in $\poly{\dimin,\secpar,\klen}$
    time such that for all $k\in \K$ and $\ket{\phi}\in\Sin$
    \[ Q\ket{k}\!\ket{\phi} = \ket{k}\!\ket{\phi_k}\, , \] where
    $\ket{\phi_k} := \mathcal{G}^{n,\lambda}_{k}\! \ket{\phi}$.

  \item \textbf{Polynomially-bounded key length.}
    $\klen = \log |\K| = \poly{\dimout,\secpar}$. As a result,
    $\mathcal{G}^{\dimin}$ can be computed efficiently in time
    $\poly{\dimin,\secpar}$.
\end{itemize}

We informally call the keyed family of quantum states
$\st{\ket{\phi_k}}_{k\in\K}$ a \emph{(scalable) pseudorandom quantum state} in $\H$.
\label{def:scalableprsg}
\end{definition}

The existence of \prss s implies
the existence of (scalable) \prsg s~\cite{BS20,JLS18} straightforwardly by definition.

\begin{lemma}
  If $\SG{}{n,m}$ is a \prss~from $\hspacein$ to $\hspaceout$
  over a key space $\K$ with security parameter $\secpar$,
  $\st{ \SG{k}{n,m,\secpar}\ket{\phi}}_{k\in\K}$ is a (scalable) \prs~in $\hspaceout$ for any
  $\ket{\phi}\in\Sin$.
  \label{lem:prss2prs}
\end{lemma}

\paragraph{Pseudorandom Function-like States.}
We recall the definition of the pseudorandom function-like states generator 
with three levels of security:
\begin{definition}[\prfsg~\cite{AGQY22,AQY22}]
  \label{def:prfs}
  Let $\K=\bit{\secpar}$ be a key space.  Let $\CC=\bit{n}$ be a
  classical input space and $\H$ be a Hilbert space of dimension
  $2^m$.  A pair of $\poly{\secpar,m}$-time quantum algorithm $(K,G)$
  is a \emph{pseudorandom function-like state} generator if the
  following holds:
  \begin{itemize}
  \item \textbf{Key Generation.}  For all $\secpar\in\N$,
    $K(1^\secpar)$ outputs a uniform key $k\in \K$.

  \item \textbf{State Generation.} 
    For all $k\in \K$ and $x\in \CC$, $G(1^\secpar, k, x)$ computes
    $\ket{\phi_{x,k}} \in\S(\H)$.

  \item \textbf{Pseudorandomness.}  The pseudorandomness can be
    defined in three different levels (from weaker to stronger):

    \begin{itemize}
    \item \textbf{Selective security.}  For any $\secpar\in \N$,
      $s\in \poly{\secpar}, \ncopy\in \poly{\secpar }$, any efficient
      quantum algorithm (non-uniform) $\A$ and a set of pre-declared
      input $\{ x_1, \dots, x_{s}\}\subseteq \CC$,
			\begin{align*}
				\bigg|\Pr_{k\gets\K}
				& \Br{\A(x_1, \dots, x_s, \ket{\phi_1}^{\tp\ncopy}, \dots, \ket{\phi_s}^{\tp\ncopy})=1} \\
        		& -\Pr_{\ket{\psi_1}, \dots, \ket{\psi_s}\gets\haar}
        		  \Br{\A(x_1, \dots, x_s, \ket{\psi_1}^{\tp\ncopy}, \dots, \ket{\psi_s}^{\tp\ncopy})=1}
        		\bigg|\leq \negl{\secpar},
			\end{align*}
		where, for each $i$, $\ket{\phi_i}$ denotes $\ket{\phi_{x_i,k}}$ generated by $G$; and $\ket{\psi_i}$ is a Haar random state.
			
		\item \textbf{Adaptive security.}
			Given \emph{classical-access} to either the \prfs~oracle $\oracle_{\prfs}$
			or the Haar-random oracle $\oracle_{\Haar}$,
			for any $\secpar\in\N$,
			any efficient quantum algorithm (non-uniform) $\A$
			with polynomial length quantum advice $\rho_\secpar$,
			\[
				\abs{
					\prob{k\gets\K}{\A^{{\oracle_{\prfs}(k, \cdot)}}(\rho_\secpar)=1}
					-
					\prob{\oracle_{\Haar}}{\A^{{\oracle_{\Haar}(\cdot)}}(\rho_\secpar)=1}
				}\leq\negl{\secpar},
			\]
			where on input $x\in\CC$,
			$\oracle_{\prfs}(k, \cdot)$ outputs $G(1^\secpar, k, x)$;
			and $\oracle_{\Haar}(\cdot)$ outputs a Haar random $\ket{\psi_x}$.
		\item \textbf{Quantum-accessible adaptive security.}
			Given \emph{quantum-access} to a \prfs~or Haar-random oracle,
			for any $\lambda\in \N$,
			for any efficient quantum algorithm (non-uniform) $\A$
			with polynomial length quantum advice $\rho_\secpar$,
			\[
				\abs{
					\prob{k\leftarrow \K}{\A^{\ket{\oracle_{\prfs}(k, \cdot)}}(\rho_\secpar)=1}
					-
					\prob{\oracle_{\Haar}}{\A^{\ket{\oracle_{\Haar}(\cdot)}}(\rho_\secpar)=1}
				}\leq\negl{\secpar},
			\]
			where on input a $n$-qubit register $\regx$,
			$\oracle_{\prfs}(k, \cdot)$ applies a channel that controlled on the
			register $\regx$ containing $x$,
			and stores $G(1^\secpar, k, x)$ in the register $\regy$,
			then output $(\regx, \regy)$;
			instead,  $\oracle_{\Haar}(\cdot)$ stores a Haar random
			$\ketbra{\psi_x}$ on the register $\regy$,
			then output $(\regx, \regy)$.
		\end{itemize}
	\end{itemize}
	\label{def:prfs}
\end{definition}

As previously mentioned, the security of our \prss~is maintained
when applied to any arbitrary initial pure state,
making it straightforward to derive a \prfsg~with selective security. 

\begin{lemma}
	Assume $\SG{}{n,m}$ is a \prss~from $\hspacein$ to $\hspaceout$
	over a key space $\K$
	with security parameter $\secpar$ s.t. $n=O(\log\secpar)$.
	Construct $(\hat{K},\hat{G})$ in the following way:
	\begin{itemize}
		\item (Key generation.)
			For all $\lambda\in\N$, $\hat{K}(1^\secpar)$  sets a large enough $s\in \poly{\secpar}$ 
			and generates a key 
			$k=\{k_1, \dots, k_s\}$ such that for $i\in[s]$, 
			$k_i$ is chosen uniformly and independently from $\K$;
		\item (State generation.)
			For all $k$ and classical queries $\{x_i\in \CC\}_{i=1}^s$,
			$\hat{G}(1^\lambda,k,x_i)$
			computes $\ket{\phi_{i}}= \SG{k_i}{n,m,\lambda}\ket{x_i}$.
		\end{itemize}
	Then, $(\hat{K},\hat{G})$ is a \prfsg~with selective security.
\label{lem:prss2prfs-select}
\end{lemma}
\begin{proof}
	The algorithm $(\hat{K},\hat{G})$ is efficient as 
	it essentially performs \prss~a polynomial number of times.
	Meanwhile, the indistinguishability holds because all output states are obtained via independent keys.
\end{proof}

When we consider log-inputs by restricting $n=O(\log \secpar)$ and setting $s=O(2^n)\in\poly{\secpar}$,
the construction in \lem{prss2prfs-select} produces sufficient key segments
to ensure that every $x\in\CC$ has its own independent key, without assuming the number of queries from the adversary.
As a consequence, log-input \prfsg s with adaptive security can be achieved
through the use of \prss s.
Furthermore, as demonstrated by the results in~\cite{Zhandry21_qprf},
quantum superposition queries over the input domain do not provide
any additional advantages when the output state is known for every input string,
which results in quantum-accessible adaptive security.

\begin{lemma}\label{lem:prfs-adapt}
	If $\SG{}{n,m}$ is a \prss~from $\hspacein$ to $\hspaceout$
	over a key space $\K$
	with security parameter $\secpar$ s.t. $n=O(\log\secpar)$,
	then $(\hat{K},\hat{G})$ is a \prfsg~satisfying
	adaptive security and quantum-accessible adaptive security.
\end{lemma}

Prior works~\cite{AGQY22,AQY22} have demonstrated that $\log$-input
$\prfsg$s can be constructed from a $\prs$. However, this approach
requires a \emph{test procedure} involving a post-selection among the
classical input domain, which introduces errors and incurs computation
overhead. Our $\prss$, with its flexibility in choosing initial states,
provides a cleaner method. It is worth noting that $\prfsg$s on long inputs
(i.e., exponentially large domain) may appear stronger and might not be
achievable from either $\prs$s or $\prss$s in a black-box manner.

\newpage
\section{Deferred Proofs}
\label{sec:deferred proofs}
\subsection{Proof of \lem{h3h4}}
\begin{proof}[Proof of \lem{h3h4}]
We demonstrate that the real case and the complex case can be established using the same approach.

Let $\zeta^+-\zeta^-$ be the Hahn decomposition of the signed measure $\mu-\nu$. Then
\begin{align*}
&\norm{
		\E_{\ket{\psi}\sim \mu}\!\Br{ \br{\ketbra{\psi}}^{\otimes l} }
		- \E_{\ket{\varphi}\sim \nu}\!\Br{\br{\ketbra{\varphi}}^{\otimes l}}
		}_1\\
		=~&\norm{\int_{\SR^{2^n}}\br{\ketbra{u}}^{\otimes l}\br{\mu-\nu}\br{\d\ket{u}}}_1\\
		=~&\norm{\int_{\SR^{2^n}}\br{\ketbra{u}}^{\otimes l}\br{\zeta^+-\zeta^-}\br{\d\ket{u}}}_1\\
		\leq~&\norm{\int_{\SR^{2^n}}\br{\ketbra{u}}^{\otimes l}\zeta^+\br{\d\ket{u}}}_1+\norm{\int_{\SR^{2^n}}\br{\ketbra{u}}^{\otimes l}\zeta^-\br{\d\ket{u}}}_1\\
		\leq~&\int_{\SR^{2^n}}\norm{\br{\ketbra{u}}^{\otimes l}}_1\zeta^+\br{\d\ket{u}}+\int_{\SR^{2^n}}\norm{\br{\ketbra{u}}^{\otimes l}}_1\zeta^-\br{\d\ket{u}}\\
		=~&\zeta^-\br{\SR^{2^n}}+\zeta^-\br{\SR^{2^n}}\\
		=~&\norm{\mu-\nu}_{\mathrm{TV}}\enspace.
\end{align*}
\end{proof}

\subsection{Proof of \lem{twokeys}}
\begin{proof}[Proof of \lem{twokeys}] We first prove the case that $F: \K_1 \times \X_1 \to \Y_1$ and $G: \K_2 \times \X_2 \to \Y_2$ are both \qprf s. We prove the lemma by contradiction.
	Suppose there exists a polynomial-time quantum oracle algorithm $\A$ who queries $q$ times such that
	\begin{align*}
		\abs{
			\prob{ k_1\gets\K_1, k_2\gets \K_1 }{\A^{F_{k_1}, G_{k_2}}\!\br{1^\lambda}=1}-\prob{f\gets\Y_1^{\X_1},g\gets\Y_2^{\X_2}}{\A^{f,g}\!\br{1^\lambda}=1}
		} = \dfrac{1}{\poly{\lambda}} \enspace .
	\end{align*}
	By the triangle inequality,
	{\small
	\begin{align*}
		&\abs{
			\prob{ k_1\gets\K_1, k_2\gets \K_1 }{\A^{F_{k_1}, G_{k_2}}\!\br{1^\lambda}=1}-\prob{k_1\gets\K_1, g\gets\Y_2^{\X_2}}{\A^{F_{k_1},g}\!\br{1^\lambda}=1}
		} \\
		&~~~~~~+~~~~~ \abs{
		\prob{k_1\gets\K_1, g\gets\Y_2^{\X_2}}{\A^{F_{k_1},g}\!\br{1^\lambda}=1}-\prob{f\gets\Y_1^{\X_1},g\gets\Y_2^{\X_2}}{\A^{f,g}\!\br{1^\lambda}=1}
		} \geq  \dfrac{1}{\poly{\lambda}} \enspace .
	\end{align*}}
	Without loss of generality, we assume
	\begin{align*}
		\abs{
		\prob{k_1\gets\K_1, g\gets\Y_2^{\X_2}}{\A^{F_{k_1},g}\!\br{1^\lambda}=1}-\prob{f\gets\Y_1^{\X_1},g\gets\Y_2^{\X_2}}{\A^{f,g}\!\br{1^\lambda}=1}
		} \geq  \dfrac{1}{2\cdot \poly{\lambda}}\enspace.
	\end{align*}
	
	Thus we can define a polynomial-time quantum oracle algorithm
        $\A'$ that is able to distinguish $F_{k_1}$ from a random
        function $f$.  Once $\A'$ gets an oracle access to some
        function $h\in\st{F_{k_1},f}$, it simulates the execution of
        $\A$ with oracle access to $h$ and a random function $g$.
        Since $\A$ makes at most $q$ queries, $\A'$ can efficiently
        simulate a random oracle using $2q$-wise independent function
        (see~\cite[Theorem 6.1]{Zhandry15_ibe}).  This contradicts the
        quantum-security property of $F$.

To prove the case that $F: \K_1 \times \X_1 \to \Y_1$ is a \qprf~and $G: \K_2 \times \X_2 \to \X_2$ is a \qprp, we may assume as above that there exists a polynomial-time quantum oracle algorithm $\A$ who queries $q$ times such that
\begin{align*}
		\abs{
		\prob{k_1\gets\K_1, g\gets S_{\X_2}}{\A^{F_{k_1},g}\!\br{1^\lambda}=1}-\prob{f\gets \Y_1^{\X_1},g\gets S_{\X_2}}{\A^{f,g}\!\br{1^\lambda}=1}
		} \geq  \dfrac{1}{2\cdot \poly{\lambda}}\enspace.
\end{align*}
Then we define the following efficient algorithm $\A''$ to distinguish a \qprf~from a random function: 
\begin{enumerate}
	\item Given $h\in\st{F_{k_1},f}$,
		 it chooses a permutation $g$ uniformly at random from a $2q$-wise almost independent family of permutations to simulate a random permutation oracle.
		 This sampling procedure can be done in polynomial time
		 (see~\cite[Theorem 5.9]{KNR09}).
	\item It simulates $\A$ with oracle access to $h$ and $g$, and outputs what $\A$ returns.
\end{enumerate}
This breaks the quantum-security property of \qprf s.

\end{proof}

\subsection{Proof of ~\thm{mixing_time_w1_c}}
To prove \thm{mixing_time_w1_c}, we extend the the proportional coupling introduced in real case to complex case.
In the proportional coupling, the real case lets $(X_t[i], X_t[j])$ be collinear with $(Y_t[i], Y_t[j])$ so that the distance from $X_t$ to $Y_t$ is reduced by a constant factor in each step. However in complex case, $(X_t[i], X_t[j])$ is actually a two  dimensional complex vector and this makes it unsuitable to adopt the previous approach directly. To deal with this, in the complex case, we let
		$(\abs{X_t[i]}, \abs{X_t[j]})$ align collinearly with $(\abs{Y_t[i]}, \abs{Y_t[j]})$ in real two dimensional real plane, and make ${X_t[i]}$ and ${Y_t[i]}$, as well as ${X_t[j]}$ and ${Y_t[j]}$, share the same argument in complex plane in the meantime. We assume $n=2m$.
		
\subsubsection{Proportional Coupling}

\begin{definition}[Proportional Coupling]\label{def:prop_cpl_c}
We define a coupling of two copies
$\st{X_t}_{t\geq 0}, \st{Y_t}_{t\geq 0}$
of parallel Kac's walk on complex vectors
in the following way:
Fix $X_t$, $Y_t\in \C^{n}$.
\begin{enumerate}
	\item Choose a perfect matching of $[n]$, denoted by $P_t = \st{\br{i^{(t)}_1, j^{(t)}_1},\dots,\br{i^{(t)}_m, j^{(t)}_m}}$, uniformly at random.
	\item Let $X_{t,1} =X_t$ and $Y_{t,1}=Y_t$. For every $1\leq k\leq m$:
		\begin{enumerate}
		\item let $l^{(t)}_k = \sqrt{\abs{X_{t,k}[i^{(t)}_k]}^2+\abs{X_{t,k}[j^{(t)}_k]}^2}$ and ${l'}^{(t)}_k = \sqrt{\abs{Y_{t,k}[i^{(t)}_k]}^2+\abs{Y_{t,k}[j^{(t)}_k]}^2}$. Let $U_0$ and $U_0'$ be the unitary operators in $\su\br{2}$ which satisfy
		\begin{align*}
			U_0\br{\begin{matrix}
			X_{t,k}[i^{(t)}_k]\\
			X_{t,k}[j^{(t)}_k]
			\end{matrix}} = \br{\begin{matrix}
			l^{(t)}_k\\
			0
			\end{matrix}}
			\text{\quad and\quad}
			U_0'\br{\begin{matrix}
			Y_{t,k}[i^{(t)}_k]\\
			Y_{t,k}[j^{(t)}_k]
			\end{matrix}} = \br{\begin{matrix}
			{l'}^{(t)}_k\\
			0
			\end{matrix}}\enspace.
		\end{align*}
		\item pick $\alpha_k^{(t)},\beta_k^{(t)}\in[0,2\pi)$ and $\zeta_k^{(t)}\in[0,1)$ uniformly at random and set

\begin{align}\label{eq:theta}
\theta_k^{(t)} = \arcsin \sqrt{\zeta_k^{(t)}}\enspace.
\end{align}
		\item set
		\begin{align*}
			X_{t,k+1} = G_{\C}\!\br{i^{(t)}_k,j^{(t)}_k, \alpha^{(t)}_k,\beta^{(t)}_k,\theta^{(t)}_k, U_0X_{t,k}}\enspace,\\
			Y_{t,k+1} = G_{\C}\!\br{i^{(t)}_k,j^{(t)}_k, \alpha^{(t)}_k,\beta^{(t)}_k,\theta^{(t)}_k, U_0'Y_{t,k}}\enspace.
		\end{align*}
	\end{enumerate}
	
	\item Finally, set $X_{t+1}=X_{t,m+1}$ and $Y_{t+1}=Y_{t,m+1}$.
\end{enumerate}
\end{definition}

\begin{remark}
	In step 2(a), if $l^{(t)}_k\neq 0$, $X_{t,k}[i^{(t)}_k]=r_1e^{\i\gamma}$ and $X_{t,k}[j^{(t)}_k]=r_2e^{\i\delta}$ with $r_1,r_2\in[0,1]$ and $\gamma,\delta\in[0,2\pi)$, then $U_0$ is $U\!\br{\alpha_0,\beta_0,\theta_0}$ where
	 \begin{align*}
	 	\alpha_0 = -\gamma\enspace,\quad \beta_0 = \pi-\delta\enspace,\quad \theta_0 = \arccos\frac{r_1}{\sqrt{r_1^2+r_2^2}} \enspace .
	 \end{align*}
	 If $l^{(t)}_k = 0$, $U_0$ can be arbitrary matrix in $\su(2)$.
	 The same applies to $U_0'$.
\end{remark}

Due to the unitary invariance of Haar measures,
	 $UU_0$ is Haar distributed on $\su(2)$
	 for a random random $U$
	 sampled according to
	 Haar measure on $\su(2)$.
	 This guarantees that
	 the proportional coupling
	 is indeed a valid coupling.
	
\begin{remark}
	The proportional coupling forces $X_{t+1}[i], Y_{t+1}[i]$ and the original point to be collinear in the complex plane. In other words, $X_{t+1}[i]$ and $Y_{t+1}[i]$ have the same argument.  Specifically, we can write
	$$
		X_{t+1}[i] = e^{i\phi} l
		\text{\quad and \quad}
		Y_{t+1}[i] =  e^{i\phi}l'
	$$
	for some $l,l'\in[0,1],\phi\in[0,2\pi)$.
\end{remark}

Through proportional coupling ,
$X$ and $Y$ approach each other at an exponential rate.
Formally, we have

\begin{lemma} \label{lem:contraction_lemma_c}
	Let $X_0,Y_0\in \SC^{n}$. For $t\geq 0$, we couple $(X_{t+1},Y_{t+1})$ conditional on $(X_{t},Y_{t})$ according to the proportional coupling defined in \defi{prop_cpl_c}.
	We define
	$$
	A_{t}[i] = \abs{X_t[i]}^2\enspace,\quad B_{t}[i] = \abs{Y_t[i]}^2\enspace.
	$$
	Then for any $l\in\N$, we have
	$$
	\E\!\Br{\sum_{i=1}^n \br{A_l[i]-B_l[i]}^2}\leq 2\cdot\br{\frac{2}{3}}^l\enspace.
	$$
\end{lemma}

\begin{proof}[Proof of~\lem{contraction_lemma_c}]
 Fix $X_t,Y_t\in \SC^n$. Let $(X_{t+1},Y_{t+1})$ obtained from $(X_{t},Y_{t})$ by applying the coupling defined in \defi{prop_cpl_c}.
	Recall that $n=2m$.
	Let $N=\frac{n!}{2^{m}m!}$ be the number of perfect matchings for $[n]$.
	A perfect matching $\st{\br{i^{(t)}_1, j^{(t)}_1},\dots,\br{i^{(t)}_m, j^{(t)}_m}}$ of $[n]$ at step $t$ is denoted by $\br{\overrightarrow{i^{(t)}},\overrightarrow{j^{(t)}}}$.

We have
	\begin{align}
		&\E\!\Br{\sum_{i=1}^n \br{A_{t+1}[i]-B_{t+1}[i]}^2}\nonumber\\
		=&\frac{1}{N}\sum_{\br{\overrightarrow{i^{(t)}},\overrightarrow{j^{(t)}}}}\underbrace{\E\!\Br{\left.\sum_{i=1}^n \br{A_{t+1}[i]-B_{t+1}[i]}^2\right|P_t=\br{\overrightarrow{i^{(t)}},\overrightarrow{j^{(t)}}}}}_{(\star)} \enspace .\label{eqn:t1c}
	\end{align}

By the definition of the parallel Kac's walk on complex vectors, we have

\begin{align}
		(\star)=&\sum_{k=1}^m \E\!
		\Br{\br{\br{A_{t}[i^{(t)}_k]+A_{t}[j^{(t)}_k]}\cos(\theta_k^{(t)})^2-\br{B_{t}[i^{(t)}_k]+B_{t}[j^{(t)}_k]}\cos(\theta_k^{(t)})^2}^2}\nonumber\\
		&\quad\quad+\sum_{k=1}^m \E\!
		\Br{\br{\br{A_{t}[i^{(t)}_k]+A_{t}[j^{(t)}_k]}\sin(\theta_k^{(t)})^2-\br{B_{t}[i^{(t)}_k]+B_{t}[j^{(t)}_k]}\sin(\theta_k^{(t)})^2}^2}\nonumber\\
		=&\frac{2}{3}\sum_{k=1}^m
		\br{\br{A_{t}[i^{(t)}_k]+A_{t}[j^{(t)}_k]}-\br{B_{t}[i^{(t)}_k]+B_{t}[j^{(t)}_k]}}^2\nonumber\\
		=&\underbrace{\frac{2}{3}\sum_{k=1}^m
		\br{\br{A_{t}[i^{(t)}_k]-B_{t}[i^{(t)}_k]}^2+\br{A_{t}[j^{(t)}_k]-B_{t}[j^{(t)}_k]}^2}}_{(\star\star)} \nonumber\\	&\quad\quad+\underbrace{\frac{2}{3}\sum_{k=1}^m2\br{A_{t}[i^{(t)}_k]-B_{t}[i^{(t)}_k]}\br{A_{t}[j^{(t)}_k]-B_{t}[j^{(t)}_k]}}_{(\star\star\star)}\enspace,\label{eqn:t2c}
\end{align}
where the second equality is by $\E\!\Br{\cos(\theta_k^{(t)})^4}=\E\!\Br{\sin(\theta_k^{(t)})^4}=1/3.$

As $\st{\br{i^{(t)}_1, j^{(t)}_1},\dots,\br{i^{(t)}_m, j^{(t)}_m}}$ is a perfect matching, we have
\begin{equation}\label{eqn:t3c}
  (\star\star)=\frac{2}{3}\sum_{i=1}^{n}\br{A_t[i]-B_t[i]}^2 \enspace .
\end{equation}
Combining Eqs.~\eqref{eqn:t1c}\eqref{eqn:t2c}\eqref{eqn:t3c}, we obtain
\begin{align}\label{eqn:t4c}
  &\E\!\Br{\sum_{i=1}^n \br{A_{t+1}[i]-B_{t+1}[i]}^2}= \frac{2}{3}\sum_{i=1}^{n}\br{A_t[i]-B_t[i]}^2+\underbrace{\frac{1}{N}\sum_{\br{\overrightarrow{i^{(t)}},\overrightarrow{j^{(t)}}}}(\star\star\star)}_{(4\star)} \enspace .
\end{align}
Using the same calculation in \cref{eqn:t5}, we have
\begin{align}
(4\star)=-\frac{2}{3(n-1)}\sum_{i=1}^{n}\br{A_t[i]-B_t[i]}^2\enspace.\label{eqn:t5c}
	\end{align}

	Combining Eqs.~\eqref{eqn:t4c}\eqref{eqn:t5c}, we have
	\begin{align*}
		\E\!\Br{\sum_{i=1}^n \br{A_{l}[i]-B_{l}[i]}^2} &= \E\!\Br{\E\!\Br{\left.\sum_{i=1}^n \br{A_{l}[i]-B_{l}[i]}^2 \right|X_{l-1},Y_{l-1} } }\\
		&\leq \frac{2}{3}\E\!\Br{\sum_{i=1}^{n}\br{A_{l-1}[i]-B_{l-1}[i]}^2}\\
		&\leq \br{\frac{2}{3}}^l \sum_{i=1}^{n}\br{A_{0}[i]-B_{0}[i]}^2\leq 2\cdot\br{\frac{2}{3}}^l\enspace.
	\end{align*}
\end{proof}

\subsubsection{Proof of the Mixing Time}
\label{sec:mixing_time_w1_c}

\begin{proof}[Proof of \thm{mixing_time_w1_c}]
	Let $T=10(c+1)\log n$ for $c>0$. We couple two copies $\st{X_t}_{t\geq 0}$ and $\st{Y_t}_{t\geq 0}$ of the parallel Kac's walk with starting points $X_0=x\in \SC^n$ and $Y_0\sim \mu$, by applying the proportional coupling. We have
	\begin{equation*}
		\wone{\mathcal{L}\!\br{X_{T}}}{\mu} = \wone{\mathcal{L}\!\br{X_{T}}}{\mathcal{L}\!\br{Y_{T}}} \leq \expec{\norm{X_T-Y_T}_2} \leq \br{\expec{\norm{X_T-Y_T}_2^4}}^{1/4}\enspace.
	\end{equation*}
	Then by Cauthy-Schwarz inequality, we have 
	\begin{equation}\label{eq:w11_c}
		\wone{\mathcal{L}\!\br{X_{T}}}{\mu} \leq \br{n\expec{\norm{X_T-Y_T}_4^4}}^{1/4}\enspace.
	\end{equation}
	Note that the proportional coupling forces $X_T[i]$ and $Y_T[i]$ share the same argument for all $i\in[n]$. Therefore, for all $i\in[n]$
	$$
	\abs{X_T[i]-Y_T[i]} = \abs{\abs{X_T[i]}-\abs{Y_T[i]}} \leq \abs{X_T[i]}+\abs{Y_T[i]} \enspace .
	$$
	This gives us
	\begin{equation}\label{eq:w12_c}
		\norm{X_T-Y_T}_4^4 = \sum_{i=1}^n\abs{X_T[i]-Y_T[i]}^4 \leq \sum_{i=1}^n\br{\abs{X_T[i]}^2-\abs{Y_T[i]}^2}^2\enspace .
	\end{equation}
	Combining Eqs. \eq{w11_c} and \eq{w12_c}, we have
	\begin{align*}
		\wone{\mathcal{L}\!\br{X_{T}}}{\mu} 
		\leq &\br{n\expec{\sum_{i=1}^n\br{\abs{X_T[i]}^2-\abs{Y_T[i]}^2}^2}}^{1/4}\\
		(\text{\lem{contraction_lemma_c}})
		\leq &\br{2n\br{\frac{2}{3}}^{T}}^{1/4}
		\leq \frac{1}{2^{c\log n}} \enspace.
	\end{align*}
\end{proof}

\subsection{Proof of \thm{kactorssc}}\label{sec:proof_rssc}
To prove \thm{kactorssc}, we first introduce a new ensemble of (infinitely many) unitary operators
$\enscsrssconsc \coloneq \st{\csrssconsc }_{\secpar}$ with $\csrssconsc \coloneq $
	{\small$$\st{\csrssconsc_{\rep{\sigma}{T},\rep{\cf}{T},\rep{\cg}{T},\rep{\ch}{T}}}
	_{\rep{\sigma}{T} \subseteq S_{2^n},\rep{\cf}{T},\rep{\cg}{T},\rep{\ch}{T}\subseteq \{f: \bit{n-1}\to [0,1)\}}$$}
	and
	$$
		\csrssconsc_{\rep{\sigma}{T},\rep{\cf}{T},\rep{\cg}{T},\rep{\ch}{T}}
		= \widetilde{L}_{\tau_{T},\cf_{T},\cg_{T},\ch_{T}}\cdots
		  \widetilde{L}_{\tau_{2},\cf_{2},\cg_{2},\ch_{2}}
		  \widetilde{L}_{\tau_{1},\cf_{1},\cg_{1},\ch_{1}} \enspace
	$$
	where
	$\widetilde{L}_{\sigma,\cf,\cg,\ch}
	= U_{\sigma^{-1}} \widetilde{Q}_{\cf,\cg,\ch} U_{\sigma}$
	and
	$\widetilde{Q}_{\cf,\cg,\ch}$ is defined to be
	\begin{align}\label{eq:widetildeq}
		\widetilde{Q}_{\cf,\cg,\ch} =
		\sum_{y\in\bit{n-1}}
			U\!\br{\widetilde{\alpha}_y,\widetilde{\beta}_y,\widetilde{\theta}_y}
			\otimes\ketbra{y}\enspace,
	\end{align}
	in which $U(\alpha,\beta,\theta)$ is defined in \eq{parhaar} and for any $y\in\bit{n-1}$
	\[
	\widetilde{\theta}_y=\arcsin\br{\sqrt{\cf(y)}}\enspace, \enspace
	\widetilde{\alpha}_y=2\pi\cdot\cg(y)\enspace , \enspace
	\widetilde{\beta}_y=2\pi\cdot\ch(y)\enspace .
	\]	
	Similar to the real case,
	$\widetilde{L}_{\sigma,\cf,\cg,\ch}$ represents one step of parallel Kac's walk in complex space
	for independently and uniformly random $\sigma,\cf,\cg$ and $\ch$.
	
	\begin{lemma} \label{lem:ltildelclose}
	Let $\sigma\in S_{2^n}$
	and $\cf,\cg,\ch:\bit{n-1}\to[0,1)$.
	Let $f:\bit{n-1}\to \bit{d}$ be the function satisfying for any $y\in\bit{n-1}$,
	$f(y)$ is the $d$ digits after the binary point in $\widetilde{f}(y)$. The same applies to $g$ and $h$.
	Then
	$$
	{
		\norm{ L_{\sigma,f,g,h} - \widetilde{L}_{\sigma,\cf,\cg,\ch}
	}_\infty }
	\leq 2^{6-\frac{d}{2}} \enspace ,
	$$
	where $L_{\sigma,f,g,h}=U_{\sigma^{-1}} Q_{f,g,h} U_{\sigma}$
	is defined in \eq{normalq} and
	$\widetilde{L}_{\sigma,\cf,\cg,\ch}
	= U_{\sigma^{-1}} \widetilde{Q}_{\cf,\cg,\ch} U_{\sigma}$
	is defined in \eq{widetildeq}.
	\end{lemma}
	
	\begin{proof}[Proof of \lem{ltildelclose}]
		Recall the unitary $\widehat{L}_{\sigma,f,g,h}=U_{\sigma^{-1}} \widehat{Q}_{f,g,h} U_{\sigma}$ defined in \eq{widehatq}.
		We will prove
		\begin{itemize}
			\item $\norm{L_{\sigma,f,g,h} - \widehat{L}_{\sigma,f,g,h}}_\infty \leq 2^{3-d}\pi \enspace.$
			\item $ \norm{\widehat{L}_{\sigma,f,g,h} - \widetilde{L}_{\sigma,\cf,\cg,\ch}}_\infty \leq 2^{4-\frac{d}{2}} \enspace.$
		\end{itemize}
		The claim then follows from the triangle inequality.
		
		\paragraph{Proof of the first bound}
		Fix a $y\in\bit{n-1}$, we have
\begin{align*}
	\norm{
	\begin{pmatrix}
		\cos\br{\theta_{y}}&-\sin\br{\theta_{y}}\\
		\sin\br{\theta_{y}}&\cos\br{\theta_{y}}
	\end{pmatrix} -
	\begin{pmatrix}
		\cos\br{\frac{\pi}{2}\xi_{y}}&-\sin\br{\frac{\pi}{2}\xi_{y}}\\
		\sin\br{\frac{\pi}{2}\xi_{y}}&\cos\br{\frac{\pi}{2}\xi_{y}}
	\end{pmatrix}
	}_\infty
	& \leq 2^{-d-1}\pi\enspace,\\
	\norm{\left(
		\begin{matrix}
			e^{\i\br{\frac{\alpha_y+\beta_y}{2}} } & 0\\
			0 & e^{-\i\br{\frac{\alpha_y+\beta_y}{2}}}
		\end{matrix}
	\right) -
	\left(
		\begin{matrix}
			e^{\i2\pi\gamma_y^{+}} & 0\\
			0 & e^{-\i2\pi\gamma_y^{+}}
		\end{matrix}
	\right)
	}_\infty
	&\leq  2^{1-d}\pi\enspace,\\
	\norm{\left(
		\begin{matrix}
			e^{\i\br{\frac{\alpha_y-\beta_y}{2}} } & 0\\
			0 & e^{-\i\br{\frac{\alpha_y-\beta_y}{2}}}
		\end{matrix}
	\right) -
	\left(
		\begin{matrix}
			e^{\i2\pi\gamma_y^{-}} & 0\\
			0 & e^{-\i2\pi\gamma_y^{-}}
		\end{matrix}
	\right)
	}_\infty
	& \leq 2^{1-d}\pi\enspace.
\end{align*}
Thus, by the triangle inequality and the decomposition for matrix $U\!\br{\alpha,\beta,\theta}$, we have for any $y\in\bit{n-1}$
$$
\norm{U\!\br{\alpha_y,\beta_y,\theta_y}-U\!\br{2\pi(\gamma_y^{+}+\gamma_y^{-}),2\pi(\gamma_y^{+}-\gamma_y^{-}),\frac{\pi}{2}\xi_y}}_\infty \leq 2^{3-d}\pi\enspace.
$$
Therefore we have
	\begin{align*}
		&\quad\norm{L_{\sigma,f,g,h} - \widehat{L}_{\sigma,f,g,h}}_\infty = \norm{Q_{f,g,h} - \widehat{Q}_{f,g,h}}_\infty \\
		& = \max_{y\in\bit{n-1}} \norm{U\!\br{2\pi(\gamma_y^{+}+\gamma_y^{-}),2\pi(\gamma_y^{+}-\gamma_y^{-}),\frac{\pi}{2}\xi_y} - U\!\br{\alpha_y,\beta_y,\theta_y}}_\infty\\
		& \leq 2^{3-d}\pi\enspace.
	\end{align*}
	
	\paragraph{Proof of the second bound}
	Fix a $y\in\bit{n-1}$, we have
	$\abs{\alpha_y-\widetilde{\alpha}_y}\leq 2^{1-d}\pi$, and
	$\abs{\beta_y-\widetilde{\beta}_y}\leq 2^{1-d}\pi$.
	Therefore,
	\begin{align}
	\norm{\left(
		\begin{matrix}
			e^{\i\br{\frac{\alpha_y+\beta_y}{2}} } & 0\\
			0 & e^{-\i\br{\frac{\alpha_y+\beta_y}{2}}}
		\end{matrix}
	\right) -
	\left(
		\begin{matrix}
			e^{\i\br{\frac{\widetilde{\alpha}_y+\widetilde{\beta}_y}{2}} } & 0\\
			0 & e^{-\i\br{\frac{\widetilde{\alpha}_y+\widetilde{\beta}_y}{2}}}
		\end{matrix}
	\right)
	}_\infty
	& \leq 2^{1-d}\pi\enspace, \label{eq:m1}\\
	\norm{\left(
		\begin{matrix}
			e^{\i\br{\frac{\alpha_y-\beta_y}{2}} } & 0\\
			0 & e^{-\i\br{\frac{\alpha_y-\beta_y}{2}}}
		\end{matrix}
	\right) -
	\left(
		\begin{matrix}
			e^{\i\br{\frac{\widetilde{\alpha}_y-\widetilde{\beta}_y}{2}} } & 0\\
			0 & e^{-\i\br{\frac{\widetilde{\alpha}_y-\widetilde{\beta}_y}{2}}}
		\end{matrix}
	\right)
	}_\infty
	& \leq 2^{1-d}\pi\enspace. \label{eq:m2}
	\end{align}
	Moreover, we have $\abs{\val{f(y)}-\widetilde{f}(y)}\leq 2^{-d}$ and thus
	\begin{align}
	&\norm{
	\begin{pmatrix}
		\cos\br{\theta_{y}}&-\sin\br{\theta_{y}}\\
		\sin\br{\theta_{y}}&\cos\br{\theta_{y}}
	\end{pmatrix} -
	\begin{pmatrix}
		\cos\br{\widetilde{\theta}_y}&-\sin\br{\widetilde{\theta}_y}\\
		\sin\br{\widetilde{\theta}_y}&\cos\br{\widetilde{\theta}_y}
	\end{pmatrix}
	}_\infty \nonumber \\
	= ~	&\norm{
	\begin{pmatrix}
		\sqrt{1-\val{f(y)}}-\sqrt{1-\widetilde{f}(y)} & -\sqrt{\val{f(y)}}+\sqrt{\widetilde{f}(y)}\\
		\sqrt{\val{f(y)}}-\sqrt{\widetilde{f}(y)} & \sqrt{1-\val{f(y)}}-\sqrt{1-\widetilde{f}(y)}
	\end{pmatrix}
	}_\infty \nonumber \\
	\leq ~&\norm{
	\begin{pmatrix}
		\sqrt{1-\val{f(y)}}-\sqrt{1-\widetilde{f}(y)} & -\sqrt{\val{f(y)}}+\sqrt{\widetilde{f}(y)}\\
		\sqrt{\val{f(y)}}-\sqrt{\widetilde{f}(y)} & \sqrt{1-\val{f(y)}}-\sqrt{1-\widetilde{f}(y)}
	\end{pmatrix}
	}_2 \nonumber \\
	= ~& \sqrt{2\br{\sqrt{1-\val{f(y)}}-\sqrt{1-\widetilde{f}(y)}}^2+2\br{\sqrt{\val{f(y)}}-\sqrt{\widetilde{f}(y)}}^2}
	\leq  2^{-\frac{d}{2}+\frac{3}{2}}\enspace \label{eq:m3},
	\end{align}
	where the last inequality is by the following fact.
	\begin{fact}
		For $a,b\in[0,1]$ and $d\in\N$, if $\abs{a-b}\leq 2^{-d}$, then $\abs{\sqrt{a}-\sqrt{b}}\leq 2^{-\frac{d}{2}+\frac{1}{2}}.$
	\end{fact}
	\begin{proof}
	Let $\delta = 2^{-d+1}$. If $a\leq \delta$ and $b\leq \delta$, we have $\abs{\sqrt{a}-\sqrt{b}} \leq \max\st{\sqrt{a},\sqrt{b}}\leq \sqrt{\delta} = 2^{-\frac{d}{2}+\frac{1}{2}}$. If $a>\delta$ or $b>\delta$, we will have $a>\delta/2$ and $b>\delta/2$, and therefore
	$$
	\abs{\sqrt{a}-\sqrt{b}} \leq \frac{1}{\sqrt{2\delta}}\cdot \abs{a-b} \leq 2^{-\frac{d}{2}-1} \enspace.
	$$\end{proof}
	Hence, we have by \eq{m1}, \eq{m2} and \eq{m3}, for any $y\in\bit{n-1}$,
	\begin{align*}
		\norm{U\!\br{\alpha_y,\beta_y,\theta_y}-U\!\br{\widetilde{\alpha}_y,\widetilde{\beta}_y,\widetilde{\theta}_y}}_\infty \leq 2^{4-\frac{d}{2}}\enspace.
	\end{align*}
	Therefore,
	\begin{align*}
		\norm{\widehat{L}_{\sigma,f,g,h} - \widetilde{L}_{\sigma,\cf,\cg,\ch}}_\infty
		~& =~ \norm{\widehat{Q}_{f,g,h} - \widetilde{Q}_{\cf,\cg,\ch}}_\infty \\
		& = ~\max_{y\in\bit{n-1}}
		\norm{U\!\br{\alpha_y,\beta_y,\theta_y}-U\!\br{\widetilde{\alpha}_y,\widetilde{\beta}_y,\widetilde{\theta}_y}}_\infty\\
		& \leq ~2^{4-\frac{d}{2}}\enspace .
	\end{align*}
	
	\end{proof}

\begin{proof}[Proof of \thm{kactorssc}]
	It is easy to see that the uniformity condition is satisfied. 
	Let $\kappa$ denote the key length.
	Quantum circuit $\srssconsc$ applies the operator $\srssconsc_{\rep{\sigma}{T},\rep{f}{T},\rep{g}{T},\rep{h}{T}}$ after reading $\rep{\sigma}{T},\rep{f}{T},\rep{g}{T}$ and $\rep{h}{T}$.
	To implement $\srssconsc_{\rep{\sigma}{T},\rep{f}{T},\rep{g}{T},\rep{h}{T}}$, we need to realize each of
	the $T = 10 (\secpar + 1) n$ unitary gates $L$.
	Since each gate $L$ can be implemented in $\poly{n,\secpar,\klen}$ time,
	the total construction time for $\srssconsc_{\rep{\sigma}{T},\rep{f}{T},\rep{g}{T},\rep{h}{T}}$
	is also $\poly{n,\secpar,\klen}$.
	
	Thus, it suffices to prove the requirement of \emph{Statistical Pseudorandomness} is satisfied.
	Fix $\ket{\eta}\in\S(\H)$. Define three distributions:
	\begin{itemize}
		\item $\nu$ be the distribution of
			$\srssconsc_{\rep{\sigma}{T},\rep{f}{T},\rep{g}{T},\rep{h}{T}}\!\!\ket{\eta}$
			with independent and uniformly random permutations $\rep{\sigma}{T}\subseteq S_{2^n}$,
			and random functions $\rep{f}{T},\rep{g}{T}$, $\rep{h}{T}\subseteq \{f: \bit{n-1}\to\bit{d}\}$.
		\item $\widetilde{\nu}$ be the distribution of
			$\csrssconsc_{\rep{\sigma}{T},\rep{\cf}{T},\rep{\cg}{T},\rep{\ch}{T}}\!\!\ket{\eta}$
			with independent and uniformly random permutations
			$\rep{\sigma}{T}\subseteq S_{2^n}$,
			and random functions
			$\rep{\cf}{T},\rep{\cg}{T}$, $\rep{\ch}{T} \subseteq\{f: \bit{n-1}\to[0,1)\}$.
		\item $\mu$ be the Haar measure on $\S(\H)$.
	\end{itemize}
	
	We first proof the trace distance between $\nu$ and $\widetilde{\nu}$ is negligible.
	To this end, we construct a coupling $\gamma_0$ of $\nu$ and $\widetilde{\nu}$ by using the same permutation $\sigma_t$ and letting $f_t$ be the function satisfying $f_t(y)$ is the $d$ digits after the binary point in $\widetilde{f}_t(y)$ for all $y\in\bit{n-1}$ (The same applies to $g_t$ and $h_t$).
	Therefore, for any $\br{\ket{\phi}, \ket{\varphi}}\sim \gamma_0$,
	we have
	\begin{align}
		&\norm{\ket{\phi} - \ket{\varphi}}_2 ~\nonumber\\
	 =	~& \norm{
			\srssconsc_{\rep{\sigma}{T},\rep{f}{T},\rep{g}{T},\rep{h}{T}}
		\!\!\ket{\eta}-
			\csrssconsc_{\rep{\sigma}{T},\rep{\cf}{T},\rep{\cg}{T},\rep{\ch}{T}}\!\!\ket{\eta}
		}_2 \nonumber\\
	\leq  ~& \norm{
			\srssconsc_{\rep{\sigma}{T},\rep{f}{T},\rep{g}{T},\rep{h}{T}}
			- \csrssconsc_{\rep{\sigma}{T},\rep{\cf}{T},\rep{\cg}{T},\rep{\ch}{T}}
		}_\infty \nonumber\\
	\leq ~& 2^{6-\frac{d}{2}} T = \frac{640 (\lambda+1)n}{\lambda^{\log \lambda}\cdot n^{\log n}}
	\enspace, \nonumber
	\end{align}
	where the last inequality is from \fct{uprod} and \lem{ltildelclose}.
	Thus, for any $l\in\poly{\secpar,n}$
	\begin{align} \label{eq:h1_h2_rssc}
		~&\norm{
			\E_{\ket{\phi}\sim \nu}\!\Br{\br{\ketbra{\phi}}^{\otimes l}}
			-
			\E_{\ket{\varphi}\sim \widetilde{\nu}}\!\Br{\br{\ketbra{\varphi}}^{\otimes l}}
		}_1 \nonumber\\
		\leq ~& \E_{\br{\ket{\phi}, \ket{\varphi}}\sim \gamma_0}\!\Br{
		\norm{
			\br{\ketbra{\phi}}^{\otimes l}
			-
			\br{\ketbra{\varphi}}^{\otimes l}
		}_1} \nonumber\\
		\leq ~& l \E_{\br{\ket{\phi}, \ket{\varphi}}\sim \gamma_0}\!\Br{
		\norm{
			\ketbra{\phi}
			-
			\ketbra{\varphi}
		}_1} \nonumber\\
		\leq ~& l
		\br{\E_{\br{\ket{\phi}, \ket{\varphi}}\sim \gamma_0}\!\Br{
		\norm{
			\ket{\phi}\br{ \bra{\phi} - \bra{\varphi}}
		}_1} + \E_{\br{\ket{\phi}, \ket{\varphi}}\sim \gamma_0}\!\Br{
		\norm{
			\br{ \ket{\phi} - \ket{\varphi}}\bra{\varphi}
		}_1}}\nonumber\\
		\leq ~& 2l
		\E_{\br{\ket{\phi}, \ket{\varphi}}\sim \gamma_0}\!\Br{
		\norm{ \ket{\phi} - \ket{\varphi} }_2 } \leq \frac{1280 (\lambda+1)nl}{\lambda^{\log \lambda}\cdot n^{\log n}} \enspace.
	\end{align}
	As for the trace distance between $\widetilde{\nu}$ and $\mu$, note that $\widetilde{\nu}$ is the output distribution of $T$-step parallel Kac's walk. Thus by \thm{mixing_time_w1_c}, we have
	\begin{align*}
		\wone{\widetilde{\nu}}{\mu} \leq \frac{1}{2^{\lambda n}} \enspace.
	\end{align*}
	So there exists a coupling of $\widetilde{v}$ and $\mu$, denoted by $\gamma_1$, that achieves
	\begin{align*}
		\E_{\br{\ket{\varphi}, \ket{\psi}}\sim \gamma_1}
		\!\Br{ \norm{\ket{\varphi}-\ket{\psi}}_2 }\leq \frac{3}{2^{\lambda n}} \enspace .
	\end{align*}
	Therefore, similar to Eq. \eq{h1_h2_rssc}, we have for any $l\in\poly{\secpar,n}$
	\begin{align} \label{eq:h2_h3_rssc}
		\norm{
			\E_{\ket{\varphi}\sim \widetilde{\nu}}\!\Br{\br{\ketbra{\varphi}}^{\otimes l}}
			-
			\E_{\ket{\psi}\sim \mu}\!\Br{\br{\ketbra{\psi}}^{\otimes l}}
		}_1 
		\leq 2l
		\E_{\br{\ket{\varphi}, \ket{\psi}}\sim \gamma_1}\!\Br{
		\norm{ \ket{\varphi} - \ket{\psi} }_2 } \leq \frac{6l}{2^{\lambda n}}\enspace .
	\end{align}
	Finally, by the triangle inequality, Eqs.~\eq{h1_h2_rssc} and \eq{h2_h3_rssc}, we have
	\begin{align*}
		\norm{
			\E_{\ket{\phi}\sim \nu}\!\Br{\br{\ketbra{\phi}}^{\otimes l}} -
			\E_{\ket{\psi}\sim \mu}\!\Br{\br{\ketbra{\psi}}^{\otimes l}} }_1 \leq
			\frac{1280 (\lambda+1)nl}{\lambda^{\log \lambda}\cdot n^{\log n}} +
			\frac{6l}{2^{\lambda n}} = \negl{\secpar} \enspace.
	\end{align*}
	This establishes the \emph{Statistical Pseudorandomness} property.

\end{proof}

\subsection{Proof of \thm{complexprss}}
\begin{proof}[Proof of \thm{complexprss}]
	The key length is bounded by $4T\cdot\poly{n,d} = \poly{n,\secpar}$
	since $\tau$ and $F$ are efficient. Thus the condition of polynomial-bounded key length is satisfied.
	To implement $\sgenc_k$, we need to realize each of
	the $T = 10 (\secpar + 1) n$ unitary gates $L$
	that compose $\sgenc_k$.
	Since each gate $L$ can be implemented in $\poly{n,\secpar}$ time
	due to the efficiency of $\tau$ and $F$,
	the total construction time for $\sgenc_k$
	is also $\poly{n,\secpar}$. Thus the uniformity is also satisfied.
	
	We now prove the pseudorandomness property.
	To this end, we consider three hybrids for an arbitrary $\ket{\phi}\in\S(\H)$ and $l\in\poly{\secpar,n}$:
	\begin{itemize}
		\item[H1:] \label{itm:h1c}
		$\ket{\phi_k}^{\otimes l}$ for
		$\ket{\phi_k} = \sgenc_k\!\ket{\phi}$
		where
		$k\leftarrow\br{\K_1\times\K_2\times\K_2\times\K_2}^T$
		is chosen uniformly at random.
		
		\item[H2:] \label{itm:h2c}
		$\ket{\varphi_{\rep{\sigma}{T},\rep{f}{T},\rep{g}{T},\rep{h}{T}}}^{\otimes l}$ for
		$$
		\ket{\varphi_{\rep{\sigma}{T},\rep{f}{T},\rep{g}{T},\rep{h}{T}}}
		= \srssconsc_{\rep{\sigma}{T},\rep{f}{T},\rep{g}{T},\rep{h}{T}}\!\!\ket{\phi}
		$$
		with independent and uniformly random permutations $\rep{\sigma}{T}\subseteq S_{2^n}$
		and random functions $\rep{f}{T},\rep{g}{T},\rep{h}{T}\subseteq \{f: \bit{n-1}\to\bit{d}\}$. Here the unitary 
		$\srssconsc_{\rep{\sigma}{T},\rep{f}{T},\rep{g}{T},\rep{h}{T}}$
			is defined in \defi{kactosrssc}.
		
		\item[H3:] \label{itm:h3c}
		$\ket{\psi}^{\otimes l}$ for $\ket{\psi}$ chosen according to the Haar measure $\mu$ on $\S(\H)$.
	\end{itemize}
	
	We first prove that \itm{h1c}{H1} and \itm{h2c}{H2}
	are computationally indistinguishable.
	By the quantum-secure property of $\tau$ and $F$,
	we know the following two situations are computationally indistinguishable for any polynomial-time quantum oracle algorithm $\A$ (see \lem{twokeys}):
	\begin{itemize}
		\item given oracle access to
		$\tau_{r_1},\cdots,\tau_{r_T}$ and $F_{u_1},\cdots,F_{u_T},F_{s_1},\cdots,F_{s_T},F_{t_1},\cdots,F_{t_T}$
		where
		$\rep{r}{T}\subseteq\K_1$ and $\rep{u}{T},\rep{s}{T},\rep{t}{T}\subseteq \K_2$
		are independent and uniformly random keys.
		
		\item given oracle access to
		independent and uniformly random permutations $\rep{\sigma}{T}$ $\subseteq S_{2^n}$
		and random functions
		$\rep{f}{T},\rep{g}{T},\rep{h}{T}\subseteq\{f: \bit{n-1}\to\bit{d}\}$.
	\end{itemize}
	Thus, we have for any
	polynomial-time quantum algorithm $\A$,
	$$
	\abs{
		\Pr\Br{\A\br{\ket{\phi_k}^{\otimes l}} = 1} -
		\Pr\Br{\A\br{\ket{\varphi_{\rep{\sigma}{T},\rep{f}{T},\rep{g}{T},\rep{h}{T}}}^{\otimes l}} = 1}
	} = \negl{\secpar} \enspace .
	$$
	
	For \itm{h2c}{H2} and \itm{h3c}{H3},
	they are statistically indistinguishable since $\enssrssconsc$ defined in \defi{kactosrssc}
	is an \rss~by \thm{kactorssc}.
	Finally, by the triangle inequality we establish \itm{h1c}{H1} and \itm{h3c}{H3}
	are computationally indistinguishable.
	This accomplishes the proof.
\end{proof}

\subsection{Proof of \lem{good_start}}
\begin{proof}[Proof of \lem{good_start}]
Let graph $G_0=(V=[n],E_0=\emptyset)$.
	We recursively define $G_1,\dots,G_l$ as follows:
	given $G_i=([n],E_i)$,
	choose a perfect matching $M_i$ of $[n]$
	uniformly at random,
	and set
	$$G_{i+1}=\br{[n],E_{i+1}=E_i\cup M_i}\enspace.$$
	
	Then
	\begin{align*}
		&\quad\Pr\!\Br{\PP_{T_0,1}\neq\st{[n]}}\\
		&=\Pr\!\Br{G_l\text{ is disconnected}}\\
		&=\Pr\!\Br{\exists ~S\subseteq [n]\text{ such that there is no edge between }S\text{ and }[n]\backslash S \text{ in } G_l}\\
		&\leq \sum_{i\in[m]\atop i\text{ is even}} \sum_{S\in \binom{[n]}{i}}\Pr\!\Br{\text{There is no edge between }S\text{ and }[n]\backslash S \text{ in } G_l}\\
		&\leq \sum_{i\in[m]\atop i\text{ is even}} \sum_{S\in \binom{[n]}{i}}
		\br{\frac{\frac{i!}{2^{i/2}(i/2)!}\cdot\frac{(n-i)!}{2^{(n-i)/2}((n-i)/2)!}}{\frac{n!}{2^{n/2}(n/2)!}}	}^l\\
		&=\sum_{i\in[m]\atop i\text{ is even}} \sum_{S\in \binom{[n]}{i}}
		\br{\frac{(n-i)(n-i-1)\cdots((n-i)/2+1)}{n(n-1)\cdots ((n+i)/2+1)}\cdot \frac{i(i-1)\cdots(i/2+1)}{((n+i)/2)\cdots(n/2+1)}}^l\\
		&\leq \sum_{i\in[m],~ i\text{ is even}} \sum_{S\in \binom{[n]}{i}}
		\br{\frac{i(i-1)\cdots(i/2+1)}{((n+i)/2)\cdots(n/2+1)}}^l\\
		&\leq \sum_{i\in[m]\atop i\text{ is even}} \sum_{S\in \binom{[n]}{i}}
		\br{\frac{2i}{n+i}}^{il/2}\\
		&\leq \sum_{i\in[m]\atop i\text{ is even}} \binom{n}{i}\br{\frac{2}{3}}^{il/2}
		\leq \br{1+\br{\frac{2}{3}}^{l/2}}^n-1
		\leq n\br{\frac{2}{3}}^{l/2}\br{1+\br{\frac{2}{3}}^{l/2}}^{n-1}\enspace.
		\end{align*}
	When $l=5(1+c)\log n$ and $n$ is sufficiently large,
	$$
	n\br{\frac{2}{3}}^{l/2}\leq n^{-c}
	\quad\text{and}\quad
	\br{1+\br{\frac{2}{3}}^{l/2}}^{n-1}\leq 2\enspace.
	$$
\end{proof}

\subsection{Proof of ~\thm{mixing_time_c}}
To prove \thm{mixing_time_c}, we extend the two-stage coupling introduced in real case to complex case. We have introduced the proportional coupling of complex space in \defi{prop_cpl_c}. We now extend the the non-Markovian coupling and then prove the mixing time. We assume $n=2m$.

\subsubsection{Non-Markovian Coupling}

\begin{definition}[Non-Markovian Coupling]\label{def:non_mark_cpl_c}
Fix $T_0\leq T\in\N$. We couple $\st{X_t}_{T_0\leq t\leq T}, \st{Y_t}_{T_0\leq t\leq T}$ in the following way:
\begin{enumerate}
	\item For each $T_0\leq t < T$, choose a perfect matching $$P_t = \st{\br{i^{(t)}_1, j^{(t)}_1},\dots,\br{i^{(t)}_m, j^{(t)}_m}}$$ uniformly at random.

	\item Set $\PP_{T,1} = \st{\st{1},\dots,\st{n}}$, and define a sequence of partitions $$\st{\PP_{t,k}}_{T_0\leq t<T,~ 1\leq k\leq m+1}$$ of $[n]$ in the same way as \defi{non_mark_cpl}.
		
	\item If $\PP_{T_0,1}=\st{\st{1, \dots, n}}$, we couple $\st{X_t}_{T_0\leq t\leq T}, \st{Y_t}_{T_0\leq t\leq T}$ in the following way:
	\begin{itemize}
			\item Define the set
			$$ H=\st{(t,k):T_0\leq t<T,~1\leq k\leq m,~\PP_{t,k}\neq\PP_{t,k+1}} \enspace.$$
			
			\item Fix $T_0\leq t<T$, $X_{t}$ and $Y_{t}$, and we couple $X_{t+1}$ and $Y_{t+1}$ in the following way:
				\begin{enumerate}
					\item Set $X_{t,1}=X_t$ and $Y_{t,1}=Y_t$.
					\item For $1\leq k\leq m$,
						\begin{enumerate}
							\item If $(t,k)\notin H$, we obtained $X_{t,k+1}$ and $Y_{t,k+1}$ in the same way as the proportional coupling defined in \defi{prop_cpl_c}.
							
							\item If $(t,k)\in H$, let $$l^{(t)}_k = \sqrt{\abs{X_{t,k}[i^{(t)}_k]}^2+\abs{X_{t,k}[j^{(t)}_k]}^2}$$ and $${l'}^{(t)}_k = \sqrt{\abs{Y_{t,k}[i^{(t)}_k]}^2+\abs{Y_{t,k}[j^{(t)}_k]}^2}.$$ Let $U_0$ and $U_0'$ be the unitary operators which satisfy
								\begin{align*}
									U_0\br{\begin{matrix}
									X_{t,k}[i^{(t)}_k]\\
									X_{t,k}[j^{(t)}_k]
									\end{matrix}} = \br{\begin{matrix}
									l^{(t)}_k\\
									0
									\end{matrix}}
									\text{\quad and\quad}
									U_0'\br{\begin{matrix}
									Y_{t,k}[i^{(t)}_k]\\
									Y_{t,k}[j^{(t)}_k]
									\end{matrix}} = \br{\begin{matrix}
									{l'}^{(t)}_k\\
									0
									\end{matrix}}\enspace.
							\end{align*}
							Then we choose the best distribution $\nu$ among all joint distributions on $[0,1)\times [0,1)$ with both marginal distributions uniformly distributed on $[0,1)$ which maximizes the probability of the following events when $(\zeta,\zeta')\sim\nu$ and $\alpha, \beta$ are uniformly sample from $[0,2\pi)$:
								\begin{align*}
									\sum_{i\in S_{r}(t,k+1)} \abs{ X_{t,k+1} [i]}^2 &= \sum_{i\in S_{r}(t,k+1)} \abs{Y_{t,k+1} [i]}^2, \enspace 1\leq r\leq l_{t,k+1}
								\end{align*}
								where
								\begin{align*}
								X_{t,k+1} = G_{\C}\!\br{i^{(t)}_k,j^{(t)}_k,\alpha,\beta,\arcsin \sqrt{\zeta}, U_0X_{t,k}}\enspace,\\
								Y_{t,k+1} = G_{\C}\!\br{i^{(t)}_k,j^{(t)}_k,\alpha,\beta,\arcsin \sqrt{\zeta'}, U_0'Y_{t,k}}\enspace.
								\end{align*}
								Then choose $(\zeta^{(t)}_k,{\zeta'}^{(t)}_k)\sim\nu$ and $\alpha^{(t)}_k, \beta^{(t)}_k$ uniformly from $[0,2\pi)$, and set
								\begin{align*}
								X_{t,k+1} = G_{\C}\!\br{i^{(t)}_k,j^{(t)}_k,\alpha^{(t)}_k,\beta^{(t)}_k,\arcsin \sqrt{\zeta^{(t)}_k}, U_0X_{t,k}}\enspace,\\
								Y_{t,k+1} = G_{\C}\!\br{i^{(t)}_k,j^{(t)}_k,\alpha^{(t)}_k,\beta^{(t)}_k,\arcsin \sqrt{{\zeta'}^{(t)}_k}, U_0'Y_{t,k}}\enspace.
								\end{align*}
						\end{enumerate}
					\item Set $X_{t+1} = X_{t,m+1}$ and $Y_{t+1} = Y_{t,m+1}$.
				\end{enumerate}
		\end{itemize}
	
	\item If $\PP_{T_0,1}\neq \st{\st{1, \dots, n}}$, for $T_0\leq t\leq T$, we couple $X_{t+1}$ and $Y_{t+1}$ in the following way:
	 	choose $2m$ independent angles
		$$\alpha_1^{(t)},\dots,\alpha_m^{(t)},\beta_1^{(t)},\dots,\beta_m^{(t)}\in[0,2\pi)$$
		uniformly at random.
		Additionally, $m$ independent real numbers $\zeta_1^{(t)}\dots,\zeta_m^{(t)}\in[0,1)$ are selected uniformly at random and compute
		$$
		\theta_k^{(t)}=\arcsin\br{\sqrt{\zeta_k^{(t)}}}
		$$
		for all $k\in\st{1,\dots,m}$.
		We set
		$$X_{t+1}=F_{\C}\!\br{P_t,\st{\alpha_k^{(t)}}_{k=1}^m,\st{\beta_k^{(t)}}_{k=1}^m,\st{\theta_k^{(t)}}_{k=1}^m,X_t} \enspace,$$
		$$Y_{t+1}=F_{\C}\!\br{P_t,\st{\alpha_k^{(t)}}_{k=1}^m,\st{\beta_k^{(t)}}_{k=1}^m,\st{\theta_k^{(t)}}_{k=1}^m,Y_t} \enspace.$$
		
\end{enumerate}
\end{definition}

For $T_0\leq t \leq T$, $1\leq k \leq m+1$ and $1\leq i\leq n$ we define
$$
A_{t,k}[i] = \abs{X_{t,k}[i]}^2\enspace,\quad B_{t,k}[i] = \abs{Y_{t,k}[i]}^2\enspace,
$$
and define the event $\A(t,k)$ by
$$
\A(t,k) = \st{\text{Eq. \eq{cond1_c} are satisfied for all $(t',k')\sqsubseteq (t,k)$ such that $T_0\leq t' \leq t$}}.
$$
\begin{align}
	\sum_{i\in S_{r}(t',k')}A_{t',k'}[i] &= \sum_{i\in S_{r}(t',k')} B_{t',k'}[i]\enspace, \quad 1\leq r\leq l_{t',k'}\enspace.\label{eq:cond1_c}	
\end{align}

Similarly, we have the following two lemmas as \lem{never_drift} and \lem{good_dist} before.

\begin{lemma}
\label{lem:never_drift_c}
	Fix $T_0<T$ and two chains $\st{X_t}_{T_0\leq t\leq T}, \st{Y_t}_{T_0\leq t\leq T}$ are coupled using the non-Markovian coupling defined in \defi{non_mark_cpl_c}.
	Fix $T_0\leq t \leq T$ and $1\leq k \leq m+1$.
	Then, on the event $\A(t,k)\cap\st{\PP_{T_0,1}=\st{1, \dots, n}}$, we have
	$$
	\norm{A_{t',k'}-B_{t',k'}}_{1,S}\leq \norm{A_{T_0}-B_{T_0}}_{1}
	$$
	for all $(t',k')\sqsubseteq(t,k)$ such that $T_0\leq t'\leq t$ and $S\in\PP_{t',k'}$. Moreover, for all $(t',k')\sqsubseteq(t,k)$ such that $T_0\leq t'\leq t$,
	$$
	\norm{A_{t',k'}-B_{t',k'}}_{1}\leq n\norm{A_{T_0}-B_{T_0}}_{1}\enspace.
	$$
\end{lemma}
The proof of above lemma is the same as Lemma 4.4 in \cite{PS17}.

\begin{lemma} \label{lem:good_dist_c}
	Fix positive reals $1<p<q$. Let $\zeta,\zeta'\sim\mathrm{Unif}[0,1)$ and let
	$$
	S=A+B\zeta \quad\text{ and }
	\quad S'=C+D\zeta'
	$$
	for some $0\leq A,B,C,D\leq 1$ that satisfy
	$$
	\abs{A-C},\abs{B-D}\leq n^{-q}\quad \text {and}\quad B,D\geq n^{-p}\enspace.
	$$
	Then for sufficiently large $n$, there exists a coupling of $\zeta,\zeta'$ so that
	$$
	\Pr\!\Br{S=S'}\geq 1-3 n^{-(q-p)}\enspace.
	$$
\end{lemma}

\begin{proof}
	Without loss of generality, we assume $B\geq D$. The total variation distance between $S$ and $S'$ is
	\begin{align*}
		\norm{S-S'}_{\mathrm{TV}} &\leq 1 - \int_{\br{A,A+B}\cap \br{C,C+D}} \frac{1}{B}~ \de{x}\\
		& \leq 1 - \int_{A+n^{-q}} ^{A+B-2n^{-q}}  \frac{1}{B}~ \de{x}\\
		& = 1 - \br{B-3n^{-q}}\frac{1}{B}\\
		& = \frac{3n^{-q}}{B} \leq 3 n^{-(q-p)}\enspace.
	\end{align*}
	This implicitly defines a coupling of $\zeta,\zeta'$ that satisfies $
	\Pr\!\Br{S=S'}\geq 1-3 n^{-(q-p)}.
	$
\end{proof}

\subsubsection{Proof of \thm{mixing_time_c}}
Let $a=30,~b=24,~ T_0=500\log n, ~ T_1=15\log n,~ T=T_0+T_1 = 515\log n$. We construct a coupling of two copies $\st{X_t}_{t\geq 0}, \st{Y_t}_{t\geq 0}$ with starting point $X_0\in \SC^{n}$ and $Y_0\sim \mu_{\C}$. The coupling is as follows:
\begin{enumerate}
	\item Couple $\st{X_t}_{0\leq t\leq T_0}, \st{Y_t}_{0\leq t\leq T_0}$ by using the proportional coupling defined in \defi{prop_cpl_c}.
	\item Couple $\st{X_t}_{T_0\leq t\leq T}, \st{Y_t}_{T_0\leq t\leq T}$ by using the non-markovian coupling defined in \defi{non_mark_cpl_c}.
\end{enumerate}
Define the event
\begin{align*}
	\EE_1 &= \st{\norm{A_{T_0}-B_{T_0}}_1\geq n^{-a}}\enspace,\\
	\EE_2 &= \st{\PP_{T_0,1}\neq \st{\st{1,\dots,n}}}\enspace,\\
	\EE_3 &= \st{X_{T}\neq Y_{T}}\enspace.
\end{align*}
By \lem{coupling_lemma},
\begin{align}
	\sup_{X_0\in \SC^{n}} \norm{\mathcal{L}\br{X_{T}}-\mu_{\C}}_{\mathrm{TV}}
	&\leq \sup_{X_0\in \SC^{n}}\Pr\!\Br{\EE_3}\nonumber\\
	&\leq \sup_{X_0\in \SC^{n}}
	\br{\Pr\!\Br{\EE_1}+\Pr\!\Br{\EE_2}+\Pr\!\Br{\EE_3\cap \EE_1^c \cap \EE_2^c}}\enspace.\label{eq:tv_c}
\end{align}
By \lem{contraction_lemma_c} and Markov's inequality, we have
\begin{align}
	\Pr\!\Br{\EE_1} &= \Pr\!\Br{\norm{A_{T_0}-B_{T_0}}_1\geq n^{-a}}\nonumber\\
	&\leq \Pr\!\Br{\norm{A_{T_0}-B_{T_0}}_2\geq n^{-a-1/2}}\nonumber\\
	&\leq n^{2a+1}\cdot2\cdot\br{\frac{2}{3}}^{T_0}\leq \frac{1}{n^2}\enspace. \label{eq:event1_c}
\end{align}
Moreover, by \lem{good_start}, we have
\begin{align}
	\Pr\!\Br{\EE_2}\leq 2n^{-2}\enspace.\label{eq:event2_c}
\end{align}
In order to bound $\Pr\!\Br{\EE_3\cap \EE_1^c \cap \EE_2^c}$,
recall the definition of $\A(t,k)$ in \eq{cond1_c}.
If $\A(T,1)$ occurs, we have $\abs{X_t[i]}=\abs{Y_t[i]}$ for all $i\in\st{1,\dots,n}$.
Meanwhile, our coupling ensures $X_t[i]$ and $Y_t[i]$ have the same argument. This means $\A(T,1)$ implies $\EE_3^c$. As a result, $\EE_3\cap \EE_1^c \cap \EE_2^c$ implies $\A(T,1)^c \cap \EE_1^c \cap \EE_2^c$. So, we have
\begin{align}
&\quad\ \Pr\!\Br{\EE_3\cap \EE_1^c \cap \EE_2^c}\nonumber\\
&\leq
\Pr\!\Br{\A(T,1)^c \cap \EE_1^c \cap \EE_2^c} \nonumber\\
&\leq
\sum_{t=T_0}^{T-1} \Pr\!\Br{\A(t+1,1)^c \cap \A(t,1) \cap \EE_1^c \cap \EE_2^c} + \Pr\!\Br{\A(T_0,1)^c \cap \EE_1^c \cap \EE_2^c} \nonumber \\
&=
\sum_{t=T_0}^{T-1} \Pr\!\Br{\A(t,m+1)^c \cap \A(t,1) \cap \EE_1^c \cap \EE_2^c} + \Pr\!\Br{\A(T_0,1)^c \cap \EE_1^c \cap \EE_2^c} \nonumber \\
&\leq
\sum_{t=T_0}^{T-1} \sum_{k=1}^{m} \Pr\!\Br{\A(t,k+1)^c \cap \A(t,k) \cap \EE_1^c \cap \EE_2^c} + \Pr\!\Br{\A(T_0,1)^c \cap \EE_1^c \cap \EE_2^c}\enspace.\label{eq:e3_c}
\end{align}
Notice that if $\EE_2^c$ happens, we have
$$
\sum_{i\in \st{1, \dots, n}} \abs{X_{T_0,1} [i]}^2 = \sum_{i\in \st{1, \dots, n}} \abs{Y_{T_0,1} [i]}^2 = 1\enspace .\\
$$
Therefore,
\begin{align}
	\Pr\!\Br{\A(T_0,1)^c \cap \EE_1^c \cap \EE_2^c}=0\enspace.\label{eq:imps_c}
\end{align}
Combining \eq{e3_c} and \eq{imps_c}, we have
\begin{equation}
	\Pr\!\Br{\EE_3\cap \EE_1^c \cap \EE_2^c}\leq \sum_{t=T_0}^{T-1} \sum_{k=1}^{m} \Pr\!\Br{\A(t,k+1)^c \cap \A(t,k) \cap \EE_1^c \cap \EE_2^c}\enspace.\label{eq:e3final_c}
\end{equation}

We are now left to find a upper bound for $\Pr\!\Br{\A(t,k+1)^c \cap \A(t,k) \cap \EE_1^c \cap \EE_2^c}$ when $T_0\leq t\leq T-1$ and $1\leq k\leq m$. To this end, we define
$$
\B(t,k) = \st{\min_{(t',k')\sqsubseteq (t,k):t'\leq t}\min_{1\leq i\leq n} \abs{Y_{t',k'}[i]}^2\geq (2n)^{-b}}\enspace.
$$
Note that
\begin{align}
&\Pr\!\Br{\A(t,k+1)^c \cap \A(t,k) \cap \EE_1^c \cap \EE_2^c}\nonumber\\
&\hspace{2cm}\leq \Pr\!\Br{\A(t,k+1)^c \cap \A(t,k) \cap \B(t,k) \cap \EE_1^c \cap \EE_2^c} + \Pr\!\Br{\B(t,k)^c}\enspace.	\label{eq:e3single_c}
\end{align}
By \lem{marg_of_haar_c} and a union bound over all $(t',k')$ such that $(t',k')\sqsubseteq (t,k)$ and $t'\leq t$, we have for sufficiently large $n$,
\begin{equation}\label{eq:e3single2_c}
	\Pr\!\Br{\B(t,k)^c} \leq 15\cdot 2^{1-\frac{b}{3}}\cdot n^{3-\frac{b}{3}}\log(n)\leq \frac{1}{n^4}\enspace.
\end{equation}

Next, we consider two cases of the term $$\Pr\!\Br{\A(t,k+1)^c \cap \A(t,k) \cap \B(t,k) \cap \EE_1^c \cap \EE_2^c}$$ in \eq{e3single_c}: $(t,k)\notin H$ and $(t,k)\in H$,
where $H$ is defined in \defi{non_mark_cpl_c}.
In the case that $(t,k)\notin H$, we have $\PP_{t,k} = \PP_{t,k+1}$ and we apply the proportional coupling. Thus
$\A(t,k)$ implies $\A(t,k+1)$ which means
\begin{align}\label{eq:e3single11_c}
	\Pr\!\Br{\A(t,k+1)^c \cap \A(t,k) \cap \B(t,k) \cap \EE_1^c \cap \EE_2^c}=0\enspace.
\end{align}

In the other case that $(t,k)\in H$,
let
$$
A=\sum_{i\in S_{v_{t,k}}(t,k+1)\backslash j^{(t)}_k} \abs{X_{t,k}[i]}^2\enspace,
\quad \quad
B= \abs{X_{t,k} [i^{(t)}_k]}^2+\abs{X_{t,k} [j^{(t)}_k]}^2\enspace,
$$

$$
C=\sum_{i\in S_{v_{t,k}}(t,k+1)\backslash j^{(t)}_k} \abs{Y_{t,k}[i]}^2
\enspace,\quad \quad
D= \abs{Y_{t,k}[i^{(t)}_k]}^2+\abs{Y_{t,k}[j^{(t)}_k]}^2\enspace,
$$

$$
S = A + B\zeta^{(t)}_k
\enspace,\quad \quad
S'= C + D{\zeta'}^{(t)}_k\enspace.
$$
On the event $\A(t,k) \cap \B(t,k) \cap \EE_1^c \cap \EE_2^c$, we have by \lem{never_drift}
\begin{align*}
	\abs{A-C} \leq \norm{A_{t,k}-B_{t,k}}_1\leq n\norm{A_{T_0}-B_{T_0}}_1 \leq n^{1-a}\enspace.
\end{align*}
Similarly,
\begin{align*}
	\abs{B-D} \leq \norm{A_{t,k}-B_{t,k}}_1\leq n\norm{A_{T_0}-B_{T_0}}_1 \leq n^{1-a}\enspace.
\end{align*}
Moreover, $D\geq n^{-(b+1)}$ and $B\geq D-\abs{B-D}\geq n^{-(b+1)}$ for sufficiently large $n$.
Then apply \lem{good_dist_c} with $p=b+1,~ q=a-1$, we know there exists a distribution $\nu_0$ such that when $\br{\zeta^{(t)}_k,{\zeta'}^{(t)}_k}\sim \nu_0$, we have
$$
\Pr_{\br{\zeta^{(t)}_k,{\zeta'}^{(t)}_k}\sim \nu_0}\!
\Br{ S\neq S' | \A(t,k) \cap \B(t,k) \cap \EE_1^c \cap \EE_2^c} \leq 3 n^{-(a-b-2)}\enspace.
$$
We choose the best distribution which maximizes the probability of event described in \eq{cond1_c} , so
\begin{align} \label{eq:e3single12_c}
	&\Pr\!\Br{\A(t,k+1)^c \cap \A(t,k) \cap \B(t,k) \cap \EE_1^c \cap \EE_2^c} \nonumber\\
	\leq  &\Pr\Br{\left.\A(t,k+1)^c \right| \A(t,k) \cap \B(t,k) \cap \EE_1^c \cap \EE_2^c} \nonumber \\
	\leq &\Pr_{\br{\zeta^{(t)}_k,{\zeta'}^{(t)}_k}\sim \nu_0}\Br{\left.\A(t,k+1)^c \right|\A(t,k) \cap \B(t,k) \cap \EE_1^c \cap \EE_2^c} \nonumber \\
	= &\Pr_{\br{\zeta^{(t)}_k,{\zeta'}^{(t)}_k}\sim \nu_0}\Br{\left.S\neq S' \right|\A(t,k) \cap \B(t,k) \cap \EE_1^c \cap \EE_2^c} \leq 3 n^{-(a-b-2)}\enspace.
\end{align}

Combining \eq{e3final_c}, \eq{e3single_c}, \eq{e3single11_c}, \eq{e3single12_c} and \eq{e3single2_c}, we have for sufficiently large $n$
\begin{align} \label{eq:event3_c}
	\Pr\!\Br{\EE_3\cap \EE_1^c \cap \EE_2^c}\leq \sum_{t=T_0}^{T-1} \sum_{k=1}^{m} 3 n^{-(a-b-2)} + n^{-4} \leq  \frac{1}{n^2}\enspace.
\end{align}

By \eq{tv_c}, \eq{event1_c}, \eq{event2_c} and \eq{event3_c}, we have for sufficiently large $n$ and $T=515\log n$
\begin{align*}
	\sup_{X_0\in \SC^{n}} \norm{\mathcal{L}\br{X_{T}}-\mu_{\C}}_{\mathrm{TV}} \leq {\frac{1}{2n}}\enspace.
\end{align*}

As for $T\geq 515\log n$,
by \lem{fcts_tv}
we have
\begin{align*}
	\sup_{X_0\in \SC^{n}} \norm{\mathcal{L}\br{X_{T}}-\mu_{\C}}_{\mathrm{TV}}\leq 2\br{{\frac{1}{n}}}^{\lint{\frac{T}{515\log n}}} \leq \frac{1}{2^{\br{c/515-1}\log n-1}}\enspace.	
\end{align*}

\subsection{Proof of \lem{realNNfullmap}}\label{sec:proofofrealnnfullmap}
\begin{proof}[Proof of \lem{realNNfullmap}]
To see why $\widetilde{\NN}= \SR^{2^n}$, we prove $\SR^{2^n}\subseteq \widetilde{\NN}$ since $\widetilde{\NN} \subseteq \SR^{2^n}$ is trivial.
	Given $\ket{\xi}\in\SR^{2^n}$, we prove
	$\ket{\xi}\in \widetilde{\NN} $
	by constructing a series of
	$(\sigma_t,\cf_t)$ for $t\leq n$ such that
	$$\KK_{\sigma_n,\cf_n}\cdots\KK_{\sigma_1,\cf_1}\ket{\eta}=\ket{\xi}$$
	(we let $\KK_{\sigma_t,\cf_t} = \id$ for $t>n$).
	
	The idea is to bisect $\KK_{\sigma_t,\cf_t}\cdots\KK_{\sigma_1,\cf_1}\ket{\eta}$ and $\ket{\xi}$ accordingly and recursively, always keeping the 2-norm of each part of the two vectors equal.
	
 To illustrate the process in detail, we begin with indexing the entries of a vector in $\SR^{2^n}$ by a bit string of length $n$. In step $t\in[n]$, we divide the index set $\bit{n}$ into $2^t$ sets based on the first $t$ bits. For $y\in\bit{t}$, we define $S_y$ to be the set of all elements in $\bit{n}$ with prefix $y$. We then split $\ket{\eta}$ and $\ket{\xi}$ into $2^t$ sub-vectors, according to $S_y$ where $y\in\bit{t}$. Let $L_y$ be the function from $\SR^{2^n}$ to $\R$ that gives the length of the  sub-vector corresponding to $S_y$. Let
	$$
\ket{\eta_0}=\ket{\eta},\quad\ket{\eta_t}=\KK_{\sigma_t,\cf_t}\ket{\eta_{t-1}}\quad\text{for}\enspace t\in[n]\enspace.
$$
Our goal is to construct $(\sigma_t,\cf_t)$ at each step $t\in[n]$, such that $L_y\!\br{\ket{\eta_t}}=
L_y\!\br{\ket{\xi}}$ for every $t\in[n-1],y\in\bit{t}$ and that $\ket{\eta_n}=\ket{\xi}$.

To accomplish this, we first define $\sigma_t \in S_{2^n}$ to be
$$\sigma_t(x)= {x_tx_1\dots x_{t-1}x_{t+1}\dots x_n}\quad \text{for all }x\in\bit{n}\enspace.$$
For any $y\in\bit{t-1}$,
$\sigma_t$ matches every index in $S_{ {y0}}$
with another index in $S_{ {y1}}$
that shares a common suffix of length $n-t$.

Next, we move to construct $\cf_t$.
Define $\alpha_y$ for each $y\in\bit{t-1}$ as
$$\alpha_y=\begin{cases}\arccos\frac{L_{ {y0}}\br{\ket{\xi}}}{L_{y}\br{\ket{\xi}}}&\text{if }L_{y}\br{\ket{\xi}}\ne0\text{ and }t<n\enspace,\\
0&\text{if }L_{y}\br{\ket{\xi}}=0\text{ and }t<n\enspace.
\end{cases}$$
When $t=n$, we define $\alpha_y$ for $y\in\bit{n-1}$ to be any angle satisfying
\begin{align*}
\br{\ket{\xi}}_{ {y0}}&=L_y\br{\ket{\xi}}\cos\alpha_y \enspace,\\
\br{\ket{\xi}}_{ {y1}}&=L_y\br{\ket{\xi}}\sin\alpha_y \enspace.
\end{align*}
We want to design $\cf_{t}$ which controls the rotation of each index pair
to let each pair of indices between $S_{ {y0}}$ and $S_{ {y1}}$ (induced by $\sigma_t$) form an angle $\alpha_y$
with the $x$-axis in a two-dimensional Cartesian coordinate system.
To this end, for each $y\in\bit{t-1}$ and $z\in\bit{n-t}$,
we define $\beta_y(z)$ to be any angle satisfying
\begin{align*}
\br{\ket{\eta_{t-1}}}_{ {y0z}}&=\sqrt{\br{\ket{\eta_{t-1}}}_{ {y0z}}^2+\br{\ket{\eta_{t-1}}}_{ {y1z}}^2}\cos\beta_y(z) ,\\
\br{\ket{\eta_{t-1}}}_{ {y1z}}&=\sqrt{\br{\ket{\eta_{t-1}}}_{ {y0z}}^2+\br{\ket{\eta_{t-1}}}_{ {y1z}}^2}\sin\beta_y(z) ,
\end{align*}
and we define $\cf_t$ to be
$$
\cf_t( {yz})=\br{\alpha_y-\beta_y(z)}/(2\pi)
$$
for all $y\in\bit{t-1}$ and $z\in\bit{n-t}$. It can be easily verified that
$$L_y\br{\ket{\eta_t}}=L_y\br{\ket{\xi}}\enspace\text{ for }\enspace t\in[n-1] \enspace\text{and} \enspace y\in\bit{t} ,$$
and $\ket{\eta_n}=\ket{\xi}$.
\end{proof}

\subsection{Proof of \lem{realNNepsnet}}
\begin{proof}[Proof of \lem{realNNepsnet}]
 	Let $\ket{u} = \csrsscons_{\rep{\sigma}{T},\rep{\cf}{T}}\!\!\ket{\eta}\in\widetilde{\NN}$ for some $\rep{\sigma}{T} \subseteq S_{2^n}$ and $\rep{\cf}{T}\subseteq\{f: \bit{n-1}\to [0,1)\}$.
	For every $t\in[T]$, we define $f_t$ by letting $f_t(y)$ be the $d$ digits after the binary point in $\widetilde{f}_t(y)$ for all $y\in\bit{n-1}$.
	It is evident that $\ket{v} = \srsscons_{\rep{\sigma}{T},\rep{f}{T}}\!\!\ket{\eta}\in\NN$.
	And
	\begin{align}
		\norm{\ket{u} - \ket{v}}_2 ~
	 =	~& \norm{
			\csrsscons_{\rep{\sigma}{T},\rep{\cf}{T}}\!\!\ket{\eta} -
			\srsscons_{\rep{\sigma}{T},\rep{f}{T}}\!\!\ket{\eta}
		}_2 \nonumber\\
	\leq  ~& \norm{
			\csrsscons_{\rep{\sigma}{T},\rep{\cf}{T}} -
			\srsscons_{\rep{\sigma}{T},\rep{f}{T}}
		}_\infty \leq  2^{1-d}\pi T = \frac{1030\pi (\secpar+1)n}{\secpar^{\log \secpar} n^{\log n}} \enspace, \label{eq:uvclose}
	\end{align}
	where the last inequality is from \fct{uprod} and \lem{k_hk_close}.
	This proves that $\NN$ is indeed an $\epsilon$-net
	for $\widetilde{\NN}$ where $\epsilon  = \frac{1030\pi (\secpar+1)n}{\secpar^{\log \secpar} n^{\log n}} =  \negl{\secpar}$. Combining with \lem{realNNfullmap}, we conclude the result.
\end{proof}

\subsection{Proof of \thm{kactosrssc}}
To prove \thm{kactosrssc}, recall the ensemble of (infinitely many) unitary operators
$\enscsrssconsc \coloneq \st{\csrssconsc }_{\secpar}$ defined in \sect{proof_rssc}.
	
	\begin{proposition}
	For $T = 515 (\secpar + 1) n $,
	the ensemble of unitary operator $\enscsrssconsc$
	is a \csrss.
	\end{proposition}

	\begin{proof}
	Note that a uniformly random
	$\csrssconsc_{\rep{\sigma}{T},\rep{\cf}{T},\rep{\cg}{T},\rep{\ch}{T}}$
	corresponds to a $T$-step parallel Kac's walk on $\SC^{2^n}$.
	The proposition then follows from \thm{mixing_time_c} and the definition of \csrss.
	\end{proof}
	
	Let $\eta\in\S(H)$ be an arbitrary state. Denote
	$$
	\NN = \st{
	\srssconsc_{\rep{\sigma}{T},\rep{f}{T},\rep{g}{T},\rep{h}{T}}
		\!\!\ket{\eta} }
	\enspace \text{and}\enspace
	\widetilde{\NN} = \st{\csrssconsc_{\rep{\sigma}{T},\rep{\cf}{T},\rep{\cg}{T},\rep{\ch}{T}}\!\!\ket{\eta}}\enspace.
	$$
	We prove the following two lemmas.
	
	\begin{lemma}\label{lem:complexNNfullmap}
		$\widetilde{\NN}= \S(\H)$ .
	\end{lemma}
	\begin{proof}
	We prove $\SC^{2^n}\subseteq \widetilde{\NN}$ since $\widetilde{\NN} \subseteq \SC^{2^n}$ is trivial.
	Given $\ket{\xi}\in\SC^{2^n}$, we prove
	$\ket{\xi}\in \widetilde{\NN} $
	by constructing a series of
	$(\sigma_t,\cf_t,\cg_t,\ch_t)$ for $t\leq n+1$ such that
	$$\LL_{\sigma_n,\cf_{n+1},\cg_{n+1},\ch_{n+1}}\cdots\LL_{\sigma_1,\cf_1,\cg_1,\ch_1}\ket{\eta}=\ket{\xi}$$
	(we let $\LL_{\sigma_t,\cf_t,\cg_t,\ch_t} = \id$ for $t>n+1$).
	
	The proof idea is similar to \app{proofofrealnnfullmap}. For any $t\in[n]$ and $y\in\bit{t}$, we use the definition of $S_y$, $L_y$ (change the domain to $\SC^{2^n}$), and $\ket{\eta_t}$ (change $\KK_{\sigma_t,\cf_t}$ to $\LL_{\sigma_t,\cf_t,\cg_t,\ch_t}$) there. Our goal is to construct $(\sigma_t,\cf_t,\cg_t,\ch_t)$ at each step $t\in[n+1]$, such that
$L_y\!\br{\ket{\eta_t}}=
L_y\!\br{\ket{\xi}}$ for every $t\in[n-1],y\in\bit{t}$ and that $\ket{\eta_{n+1}}=\ket{\xi}$.
That is, after $n-1$ steps, for every $y\in\bit{n-1}$, the two-dimensional sub-vectors of $\ket{\eta_n}$ and $\ket{\xi}$ induced by $S_y$ have the same length. In the final two steps, we adjust the two sub-vectors to be equal.

For any $t\in[n]$, let $\sigma_t$ be defined as in \app{proofofrealnnfullmap}. We now construct $\cf_t,\cg_t,\ch_t$.

	For any $t\in[n],y\in\bit{t-1},z\in\bit{n-t}$, suppose that
	$$\br{\ket{\eta_{t-1}}}_{ {y0z}}=e^{\i\theta^\eta_{y,0,z}}r^\eta_{y,0,z}\quad\text{and}\quad\br{\ket{\eta_{t-1}}}_{ {y1z}}=e^{\i\theta^\eta_{y,1,z}}r^\eta_{y,1,z}\enspace,$$
	$$\br{\ket{\xi}}_{ {y0z}}=e^{\i\theta^\xi_{y,0,z}}r^\xi_{y,0,z}\quad\text{and}\quad\br{\ket{\xi}}_{ {y1z}}=e^{\i\theta^\xi_{y,1,z}}r^\xi_{y,1,z}\enspace,$$
	where $\theta^\eta_{y,0,z},\theta^\eta_{y,1,z},\theta^\xi_{y,0,z},\theta^\xi_{y,1,z}\in[0,2\pi),r^\eta_{y,0,z},r^\eta_{y,1,z},r^\xi_{y,0,z},r^\xi_{y,1,z}\in[0,1].$
	
	For $t\in[n-1]$, define
	$$\alpha_y=\begin{cases}\arccos\frac{L_{ {y0}}\br{\ket{\xi}}}{L_{y}\br{\ket{\xi}}}&\text{if }L_{y}\br{\ket{\xi}}\ne0\enspace,\\
0&\text{otherwise.}
\end{cases}$$
and
	$$\rho_y(z)=\begin{cases}\arccos\frac{r^\eta_{y,0,z}}{\sqrt{\br{r^\eta_{y,0,z}}^2+\br{r^\eta_{y,1,z}}^2}}&\text{if }\br{r^\eta_{y,0,z}}^2+\br{r^\eta_{y,1,z}}^2>0\enspace,\\
0&\text{otherwise.}
\end{cases}$$

	We then define $\cf_t,\cg_t,\ch_t$ to be
	\begin{align*}
\cf_t( {yz})&=\sin^2\br{\alpha_y-\rho_y(z)},\\
\cg_t( {yz})&=\theta^\eta_{y,1,z}/(2\pi),\\
\ch_t( {yz})&=\theta^\eta_{y,0,z}/(2\pi),
\end{align*}
for all $y\in\bit{t-1}$ and $z\in\bit{n-t}$. It can be easily verified that
$$L_y\br{\ket{\eta_t}}=L_y\br{\ket{\xi}}\enspace\text{ for }\enspace t\in[n-1] \enspace\text{and} \enspace y\in\bit{t} .$$
For $t=n$, we set the second entry of the sub-vector of $\ket{\eta_{n-1}}$ induced by $S_y$ to zero for all $y\in\bit{n-1}$. That is, define
\begin{align*}
\cf_n(y)&=\begin{cases}
\frac{\br{r^\eta_{y,1}}^2}{\br{r^\eta_{y,0}}^2+\br{r^\eta_{y,1}}^2}&\text{if }\br{r^\eta_{y,0}}^2+\br{r^\eta_{y,1}}^2>0,\\
0&\text{otherwise,}
\end{cases}\\
\cg_n(y)&=-\theta^\eta_{y,0}/(2\pi),\\
\ch_n(y)&=\br{\pi-\theta^\eta_{y,1}}/(2\pi),
\end{align*}
for all $y\in\bit{n-1}$. We then have $$L_y\br{\ket{\eta_n}}=L_y\br{\ket{\eta_{n-1}}}=L_y\br{\ket{\xi}}\text{ for all }y\in\bit{n-1}.$$ For the final step, we let $\sigma_{n+1}=\sigma_n$, and define
\begin{align*}
\cf_{n+1}(y)&=\begin{cases}
\frac{\br{r^\xi_{y,1}}^2}{\br{r^\xi_{y,0}}^2+\br{r^\xi_{y,1}}^2}&\text{if }\br{r^\xi_{y,0}}^2+\br{r^\xi_{y,1}}^2>0,\\
0&\text{otherwise,}
\end{cases}\\
\cg_{n+1}(y)&=\theta^\xi_{y,0}/(2\pi),\\
\ch_{n+1}(y)&=-\theta^\xi_{y,1}/(2\pi),
\end{align*}
for all $y\in\bit{n-1}$. It can be easily verified that $\ket{\eta_{n+1}}=\ket{\xi}$.
\end{proof}

	\begin{lemma}\label{lem:complexNNepsnet}
		There exists an $\epsilon  = \negl{\secpar}$ such that
		$\NN$
		is an $\epsilon$-net for
		$\S(\H)$ .
	\end{lemma}
	
	\begin{proof}
	By \lem{complexNNfullmap}, it suffices to prove that
	there exists an $\epsilon  = \negl{\secpar}$ such that
	$\NN$ is an $\epsilon$-net for $\widetilde{\NN}$.
	
	Let $\ket{u} = \csrssconsc_{\rep{\sigma}{T},\rep{\cf}{T},\rep{\cg}{T},\rep{\ch}{T}}\!\!\ket{\eta}\in\widetilde{\NN}$ for some $\rep{\sigma}{T} \subseteq S_{2^n}$ and $\rep{\cf}{T}$,
		$\rep{\cg}{T}$,
		$\rep{\ch}{T}\subseteq\{f: \bit{n-1}\to [0,1)\}$.
	For every $t\in[T]$, we define $f_t$ by letting $f_t(y)$ be the $d$ digits after the binary point in $\widetilde{f}_t(y)$ for all $y\in\bit{n-1}$.
	We define $g_t$ and $h_t$ for $t\in[T]$ in the same way.
	It is evident that $\ket{v} = \srssconsc_{\rep{\sigma}{T},\rep{f}{T},\rep{g}{T},\rep{h}{T}}
		\!\!\ket{\eta}\in\NN$.
	And
	\begin{align}
		&\norm{\ket{u} - \ket{v}}_2 \nonumber\\
	 =	~& \norm{
			\srssconsc_{\rep{\sigma}{T},\rep{f}{T},\rep{g}{T},\rep{h}{T}}
		\!\!\ket{\eta}-
			\csrssconsc_{\rep{\sigma}{T},\rep{\cf}{T},\rep{\cg}{T},\rep{\ch}{T}}\!\!\ket{\eta}
		}_2 \nonumber\\
	\leq  ~& \norm{
			\srssconsc_{\rep{\sigma}{T},\rep{f}{T},\rep{g}{T},\rep{h}{T}}
			- \csrssconsc_{\rep{\sigma}{T},\rep{\cf}{T},\rep{\cg}{T},\rep{\ch}{T}}
		}_\infty \nonumber\\
	\leq ~& 2^{6-\frac{d}{2}} T = \frac{32960(\secpar+1)n}{\secpar^{\log \secpar} n^{\log n}}
	\enspace, \label{eq:uvclosec}
	\end{align}
	where the last inequality is from \fct{uprod} and \lem{ltildelclose}.
	This proves that $\NN$ is indeed an $\epsilon$-net
	for $\widetilde{\NN}$ where $\epsilon  = \frac{32960(\secpar+1)n}{\secpar^{\log \secpar} n^{\log n}} = \negl{\secpar}$.
\end{proof}

	\begin{proof}[Proof of \thm{kactosrssc}]
	It is easy to see that the uniformity condition is satisfied. 
	Let $\kappa$ denote the key length.
	Quantum circuit $\srssconsc$ applies $\srssconsc_{\rep{\sigma}{T},\rep{f}{T},\rep{g}{T},\rep{h}{T}}$ after reading $\rep{\sigma}{T},\rep{f}{T},\rep{g}{T}$ and $\rep{h}{T}$.
	To implement $\srssconsc_{\rep{\sigma}{T},\rep{f}{T},\rep{g}{T},\rep{h}{T}}$, we need to realize each of
	the $T = 515 (\secpar + 1) n$ unitary gates $L$.
	Since each gate $L$ can be implemented in $\poly{n,\secpar,\klen}$ time,
	the total construction time for $\srssconsc_{\rep{\sigma}{T},\rep{f}{T},\rep{g}{T},\rep{h}{T}}$
	is also $\poly{n,\secpar,\klen}$.

	In conjunction with \lem{complexNNepsnet}, it suffices to prove the existence of a good distribution $\widetilde{\nu}$ meeting the requirement in \defi{trssb}. Fix $\ket{\eta}\in\S(\H)$. 
	Define three distributions:
	\begin{itemize}
		\item $\nu$ be the distribution of
			$\srssconsc_{\rep{\sigma}{T},\rep{f}{T},\rep{g}{T},\rep{h}{T}}\!\!\ket{\eta}$
			with independent and uniformly random permutations $\rep{\sigma}{T}\subseteq S_{2^n}$,
			and random functions $\rep{f}{T},\rep{g}{T}$, $\rep{h}{T}\subseteq\{f: \bit{n-1}\to\bit{d}\}$.
		\item $\widetilde{\nu}$ be the distribution of
			$\csrssconsc_{\rep{\sigma}{T},\rep{\cf}{T},\rep{\cg}{T},\rep{\ch}{T}}\!\!\ket{\eta}$
			with independent and uniformly random permutations
			$\rep{\sigma}{T}\subseteq S_{2^n}$,
			and random functions
			$\rep{\cf}{T},\rep{\cg}{T}$, $\rep{\ch}{T}\subseteq\{f:\bit{n-1}\to[0,1)\}$.
		\item $\mu$ be the Haar measure on $\S(\H)$.
	\end{itemize}
	Note that $\widetilde{\nu}$ is the output distribution of $T$-step parallel Kac's walk in $\SC^{2^n}$. Thus by \thm{mixing_time_c}, we have
	\begin{align}\label{eq:closetvc}
		\tv{\widetilde{\nu} - \mu} \leq \frac{1}{2^{\secpar n - 1}} = \negl{\secpar}\enspace.
	\end{align}
	We are left to show the Wasserstein $\infty$-distance between $\nu$ and $\widetilde{\nu}$ is negligible.
	To this end, we construct a coupling $\gamma_0$ of $\nu$ and $\widetilde{\nu}$ by using the same permutation $\sigma_t$ and letting $f_t$ be the function satisfying $f_t(y)$ is the $d$ digits after the binary point in $\widetilde{f}_t(y)$ for all $y\in\bit{n-1}$ (The same applies to $g_t$ and $h_t$).
	Therefore
	{\small\begin{align}
		W_{\infty} ( \nu, \widetilde{\nu}) =
		~& \lim_{p\to\infty}\br{ \inf_{\gamma\in\Gamma(\nu,\widetilde{\nu})} \E_{(\ket{v},\ket{u})\sim \gamma}\!\Br{\norm{\ket{v}-\ket{u}}_2^p} }^{1/p} \nonumber\\
		\leq ~& \lim_{p\to\infty}\br{ \E_{(\ket{v},\ket{u})\sim \gamma_0}\!\Br{\norm{\ket{v}-\ket{u}}_2^p} }^{1/p}
		~\overset{\text{(Eq.~\eq{uvclosec})}}{\leq} ~
		\frac{32960(\secpar+1)n}{\secpar^{\log \secpar} n^{\log n}}= \negl{\secpar} \label{eq:closewinftyc}\enspace .
	\end{align}}
	Thus we conclude the result.
\end{proof}

\end{document}